\documentclass[10pt,reqno]{amsart}
\usepackage{amscd,amssymb,amsmath,latexsym,enumerate}
\usepackage[mathscr]{euscript}
\usepackage{epsfig}
\usepackage{fancybox}
\usepackage{manfnt}
\usepackage{verbatim}
\usepackage{tikz}
\usetikzlibrary{calc}
\usepackage{tikz-cd}
\usepackage{citeref}
\usepackage{relsize}
\usepackage{exscale}

\usepackage{color}
\newcommand{\red}{\textcolor[rgb]{0.99,0.00,0.00}}

\usepackage{xcolor}
\definecolor{MyBlue}{cmyk}{1,0.13,0,0.63}
\definecolor{MyGreen}{cmyk}{0.91,0,0.88,0.52}
\newcommand{\mylinkcolor}{MyBlue}
\newcommand{\mycitecolor}{MyGreen}
\newcommand{\myurlcolor}{black}

\usepackage{hyperref}
\hypersetup{%
	bookmarksnumbered=true,bookmarksopen=false,%
	plainpages=false,
	linktocpage=true,%
	colorlinks=true,breaklinks=true,%
	linkcolor=\mylinkcolor,citecolor=\mycitecolor,urlcolor=\myurlcolor,%
	pdfpagelayout=OneColumn,%
	pageanchor=true,%
}

\RequirePackage{geometry}
\title[Framework for higher-order bulk-boundary correspondences]{$C^\ast$-framework for higher-order bulk-boundary correspondences}

\author[D. Ojito]{Danilo Polo Ojito}
\address{Department of Physics and Department of Mathematical Sciences, Yeshiva University 
	\\New York, NY 10016, USA \\
	\href{mailto:danilo.poloojito@yu.edu}{danilo.poloojito@yu.edu}}

\author[E. Prodan]{Emil Prodan}

\address{Department of Physics and Department of Mathematical Sciences 
	\\Yeshiva University 
	\\New York, NY 10016, USA \\
	\href{mailto:prodan@yu.edu}{prodan@yu.edu}}

\author[T. Stoiber]{Tom Stoiber}

\address{Department of Physics and Department of Mathematical Sciences, Yeshiva University 
	\\New York, NY 10016, USA \\
	\href{mailto:tom.stoiber@yu.edu}{tom.stoiber@yu.edu}}

\date{\today}

\newtheorem{theorem}{Theorem}[section]
\newtheorem{definition}[theorem]{Definition}
\newtheorem{proposition}[theorem]{Proposition}

\newtheorem{corollary}[theorem]{Corollary}
\newtheorem{remark}[theorem]{Remark}
\newtheorem{example}[theorem]{Example}

\newcommand{\CM}{{\mathbb C}}
\newcommand{\NM}{{\mathbb N}}
\newcommand{\RM}{{\mathbb R}}

\newcommand{\TM}{{\mathbb T}}
\newcommand{\ZM}{{\mathbb Z}}

\newcommand{\KM}{{\mathbb K}}

\newcommand{\Aa}{{\mathcal A}}
\newcommand{\Ee}{{\mathcal E}}
\newcommand{\Pp}{{\mathcal P}}

\newcommand{\KK}{{\bf K}}
\newcommand{\Bb}{{\mathcal B}}

\newcommand{\Ff}{{\mathcal F}}
\newcommand{\Gg}{{\mathcal G}}
\newcommand{\Ww}{{\mathcal W}}
\newcommand{\Uu}{{\mathcal U}}
\newcommand{\Vv}{{\mathcal V}}
\newcommand{\Ss}{{\mathcal S}}
\newcommand{\Oo}{{\mathcal O}}
\newcommand{\Tt}{{\mathcal T}}

\newcommand{\Cc}{{\mathcal C}}
\newcommand{\Jj}{{\mathcal J}}
\newcommand{\Ii}{{\mathcal I}}
\newcommand{\Ll}{{\mathcal L}}
\newcommand{\Qq}{{\mathcal Q}}
\newcommand{\Kk}{{\mathcal K}}
\newcommand{\Hh}{{\mathcal H}}

\newcommand{\one}{{\bf 1}}

\newcommand{\Ch}{{\rm Ch}} 
 
\newcommand{\Ker}{{\rm Ker}} 
 
\newcommand{\sgn}{{\rm sgn}}

\newcommand{\idmap}{\textup{id}}

\providecommand{\abs}[1]{\left \lvert#1 \right \rvert}

\begin{document}

\begin{abstract} A typical crystal is a finite piece of a material which may be invariant under some point symmetry group. If it is a so-called {\it intrinsic} higher-order topological insulator or superconductor, then it displays boundary modes at hinges or corners protected by the crystalline symmetry and the bulk topology. We explain the mechanism behind such phenomena using operator K-theory. Specifically, we derive a groupoid $C^\ast$-algebra that 1) encodes the dynamics of the electrons in the infinite size limit of a crystal; 2) remembers the boundary conditions at the crystal's boundaries, and 3) admits a natural action by the point symmetries of the atomic lattice. The filtrations of the groupoid's unit space by closed subsets that are invariant under the groupoid and point group actions supply equivariant cofiltrations of the groupoid $C^\ast$-algebra. We show that specific derivations of the induced spectral sequences in twisted equivariant K-theories enumerate all non-trivial higher-order bulk-boundary correspondences.
\end{abstract}


\maketitle


{\scriptsize \tableofcontents}

\section{Introduction and Main Statements}\label{Sec:Intro}

Bulk-boundary correspondence is one of the hallmark features of topological insulators and superconductors \cite{RyuNJP2010}[Sec.~1.2]. In very general terms, such correspondence supplies a prediction about the dynamics of the electrons close to a flat boundary of a sample, based {\it solely} on input coming from bulk properties of the material. In more precise terms, a topological material develops propagating wave-channels along flat boundaries, which are  active at energies or frequencies where such channels are entirely inexistent in the bulk of the material. In terms of the Hamiltonians generating the dynamics of the electrons, this can be phrased by saying that the Hamiltonian is spectrally gapped in a pristine infinite sample, but this gap fills with spectrum when a flat boundary is cut into the sample. The bulk-boundary correspondences have been the subject of intense research and, from the mathematical point of view, the subject is in a good shape for the one-particle sector \cite{KellendonkRMP2002,ProdanSpringer2016,BourneAHP2017,
BourneAHP2020,AlldridgeCMP2020,ProdanJPA2021}. 

Further innovation in the field came from the works \cite{BenalcazarScience2017,SchindlerSciAdv2018}, where it was observed that several flat boundaries meeting along one hinge or at a corner can induce non-trivial  electron dynamics that can be predicted {\it entirely} from the bulk properties of the material. These works also laid down the general principles behind this new phenomena, which were dubbed higher-order bulk-boundary correspondences. We will try to explain these principles and their challenges, when it comes to a rigorous mathematical formulation, using the diagrams from Fig.~\ref{Fig:Intro}. There, we show a regular 2-dimensional lattice that has been cut to a finite sample with several flat boundaries. For some materials, which are insulators, i.e. the Hamiltonian has a gapped bulk energy spectrum,\footnote{Throughout, a spectral gap will mean an open interval which is not contained in the spectrum of the Hamiltonian and which contains the Fermi energy, fixed w.l.o.g. to be zero.}  one can witness mid-gap electron states localized at the exposed corners. These corner states are in general unstable, unless they exist in a spectral region made up exclusively out of corner-supported states. In contrast, if the entire boundary hosts wave channels, then the mentioned corner states are embedded eigenvalues inside the continuous spectrum and may dissolve under perturbations. If in contrast the edges are insulating, we say that the edges are gapped. In that case there are still other factors of indeterminacy. Indeed, one could modify the termination of the lattice by depositing on the boundary  quasi 1-dimensional topological insulators that host topological end modes, which are completely indistinguishable from the corner modes. If one or more such boundary layers are deposited along the boundaries of the sample, as schematically shown by the colored layers in Fig.~\ref{Fig:Intro}(b), then the multiplicities of the corner states will obviously be altered. Moreover, by coupling these additional boundary layers with the rest of the material, one may be able to remove some or all of the corner states. It is therefore clear that the corner states are, in general, very sensitive to the physical conditions close to the boundaries. 

\begin{figure}[t]
\center
\includegraphics[width=0.8\textwidth]{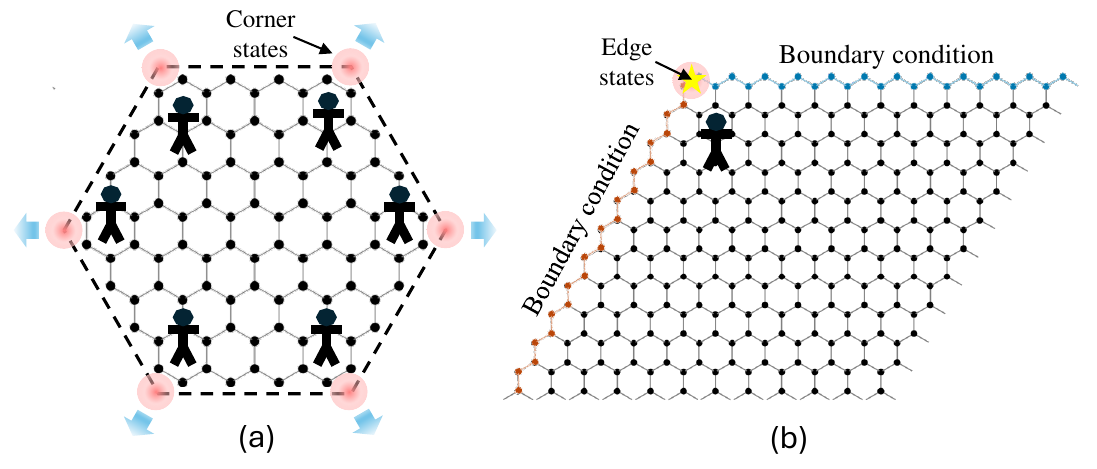}\\
  \caption{\small (a) A crystal with a honeycomb atomic arrangement, displaying edges and corners. One can take infinite-volume limits by letting the crystal is grow in different directions while fixing the positions of the observers at the corners. (b) The crystal seen by one of the observers in the infinite size limit. 
}
 \label{Fig:Intro}
\end{figure}

If the edges are gapped and corner states are present, then the latter are insensitive to perturbations as long as the edges remain gapped \cite{HayashiCMP2018,HayashiLMP2019} (see also \cite{Polo2022} for magnetic interfaces). Those type of corner states are generally so-called extrinsic higher-order boundary states because they require protection by both the bulk and the edge spectral gaps \cite{GeierPRB2018}. Calling them extrinsic is justified because the number and the characteristics of the corner-localized modes cannot be predicted from the bulk properties of the material. For example, in this approach, the bulk material can be topologically trivial by all standards, yet a corner geometry can display corner states depending on the details of how the boundary modifies the bulk dynamics. In the absence of crystalline symmetries, Hayashi proposes a classification of extrinsic higher-order correspondences in \cite{HayashiLMP2021}.

The higher-order bulk-boundary correspondence proposed in \cite{BenalcazarScience2017,SchindlerSciAdv2018, Trifunovic-Brouder}, and widely adopted by the physics community, is different. As the name suggests, the existence and qualitative properties of the corner modes must be determined by the bulk, hence be insensitive to the boundary conditions. It turns out that this is generally only possible if the sample has crystalline symmetries. \cite{SchindlerSciAdv2018}  Under this setting, the boundary conditions at different parts of the flat boundaries are related by a combination of space and possibly fundamental symmetries\footnote{Fundamental symmetries refer to the time-reversal, particle-hole and chiral symmetries.}. In special cases, this constraint is enough to make it impossible to remove all of the corner modes by a change of symmetry-preserving boundary condition. In this case, one speaks of an {\it intrinsic} higher-order topological insulator, since topological invariants of the bulk Hamiltonian in conjunction with the symmetry allow to predict the existence of corner modes whenever the edges are gapped. 

The same principle can be generalized to crystals cut out of 3-dimensional materials, which can display either hinge or corner topological states. These cases can be distinguished by introducing an order for the correspondences. One speaks of $n$-th order bulk-boundary correspondence if $n$ is the difference between the dimensions of the bulk and the boundary. 
Specifically, the bulk topological invariants of a gapped bulk Hamiltonian in three dimensions can result in protected boundary states at faces (order $1$), hinges (order 2) or corners (order 3) and similarly for different geometries. Since synthetic dimensions are possible in material science (e.g. \cite{ProdanPRB2015}),  there is no actual limit on the ``effective'' dimension of the bulk material and the principles of higher-order bulk-boundary correspondences work as well for such settings.

While the above principles are now well understood, researched and explored by the physics community \cite{XieNatRevPhys2021,ZhangNature2023}, there still remains a need for a mathematical framework to thoroughly explain and formalize these concepts at the same level of rigor as the ordinary bulk-boundary correspondence. What is also missing is a rigorous device that enumerates all possible non-trivial higher-order bulk-boundary correspondences for a specified geometry and symmetry group. For ordinary bulk-boundary correspondence this can be done using $C^*$-algebras and operator $K$-theory, however, there are several good reasons why higher-order bulk-boundary correspondences so far resisted a similar treatment. Firstly, we note that, rigorously speaking, such phenomena can only take place in the infinite-size limit of a sample.\footnote{Topological phases are not separated by sharp phase boundaries in large but finite samples, but the footprints of the topological dynamics are still observed to excellent approximation.} Thus, we are presented with the new challenge of building a $C^\ast$-algebra of observations which, although describing the infinite-size limit, still continues to encode precise information about all boundaries and where symmetries can be implemented as automorphisms. While this may sound paradoxical at first, this can be indeed accomplished if we think of the algebra as encoding the joint observations of a team of several experimenters. For example, the experimenters shown in Fig.~\ref{Fig:Intro}(a) observe the electron dynamics as the sample grows indefinitely, always having a corner in their field of view or reach. In the infinite-size limit, a single observer, {\it e.g.} the one depicted in Fig.~\ref{Fig:Intro}(b), will always see a single corner and a pattern extending infinitely outwards. As such, the symmetry of the original sample is lost to this experimenter.\footnote{This is most obvious if the symmetries of the crystal include rotations or space inversion, which permute all corners with each other.} However, it is recovered when one compares the measurements of several experimenters at symmetry-related corners. In the mathematical formulation, an experimenter will correspond to an irreducible representation of the algebra to be constructed and the need for a team of experimenters expresses that we need to employ distinct irreducible representations on {\it a priori} unrelated Hilbert spaces to get the full picture of the electron dynamics. Another strong reason for why one needs a whole team of observers is that the corner modes may not appear on all corners. The higher-order bulk-boundary correspondence is a global statement about the topological modes carried collectively by all corners. These modes can in some symmetry classes be redistributed among different corners by a change of boundary condition and thus will not be protected when one considers only a single corner.  Secondly, even if the mentioned $C^\ast$-algebra can be successfully constructed, one may find that it has a rich lattice of ideals and that there is no direct connection between the ideal corresponding to the physical observations around the corners and the algebra of bulk observations (as we recall in Section~\ref{Sec:Groupoid}, the ordinary bulk-edge correspondence is based on short exact sequences linking ideals of boundary observables with the bulk algebra). Thirdly, one now has to navigate a hierarchy of boundaries and has to determine which bulk models remain gappable up to boundaries of which order. Conversely, the models which are not gappable at a boundary must exhibit topologically protected boundary states and one needs to enumerate them together with the possible manifestations of their boundary states which may depend on the boundary condition. \footnote{A bulk model is called gappable at some boundary if it is possible to realize it by a gapped Hamiltonian on a geometry which has that specified boundary.} Lastly, once the correct constructions for enumerating the non-trivial higher-order bulk-boundary correspondences are identified, explicit computations will, in many cases, lead to a difficult exercise which involves twisted equivariant $K$-theory. 

\begin{figure}[t]
\center
\includegraphics[width=\textwidth]{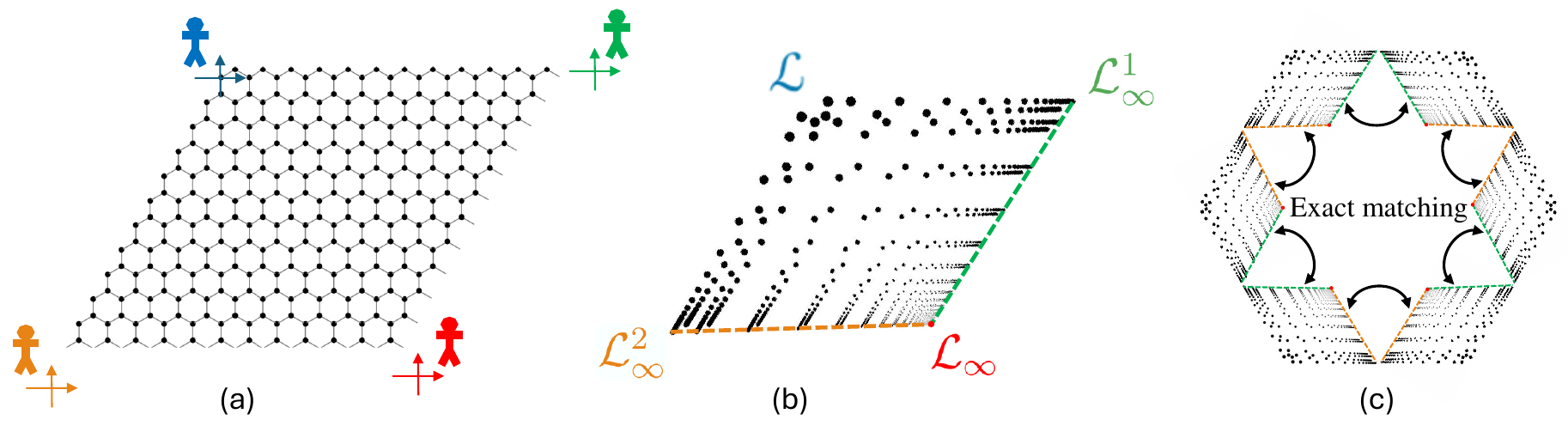}\\
  \caption{\small The transversal of the infinite hexagonal pattern: (a) The hexagonal pattern with different asymptotic observer positions indicated. (b) A picture of the transversal $\Xi_\Ll$ in the space $\Cc(\RM^2)$ of patterns. The discrete points are the orbit of $\Ll$ under discrete translations with distances symbolizing convergence to limit points in the Fell topology. (c) Any neighboring observers from Fig.~\ref{Fig:Intro}(a) will asymptotically see identical patterns at shared edges, therefore the boundaries of the transversals match perfectly along the indicated boundaries.
}
 \label{Fig:T0}
\end{figure}

After presenting the problem and its challenges, we now describe our solutions and how they fit into the existing mathematical landscape. The search for model $C^\ast$-algebras related to locally compact spaces with boundaries can be traced back to the works of Douglas {\it et al} \cite{CoburnPNAS1969,Coburn1IHESPM1971,Coburn2IHESPM1971} for the half-plane, and \cite{DouglasTAMS1071,CoburnJDG1972} for the quarter-plane. A more general and modern machinery for building $C^\ast$-algebras over cones of the discrete plane was supplied by Park in \cite{ParkJOT1990}.\footnote{The work \cite{HayashiCMP2018} by Hayashi on extrinsic higher-order correspondences builds on them.} At the core of these constructions sit the Toeplitz extensions and their generalizations via pullback constructions. These extensions have been dressed up with physical meaning in a remarkable paper by Kellendonk, Richter and Schulz-Baldes \cite{KellendonkRMP2002}, where a rigorous explanation of the bulk-boundary correspondence observed in quantum Hall experiments was communicated for the first time. Much of the subsequent mathematical works on bulk-boundary correspondences use this mentioned work as a template.

On another front, Bellissard and Kellendonk developed a groupoid formalism \cite{Bellissard1986,KellendonkRMP95}, which delivers model $C^\ast$-algebras for generic atomic configurations. While their work focused mainly on Delone sets that are associated with the bulk of a material, it was observed in \cite{ProdanJPA2021} that, when applied to half-spaces, this formalism reproduces all $C^\ast$-algebras and the associated exact sequences appearing in the standard bulk-boundary correspondences. Briefly, the closure of the orbit of the atomic lattice $\Ll$ under the translation action of $\RM^d$ on the space $\Cc(\RM^d)$ of closed subsets of $\RM^d$ supplies the hull of the pattern $\Omega_\Ll^\times$ and the transformation groupoid associated to the dynamical system $(\Omega_\Ll^\times, \RM^d)$. The latter admits an abstract transversal $\Xi_\Ll$ consisting only of those patterns in $\Omega_\Ll^\times$ which contain the origin of $\RM^d$. The reduction of the initial groupoid to $\Xi_\Ll$ supplies what we call the canonical \'etale groupoid $\Gg_\Ll$ associated to $\Ll$. The left regular representations of $\Gg_\Ll$ and their matrix amplifications produce translation-equivariant Hamiltonians (see subsection~\ref{Sec:GAlg}).

In sections~\ref{Sec:Groupoid} and \ref{Sec:ModelAlg}, we demonstrate how to compute transversals in the presence of boundaries and, as we shall see, they all share several common features. For the pattern seen by the experimenter from Fig.~\ref{Fig:Intro}(b), the outcome of the computation is illustrated in Fig.~\ref{Fig:T0}: The orbit of the pattern under translations has non-trivial accumulation points given by bulk and half-space patterns, which form distinct closed invariant subsets. As we shall see in subsection~\ref{Sec:Quarter}, the groupoid algebras corresponding to the restrictions of $\Gg_\Ll$ to these subsets and their complements supply all algebras, ideals and the associated exact sequences derived for cones of the discrete plane in \cite{ParkJOT1990} (see section~\ref{Sec:Quarter}). However, not all phenomena of higher-order bulk-boundary correspondences can be described inside the algebra generated by a single pattern. Our prescription for constructing the $C^\ast$-algebra that encapsulates the infinite size limit of a crystal with multiple distinct boundaries is as follows: For any infinite pattern $\Ll^\lambda\in \Cc(\RM^d)$ which can be obtained from an infinite-volume limit of finite samples (corresponding e.g. to a fixed observer position as in Figure~\ref{Fig:Intro}) one constructs the transversal $\Xi_{\Ll^\lambda}$ as sketched above. For the typical crystal this produces only a finite number of distinct subsets of $\Cc(\RM^d)$ and we take the (global) transversal $\Xi\subset\Cc(\RM^d)$ of the infinite crystal to be their union. As indicated in Fig.\ref{Fig:T0}(c),  the transversals of the patterns seen by different observers match along their shared boundaries. Not only the transversals but also the corresponding groupoids can be glued consistently using pushouts. The outcome is a groupoid $C^*$-algebra associated to $\Xi$, which can be interpreted very concretely using our team of observers.  Here is our exact statement formulated for a general context (see subsection~\ref{Sec:TS} for technical details):

\begin{proposition}\label{Prop:ModelAlg} Let $\Lambda$ be a finite set labeling uniformly discrete patterns $(\Ll^\lambda)_{\lambda\in \Lambda}$ in $\RM^d$. Then:
\begin{enumerate}[\rm i)]

\item For each pair $\lambda,\lambda'$, the intersection $\Xi_{\Ll^\lambda} \cap \Xi_{\Ll^{\lambda'}}$ in $\Cc(\RM^d)$ is a closed subset of the unit spaces of both $\Gg_{\Ll^\lambda}$ and $\Gg_{\Ll^{\lambda'}}$, invariant under the actions of both groupoids, and
\begin{equation}
\Gg_{\Ll^\lambda} \cap \Gg_{\Ll^{\lambda'}} = \left . \Gg_{\Ll^\lambda}\right |_{\Xi_{\Ll^\lambda} \cap \Xi_{\Ll^{\lambda'}}} = \left . \Gg_{\Ll^{\lambda'}}\right |_{\Xi_{\Ll^\lambda} \cap \Xi_{\Ll^{\lambda'}}}.
\end{equation}
\item The co-limit under the diagrams\footnote{The category of topological groupoids is co-complete \cite{BrownMN1976}.}
\begin{equation}
\begin{tikzcd}
	\Gg_{\Ll^\lambda}   & \Gg_{\Ll^{\lambda}} \cap \Gg_{\Ll^{\lambda'}} & \Gg_{\Ll^{\lambda'}}
		\arrow[leftarrowtail ,from=1-1, to=1-2]
	\arrow[rightarrowtail ,from=1-2, to=1-3]
\end{tikzcd}, \quad (\lambda,\lambda') \in \Lambda \times \Lambda,
\end{equation}
generates the \'etale groupoid $\Gg_\Xi = \bigcup_{\lambda \in \Lambda} \Gg_{\Ll^{\lambda}}$ with unit space $\Xi =\bigcup_{\lambda \in \Lambda} \Xi_{\Ll^{\lambda}}$.

\item For each $\Ss\in \Xi$ there is a representation $\pi_\Ss$ of $C^\ast \Gg_\Xi$ on $\ell^2(\Ss)$.

\item If the point group $\Sigma\subset O(d)$ acts via permutations of $\Lambda$ then this action gives rise to an action on the $C^\ast$-algebra $C^\ast \Gg_\Xi$. 

\end{enumerate} 
\end{proposition}
The property (ii) ensures that any two observers at different boundary positions whose observations in the infinite-size limit are described by the two inequivalent representations $\pi_{\Ll^\lambda}$ and $\pi_{\Ll^{\lambda'}}$ will obtain consistent results at the shared boundaries because the dynamics of electrons is determined by a single Hamiltonian from $C^\ast\Gg_\Xi$ and the observers merely experience it from different representations of this algebra.

The new $C^\ast$-algebraic framework announced above is one of the main results of the paper. As we shall see in section~\ref{Sec:ModelAlg}, all the algebraic structures seen above can be explicitly computed for the cases of interest to us. In all instances, we found the following common characteristics: 

\begin{proposition}\label{Prop:Filtrations} A $d$-dimensional crystal in the infinite-size limit is described by a transversal $\Xi\subset \Cc(\RM^d)$, which is invariant under the action of a finite group $\Sigma\subset SO(d)$, called the point group.
\begin{enumerate}[\rm i)]

\item The space $\Xi$ of units admits filtrations
\begin{equation}\label{Eq:Filt}
\{\Ll_\infty\}= \Xi_0 \subset \Xi_1 \subset \cdots \subset \Xi_d =\Xi,
\end{equation}
of length $d$ by closed subspaces that are invariant to the groupoid and point group actions. Here, $\Ll_\infty$ is the bulk lattice. 

\item In turn, this supplies a $\Sigma$-equivariant cofiltration of the groupoid $C^\ast$-algebra
\begin{equation}\label{Eq:CoFilt}
C^\ast \Gg_{\Xi_d}  \stackrel{\mathfrak p^{d}}{\twoheadrightarrow} C^\ast \Gg_{\Xi_{d-1}} \stackrel{\mathfrak p^{d-1}}{\twoheadrightarrow} \cdots \stackrel{\mathfrak p^{2}}{\twoheadrightarrow} C^\ast \Gg_{\Xi_1} \stackrel{\mathfrak p^{1}}{\twoheadrightarrow} C^\ast \Gg_{\Xi_0},
\end{equation}
where $\Gg_{\Xi_j}$ is the restriction of $\Gg_\Xi$ to $\Xi_j \subseteq \Xi$.
\item For this cofiltration, 
\begin{equation}
{\rm Ker}(C^\ast \Gg_{\Xi_r}  \twoheadrightarrow C^\ast \Gg_{\Xi_{r-1}}) =: C^\ast \Gg_{\Xi_r \setminus \Xi_{r-1}}
\end{equation} 
identifies the algebra of observables which are localized to the boundaries of codimension $r$.
\end{enumerate}
\end{proposition}

What precisely is considered a boundary of codimension $r$ as mentioned at point iii) is determined by the choice of the filtration made at point i) (see subsections~\ref{Sec:cube_mechanism} and \ref{Sec:Unifying}). The transversal $\Xi_d$ can select a subset of the crystalline geometry, comprised of some but not necessarily all patterns that occur in the infinite-volume limit of a crystal. To keep with the physical interpretation, there is then a unique choice for $\Xi_0,...,\Xi_{d-1}$ such that the ideals of Proposition~\ref{Prop:Filtrations} are algebras of observables localized at the boundaries of the respective codimension. In section~\ref{Sec:ModelAlg}, we demonstrate how Propositions~\ref{Prop:ModelAlg} and \ref{Prop:Filtrations} play out in the specific cases of quarter, square and cube geometries. We will use these examples in section~\ref{Sec:Mechanism} to identify the mechanism of higher-order bulk-boundary correspondences. 

For us, the final step $C^*\Gg_{\Xi_0}$ in the filtration will always describe a bulk material without boundaries. To this algebra and to the algebras $C^\ast \Gg_{\Xi_r \setminus \Xi_{r-1}}$ of boundary observables, one can assign topological invariants using equivariant K-functors $\KK_q$, which are equivariant homology theories for $C^*$-algebras. All gapped bulk materials and topologically protected boundary states give rise to elements of and are classified by those K-groups. The goal of K-theoretic bulk-boundary correspondence is to find maps between those groups, which explain the relation between bulk and boundary topological invariants. We construct those maps for higher-order bulk-boundary correspondences, which have the following properties:

\begin{theorem}\label{Th:Main} Consider the symmetry-adapted filtration $\{\Xi_{n}\}$  \eqref{Eq:Filt} such that the ideals of boundary states $C^*\Gg_{\Xi_{r}\setminus\Xi_{r-1}}$ will be localized to a selection of $r$-th order boundaries. Fix a subgroup $\Gamma\subset \Sigma \times \ZM_2 \times \ZM_2$ of the point group enhanced by Altland-Zirnbauer-type fundamental symmetries and a $\Gamma$-equivariant K-functor together with its suspensions $(\KK_q)_{q\in \ZM}$. Assume that at least for the specific value $q=\ast$ the functor the $\KK_\ast$ classifies stable homotopy classes of gapped symmetric Hamiltonians in the algebras above. One has:
\begin{enumerate}[\noindent \rm i)]
	\item The equivariant cofiltration~\eqref{Eq:CoFilt} induces a spectral sequence $(E^r_{p,q}, d^r_{p,q})$ whose terms $E^r_{0,q}$ are subgroups of $\KK_q(C^\ast\Gg_{\Xi_0})$ and $E^r_{p,q}$ for $p>1$ are subquotients of $\KK_{-p+q}(C^\ast\Gg_{\Xi_{p}\setminus \Xi_{p-1}})$. The differential $d^r_{0,q}: E^r_{0,q}\to E^r_{r,q+r-1}$ relates subquotients of bulk and boundary K-groups.
	\item A class $x\in\KK_\ast(C^\ast\Gg_{\Xi_0})$ is in the domain $E^r_{p,\ast}$ of $d^r_{p,\ast}$ if and only if there is a symmetric Hamiltonian $h$ in $M_N(\CM)\otimes C^*\Gg_{\Xi_d}$ such that dividing out all boundaries of codimension $r$ and greater via the surjection $(\mathfrak p^{r} \circ \cdots \circ \mathfrak p^{d})(h)$ results in a symmetric spectrally gapped Hamiltonian in $M_N(\CM)\otimes C^\ast\Gg_{\Xi_{r-1}}$ and such that its evaluation $(\mathfrak p^1 \circ \cdots \circ \mathfrak p^{d})(h)$ in the bulk represents the $K$-theory class $x$.
	\item If $d^r_{0,\ast}(x)$ is non-trivial for some $x\in\KK_\ast(C^*\Gg_{\Xi_0})$, then any symmetric Hamiltonian in $M_N(\CM)\otimes C^*\Gg_{\Xi_d}$  representing this $K$-theory class $x$ in the bulk and having a spectrally gapped image in $M_N(\CM)\otimes C^*\Gg_{\Xi_{r-1}}$ displays non-trivial topologically stable order-$r$ order boundary states. The image $d^r_{0,\ast}(x)$ as a coset in a subquotient of $\KK_{*-1}(C^*\Gg_{\Xi_{r}\setminus\Xi_{r-1}})$ enumerates all possible boundary states obtainable by choice of symmetric boundary condition.
\end{enumerate}
\end{theorem}
%

The statement (iii) is the punch line of our work. It shows that higher-order bulk-boundary correspondence is a stable and robust phenomenon protected by spectral gaps at suitable boundaries in combination with a specified crystalline symmetry group and is entirely explainable by operator K-theory. A bulk Hamiltonian will be said to be gappable at the boundaries of codimension $r$ if its K-theory class satisfies the equivalent conditions of Theorem~\ref{Th:Main}(ii). To identify all instances of order-$r$ bulk-boundary correspondence one needs to enumerate precisely the stable homotopy classes of gapped bulk Hamiltonians which are gappable at the boundaries of codimension $r-1$ but not at those of codimension $r$. The former are precisely those Hamiltonians whose K-theory class lies in the subgroup $E^r_{0,*}\subset \KK_\ast(C^*\Gg_{\Xi_0})$ and to then find out if they are gappable at the codimension $r$ boundaries one computes the differential $d^r_{0,\ast}$. Any non-trivial value of $d^r_{0,\ast}$ identifies a topological class of bulk Hamiltonians that delivers an order-$r$ bulk-boundary correspondence. Examples are supplied in sections~\ref{Sec:Mechanism} and \ref{Sec:Examples}. In particular, we will see in subsection~\ref{Sec:NoEqui} that for our crystalline examples the image of $d^r_{0,\ast}$ for $r\geq 2$ is trivial in the absence of symmetries, therefore the presence of a symmetry is a prerequisite to observe non-trivial higher-order bulk-boundary correspondences. 

\section{Ordinary Bulk-Defect Correspondences} \label{Sec:Groupoid}

The model $C^\ast$-algebras and the exact sequences relevant for the standard bulk-defect correspondence principles can all be generated by a mechanism described in \cite{ProdanJPA2021}, within the framework of specific (\'etale) groupoids and their associated $C^\ast$-algebras. Our framework for higher-order bulk-boundary correspondences builds on this formalism. The goal of this section is to introduce a proper background and to fix the notations and conventions. 

\subsection{Point patterns and their canonical groupoids}

For the start, we will be interested in the space $\Cc(G)$ of closed subsets of a locally compact second countable (lcsc) topological amenable group $G$, which is equipped with Fell's topology \cite{FellPAMS1962}. Then the space $\Cc(G)$ is automatically a compact Hausdorff $G$-space \cite[Remark 4.4.]{DonjuanTA2022}, with the $G$-action 
\begin{equation}
g \cdot C = \{x g^{-1}, \ x \in C\}, \quad g \in G, \quad C \in \Cc(G).
\end{equation} 
Throughout our exposition, a pattern in $G$ is simply an element of $\Cc(G)$. 

\begin{remark}{\rm The examples listed in our present work all simply use the abelian groups $G = \RM^d$. However the formalisms introduced in the present section and in section~\ref{Sec:Mechanism} are general enough to handle more general lcsc groups, such as groups of isometries of various Riemann manifolds. A relevant example \cite{MeslandJGP2024} is that of the Euclidean group $G=O(d)\ltimes \RM^d$ with patterns made up of artificial atoms whose macroscopic structure has a distinguishable orientation labeled by $O(d)$.
}$\Diamond$
\end{remark}  

We will be interested in the following classes of patterns:

\begin{definition}[\cite{Enstad1Arxiv2022}] For $\Ll\in \Cc(G)$ and $S \subseteq G$, one says that $\Ll$ is
\begin{enumerate}[\ {\rm 1)}]
\item $S$-separated if $|\Ll \cap g \cdot S| \leq 1$ for all $g \in G$;
\item $S$-dense if $|\Ll \cap g \cdot S| \geq 1$ for all $g \in G$.
\end{enumerate}
\end{definition}

\begin{definition}[\cite{Enstad1Arxiv2022}]\label{Def:US} A subset $\Ll \subset G$ is called uniformly separated if there exists a non-empty open set $U \subseteq G$ such that $\Ll$ is $U$-separated. The set $\Ll$ is called uniformly dense if there exists a compact subset $K \subseteq G$ such that $\Ll$ is $K$-dense. If $\Ll$ is both separated and relatively dense, the it is called a Delone set.
\end{definition} 


To a fixed pattern $\Ll \in \Cc(G)$ one associates its punctured hull\footnote{For uniformly separated patterns other than Delone sets, we have to exclude the empty set for $\Xi_\Ll$ in definition~\ref{Def:TheGroupoid} to be an abstract transversal.}
\begin{equation}
\label{eq:hull}
\Omega_{\Ll}^\times = \overline{\{g \cdot\Ll: \ g \in G\}}\setminus \emptyset.
\end{equation}
This produces a topological dynamical system $(G,\Omega_{\Ll}^\times)$ canonically associated to the pattern $\Ll$. The dynamical system gives rise to a transformation groupoid \cite[p.~5]{WilliamsBook2}. For $(G,\Omega_{\Ll}^\times)$, we denote this groupoid by $\tilde \Gg_{\Ll}=\Omega_{\Ll}^\times\rtimes G$ and its source and range maps by $\tilde{\mathfrak s}$ and $\tilde{\mathfrak r}$, respectively. We are more interested in a specific subgroupoid:

\begin{definition}\label{Def:TheGroupoid} If $e$ denotes the neutral element of $G$, then we define the canonical transversal of $\tilde \Gg_{\Ll}$ as
\begin{equation}
\Xi_{\Ll} =\overline{ \{ \Ss \in \Omega_{\Ll}^\times, \ e \in \Ss \}}
\end{equation}
and the canonical groupoid associated to $\Ll$ shall be the restriction 
\begin{equation}\label{Eq:GDefined}
\Gg_{\Ll} : = \left . \tilde \Gg_{\Ll}\right |_{\Xi_{\Ll}} = \tilde{\mathfrak s}^{-1}(\Xi_{\Ll}) \cap \tilde{\mathfrak r}^{-1}(\Xi_{\Ll}).
\end{equation}
\end{definition}

\begin{remark}{\rm For any non-empty open set $U \subseteq G$, the set of $U$-separated sets is closed in $\Cc(G)$ \cite{Enstad1Arxiv2022}. If $\Ll$ is $U$-separated, then the limit points of its orbit are thus themselves $U$-separated. As a consequence, all $\Ss \in \Xi_\Ll$ are uniformly separated. Additionally, since $\Xi_\Ll$ is a closed subset of the compact space $\Cc(G)$, it is automatically compact.
}$\Diamond$
\end{remark}

\begin{proposition}[\cite{Enstad1Arxiv2022}] If $\Ll$ is uniformly separated, then the space $\Xi_{\Ll}$ is an abstract transversal of $\tilde \Gg_{\Ll}$. As a result, $\tilde \Gg_{\Ll}$ and $\Gg_{\Ll}$ are equivalent in the sense of \cite{MuhlyJOT1987}. Furthermore $\Gg_{\Ll}$ is a lcsc \'etale groupoid with compact unit space $\Xi_{\Ll}$ in the Fell topology.
\end{proposition}

It will be useful to have an explicit characterization of $\Gg_\Ll$:

\begin{proposition}[\cite{MeslandJGP2024}]\label{Prop:GL} The topological groupoid $\Gg_{\Ll}$ canonically associated to the uniformly separated pattern $\Ll$ consists of:
\begin{enumerate}[\ \rm 1.]
\item The set $\Gg_{\Ll}$ of tuples
\begin{equation}
(g, \Ss) \in G \times \Xi_{\Ll}, \ g \in \Ss,
\end{equation}
equipped with the inversion map
\begin{equation}
(g,\Ss)^{-1} = (g^{-1},g \cdot \Ss)
\end{equation}
and with the lcsc topology inherited from $G \times \Cc(G)$.
\item The subset $\Gg_{\Ll}^2$ of composable elements
\begin{equation}
 \big((g',\Ss'),(g,\Ss)\big) \in \Gg_{\Ll} \times \Gg_{\Ll},   \  \Ss'=g\cdot \Ss,
 \end{equation}
equipped with the composition 
\begin{equation}
(g',\Ss') \cdot (g,\Ss) = (g' g, \Ss).
\end{equation}
\end{enumerate} 
\end{proposition}
\begin{remark}\label{Re:GAction}{\rm The source and range maps of $\Gg_{\Ll}$ are
\begin{equation}\label{Eq:RS}
\mathfrak s(g,\Ss) = (e,\Ss), \quad  \mathfrak r(g,\Ss) =  \left(e,g \cdot \Ss \right),
\end{equation}
and its space of units $\Gg^0_{\Ll}$ is naturally homeomorphic to $\Xi_{\Ll}$. Recall that the latter is a compact topological space. Another useful information is the action of $\Gg_{\Ll}$ on its space of units, which goes as follows: If $\Ss \in \Xi_\Ll$ and $\gamma \in \mathfrak s^{-1}(\Ss)$, then $\gamma = (g,\Ss)$ for some $g \in \Ss$ and $\gamma \cdot \Ss = g^{-1} \cdot \Ss$.
}$\Diamond$
\end{remark}

\subsection{Groupoid $C^\ast$-algebras and their physical significance}\label{Sec:GAlg}

Since $\Gg_\Ll$ is \'etale it comes equipped with a natural Haar system:

\begin{proposition} Any uniformly separated pattern $\Ll \in \Cc(G)$ carries a canonical $C^\ast$-algebra, the (full) groupoid $C^\ast$-algebra $C^\ast \Gg_{\Ll}$ corresponding to $\Gg_{\Ll}$ and to its counting measures. 
\end{proposition}

\begin{remark}{\rm  All \'etale groupoids encountered in this work are topologically amenable when considered with their Haar systems of counting measures, since they are groupoid-equivalent to transformation groupoids of amenable groups $G$. As such, there is no distinction between their reduced and full $C^\ast$-algebras \cite{DelarocheMEM2000}.
}$\Diamond$
\end{remark}

\begin{remark}\label{Re:LRReps}{\rm $C^\ast \Gg_{\Ll}$ has a family of covariant left-regular representations indexed by $\Xi_\Ll$, induced by the states $\Xi_\Ll \ni \Ss \mapsto \eta_\Ss(f) = f(e,\Ss)$,  $f \in C^\ast \Gg_\Ll$. They are supported on the Hilbert space $\ell^2\big (\mathfrak s^{-1}(\Ss)\big ) = \ell^2(\Ss)$ and act as
\begin{equation}\label{Eq:LRepGood}
\pi_\Ss(f)\, |g'\rangle = \sum_{g\in \Ss}f(gg'^{-1},g'\cdot \Ss)\, |g \rangle
\end{equation}
on the canonical basis of $\ell^2(\Ss)$. The matrix amplifications of the representations formalize the dynamics of electrons for the atomic arrangement $\Ll$, as experienced by observers located at different atom sites. Since the groupoid $\Gg_\Ll$ is amenable, the direct sum $\bigoplus_{\Ss\in \Xi_\Ll}\pi_\Ll$ is a faithful representation and a more detailed study readily shows that $\pi_\Ll$ alone is also faithful already. 
}$\Diamond$
\end{remark}

Let us now motivate the use of those operator algebras in physics. Since we focus on the low energy regime in the one-electron sector, the quantum dynamics of the electrons in a material is generated by a Hamiltonian of the type
\begin{equation}\label{Eq:H1}
H_\Ss = \sum_{x,x' \in \Ss} w_{x,x'}(\Ss) \otimes |x\rangle \langle x' |, \quad w_{x,x'}(\Ss) \in M_N(\CM),
\end{equation}
where the uniformly discrete $\Ss \subset G$ indicates the position of the atoms, $N$ is the number of relevant atomic orbitals and $w_{x,x'}(\Ss)$ are the so-called coupling or hopping matrices. Once the atomic species and their arrangement are fixed the electron dynamics is already fully determined by more fundamental laws of physics. As such, \eqref{Eq:H1} is the result of a map from the space of uniformly separated patterns to a family of coupling matrices indexed by pairs $(x,x') \in \Ss$. Each of them should vary continuously w.r.t local displacements of a finite set of points of $\Ss$. This continuity assumption is inherent to having consistent laboratory measurements, since there are necessarily deviations from an ideal lattice. Furthermore we assume that, in natural as well as synthetic materials, the coupling matrices will be too small to be resolved beyond a finite range, hence the coupling matrices $w_{x,x'}(\Ss)$ returned by an actual experiment may be assumed to vanish if $x'x^{-1}$ is outside a compact vicinity of the origin. We also impose that the coupling matrices satisfy the equivariance relation
\begin{equation}\label{Eq:Equiv1}
w_{g \cdot x,g \cdot x'}(g \cdot \Ss) = w_{x,x'}(\Ss), \ \forall \ g\in G,
\end{equation}
which means that the matrix elements are determined already by the equivalency class under $G$ of small local patches of the pattern around $x,x'$. 

For the translation group $G=\RM^d$ this can be seen as a direct consequence of the Galilean invariance of the fundamental laws of physics at low energies. The assumptions of equivariance, continuity and finite range then combine to the statement that there should exist a continuous function $f\in C_c(\Gg_\Ll, M_N(\CM))$ such that
\begin{equation}
\label{Eq:LRepMC}
w_{x,x'}(\Ss) = f(x-x',\Ss-x'),
\end{equation}
for all $\Ss\in \Xi_\Ll$. The relation \eqref{Eq:LRepMC} is in fact equivalent to \eqref{Eq:LRepGood} for an additively written group action.

In conclusion, all Hamiltonians of the type \eqref{Eq:H1} with coupling matrices satisfying the qualifications mentioned above can be generated from a left regular representation of $M_N(\CM) \otimes C^\ast \Gg_{\Ll}$ and, vice versa, those Hamiltonians are dense in the self-adjoint sector of $M_N(\CM) \otimes C^\ast \Gg_{\Ll}$. Therefore, the canonical $C^\ast$-algebra constructed from a given atomic arrangement can be equally justified by mathematical and physical means. This fundamental principle was first discovered by Jean Bellissard \cite{Bellissard1986} and it was further refined by Johannes Kellendonk \cite{KellendonkRMP95}. Developments that take into account the shape of the (artificial) atoms and are applicable to the many-electron sectors can be found in \cite{MeslandCMP2022,MeslandJGP2024}. Furthermore, the formalism can be easily adapted to account for the presence of various symmetries (see next subsection and subsection \ref{Sec:FSymm}).

\subsection{Automorphic actions}\label{Sec:Auto} The structures defined in the previous subsections behave naturally under the automorphisms of the locally compact group $G$ (which are throughout assumed to be continuous):

\begin{proposition}\label{Prop:GAction1} Let $\sigma \in Aut(G)$. Then:
\begin{enumerate}[{\rm i)}]
\item $\sigma$ induces a homeomorphism $\sigma: \Cc(G)\to \Cc(G)$;
\item For a fixed uniformly separated pattern $\Ll$, we have natural homeomorphisms
\begin{equation}
\Ll \to \sigma (\Ll), \quad \Omega^\times_\Ll \to \Omega^\times_{\sigma(\Ll)}, \quad \Xi_\Ll \to \Xi_{\sigma(\Ll)};
\end{equation}
\item There is a groupoid isomorphism 
\begin{equation}
\alpha_\sigma : \Gg_\Ll \to \Gg_{\sigma(\Ll)}, \quad (g,\Ss) \mapsto (\sigma(g),\sigma(\Ss)).
\end{equation}
\end{enumerate}
\end{proposition}

\begin{proof} i) This follows from \cite[Remark 4.4.]{DonjuanTA2022}. ii) If $e \in \Ss$, then $e$ also belongs to $\sigma(\Ss)$ because group automorphisms act trivially on the neutral element. Thus, the map $\Xi_\Ll \to \Xi_{\sigma(\Ll)}$ is obvious. iii) Let us first check the composition law. Taking two composable elements from $\Gg_\Ll$, $(g',g\cdot \Ss)$ and $(g,\Ss)$, we have
\begin{equation}
\alpha_\sigma (g',g\cdot \Ss) \cdot \alpha_\sigma (g,\Ss) =(\sigma (g'), \sigma(g\cdot \Ss)) \cdot (\sigma(g),\sigma(\Ss)).  
\end{equation}
If $\sigma$ is an automorphism, then $\sigma (g\cdot \Ss) = \sigma(g) \cdot \sigma(\Ss)$ and we can see that the two elements of $\Gg_\Ll$ remain composable after $\sigma$ is applied. Furthermore, we can complete the calculation and conclude
\begin{equation}
\alpha_\sigma((g',g\cdot \Ss)) \cdot \alpha_\sigma((g,\Ss)) = (\sigma(g'), \sigma(g)\cdot\sigma(\Ss)).
\end{equation}
On the other hand, 
\begin{equation}
\alpha_\sigma \big ((g',g\cdot \Ss) \cdot (g,\Ss)\big ) =(\sigma(g' g),\sigma(\Ss)),
\end{equation}
and the two results coincide as long as $\sigma$ is an automorphism of $G$. As for inversion, we have
\begin{equation}
\alpha_\sigma\big ((g,\Ss)^{-1}\big) = \alpha_\sigma(g^{-1},g\cdot \Ss)=\big (\sigma(g^{-1}),\sigma(g \cdot \Ss)  \big ),
\end{equation}
while
\begin{equation}
\big (\alpha_\sigma (g,\Ss)\big)^{-1} = \big (\sigma(g)^{-1},\sigma(g) \cdot  \sigma(\Ss)  \big ).
\end{equation}
The two results are identical if $\sigma$ is a group automorphism.
\end{proof}

\begin{corollary}\label{Cor:GAction1} The groupoid isomorphism induced by an automorphism of $G$ lifts to an isomorphism of $C^\ast$-algebras
\begin{equation}
C^\ast \Gg_\Ll \ni f \mapsto \alpha_\sigma^\ast(f) : = f\circ \alpha^{-1}_\sigma \in C^\ast \Gg_{\sigma(\Ll)}
\end{equation} 
\end{corollary}
\begin{proof}  $\alpha_\sigma$ is bijective and preserves the Haar system. Then the statement follows from \cite[Prop.~2.7]{AustinNYJM2021}.\end{proof}

\begin{example}{\rm For our concrete physical systems, $G = \RM^d$. Then the group of rotations, proper or otherwise, embeds in $Aut(\RM^d)$. Point groups are finite groups of rotations and they will enter in our analysis via the actions described in Corollary~\ref{Cor:GAction1}.
}$\Diamond$
\end{example}

\subsection{Mechanism of ordinary bulk-defect correspondences}\label{Sec:OrdinaryBD}
A geometrical defect in a pattern, such as a local modification or boundary, manifests itself as a feature that can be made to disappear in the limit by translating it to infinity (which is well-described by limits in the Fell topology). The groupoid $\Gg_{\Ll}$  associated to a pattern has a canonical action on its unit space $\Xi_{\Ll}$ which was spelled out in Remark~\ref{Re:GAction}. It was observed in \cite{ProdanJPA2021} that, in the presence of a geometrical defect, the space of units $\Xi_{\Ll}$ can display one or more closed subspaces that are invariant against the mentioned groupoid action. Furthermore:

\begin{proposition}[\cite{WilliamsBook2}, Thm.~5.1] If $\Xi_{\Ll}^\infty$ is a closed and invariant subspace of the unit space $\Xi_\Ll$ and $\Xi_{\Ll}^c$ is its open complement, then $\left . C^\ast \Gg_{\Ll}\right |_{\Xi_{\Ll}^c}$ is a closed ideal of the groupoid $C^\ast$-algebra and we have the following short exact sequence 
\begin{equation}\label{Eq:Ext1}
0 \rightarrow \left . C^\ast \Gg_{\Ll}\right |_{\Xi_{\Ll}^c} \stackrel{i}{\rightarrow} C^\ast \Gg_{\Ll} \stackrel{\mathfrak{p}}{\rightarrow} \left .  C^\ast \Gg_{\Ll}\right |_{\Xi_{\Ll}^\infty} \rightarrow 0.
\end{equation}
\end{proposition}

\begin{figure}[t]
\center
\includegraphics[width=0.7\textwidth]{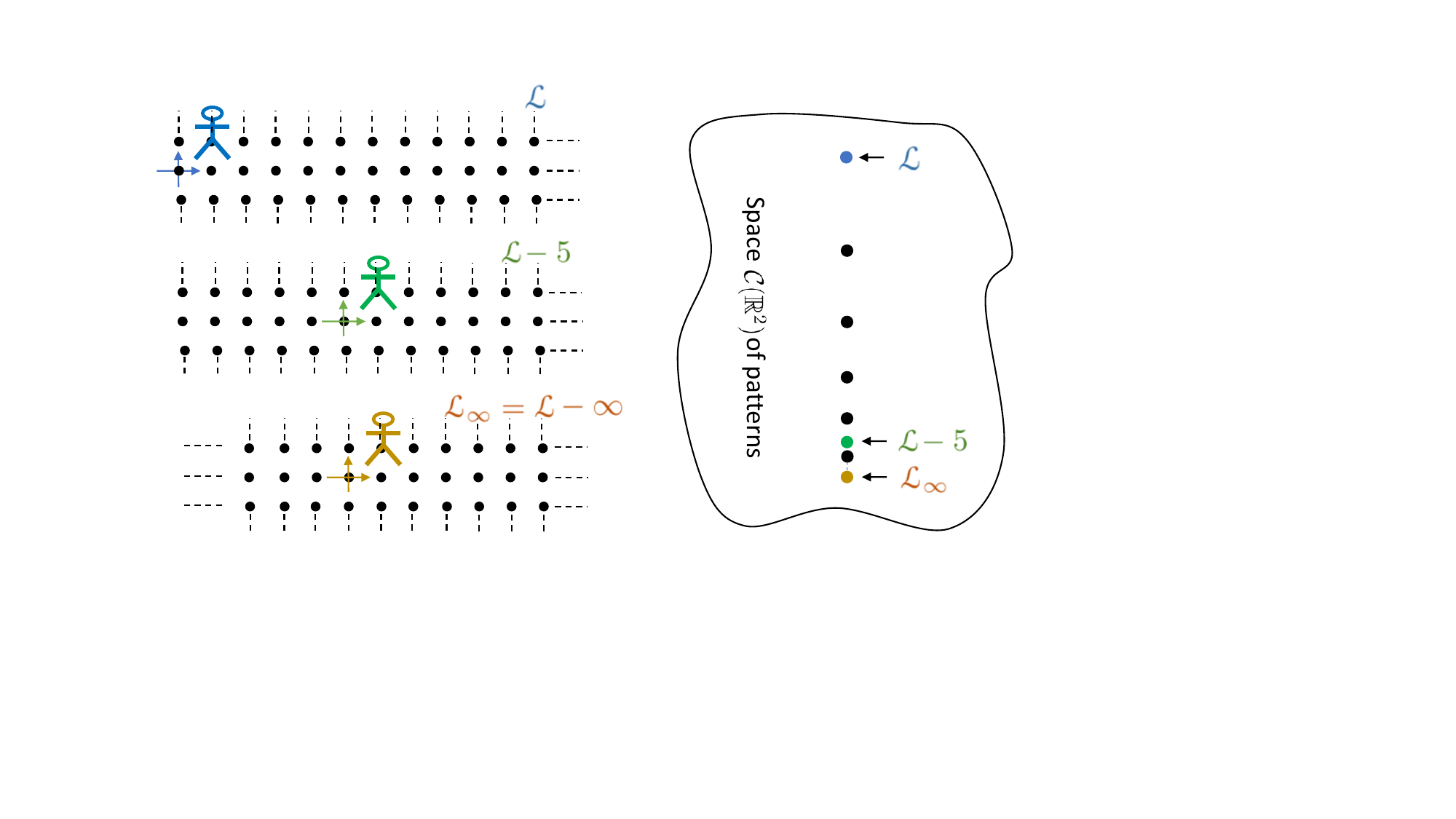}\\
  \caption{\small The process of deriving $\Xi_{\Ll}$ for the pattern $\Ll = \NM \times \ZM$ from Example~\ref{Ex:T1}.
}
 \label{Fig:T1}
\end{figure}


\begin{example}\label{Ex:T1}{\rm Consider the case of the pattern $\Ll = \NM \times \ZM^{d-1} \subset G=\RM^d$, hence a truncated lattice with a boundary at $x=0$. It is a simple exercise to derive the transversal $\Xi_{\Ll}$ for this pattern. First, if $\Ss \in \Xi_\Ll$, its orbit under the groupoid action is $\Oo(\Ss)=\{\Ss - x, \  x \in \Ss\}$ (see Remark~\ref{Re:GAction}). Then it is easy to show that
\begin{equation}
\Xi_{\Ll}=\Oo(\Ll)\cup \{\ZM^d\}
\end{equation}
as a disjoint union, i.e. the periodic lattice $\ZM^d \in \Cc(\RM^d)$ is in the closure of $\Oo(\Ll)$ and that it and the empty set are the only limit points. Furthermore, $\ZM^d$ is invariant under the groupoid action of $\Gg_\Ll$ when viewed as an element of the unit space $\Xi_\Ll$. 

\begin{remark}
    {\rm An intuitive way to see this without computation is to note that the Fell topology on $\Cc(\RM^d)$ coincides with the one generated by the Hausdorff metric on the closed subsets of the one point compactification of $\RM^d$. The latter is popular in image processing \cite{BarnsleyBook1993}, and indeed one can visualize $\Oo(\Ll)$ as the sequence of patterns seen by an observer moving along $\Ll$ (see Figure~\ref{Fig:T1}). The Fell topology expresses very well that the translates $\Ll-x$ become visually more and more similar to an infinite lattice without boundary as the observer position $x$ becomes more and more distant from the edge.} $\Diamond$
\end{remark}

The representation of the groupoid algebra $\Hh = C^\ast \Gg_\Ll$ on $\ell^2(\Ll)$ describes the physical Hamiltonians supported by a half-space pattern. The subset $\Xi_\Ll^\infty:= \{\ZM^d\}$ is translation-invariant and its groupoid algebra $C^\ast \Gg_{\Ll}\vert_{\Xi^\infty_\Ll}$, equal to the group $C^*$-algebra $C^\ast \ZM^d$,  models the bulk of the system. Since $\Xi_\Ll^c:= \Oo(\NM \times \ZM^{d-1})$ is open, the associated groupoid algebra $\Ee=C^\ast \Gg_{\Ll}\vert_{\Xi^c_\Ll}$ is non-unital and only contains elements that vanish far away from the boundary. Hence, it models the observations made around the boundary. Furthermore, the exact sequence~\eqref{Eq:Ext1} is isomorphic to the Toeplitz extension used in the standard bulk-boundary correspondence \cite{KellendonkRMP2002,ProdanSpringer2016} (see \cite{ProdanJPA2021} for the explicit mapping).
}$\Diamond$

%
\end{example}

\begin{definition}\label{Def:Defect} We say that the pattern $\Ll$ contains an elementary geometric defect if the transversal $\Xi_\Ll$ has a unique proper closed subset that is invariant under the $\Gg_\Ll$ action.
\end{definition}

Example~\ref{Ex:T1} shows that cutting a material in half results in an elementary defect. In \cite{ProdanJPA2021}, it was shown explicitly that disclinations are also elementary defects in the sense of Definition~\ref{Def:Defect}, and it was pointed further that all standard material defects (see \cite{MerminRMP1979} for a catalog) can be characterized in this way. Furthermore:

\begin{proposition}[\cite{ProdanJPA2021}] Elementary defects can be classified by the isomorphism class of the extension $ext({\Xi_\Ll^\infty})$ in \eqref{Eq:Ext1} associated to the unique decomposition of the transversal. As a consequence \cite{KasparovJSM1981}, each elementary geometric defect carries the topological charge
\begin{equation}
[ext({\Xi_\Ll^\infty})]_1 \in KK_1 (C^\ast \Gg_\Ll|_{\Xi_{\Ll}^\infty},C^\ast \Gg_\Ll|_{\Xi_{\Ll}^c}),
\end{equation}
valued in the complex $KK$-theory.
\end{proposition}

The ordinary bulk-defect correspondences involve elementary defects and all can be explained by the connecting maps induced in the appropriate $K$-theories by the exact sequence~\eqref{Eq:Ext1}. In fact, one can describe them in a unifying way using a key observation from \cite{BourneMPL2015}. In the absence of any symmetry, this description is as follows:

\begin{proposition}[\cite{ProdanJPA2021}]\label{Prop:SpecStat}
Associated to the exact sequence \eqref{Eq:Ext1} there is a natural connecting homomorphism $\partial_0$ which makes the following sequence of complex K-groups exact
\begin{equation}\label{Eq:k_bulkdefect}... \rightarrow K_0(C^\ast \Gg_{\Ll}) \stackrel{\mathfrak p_*}{\rightarrow}  K_0(C^\ast \Gg_{\Ll} |_{\Xi_{\Ll}^\infty}) \stackrel{\partial_0} \rightarrow K_1(C^\ast \Gg_{\Ll} |_{\Xi_{\Ll}^c})  \stackrel{i_*}{\rightarrow} K_1(C^\ast \Gg_{\Ll})\rightarrow...
\end{equation}
Let $h_b \in  M_N(\CM) \otimes C^\ast \Gg_{\Ll}|_{\Xi_{\Ll}^\infty}$ be a spectrally gapped bulk Hamiltonian and let $[\gamma_{h_b}]_0$ be the class of the spectral projection onto the spectrum below the gap in the complex $K$-group $K_0( C^\ast \Gg_{\Ll}|_{\Xi_{\Ll}^\infty})$.  If the class $\partial_0([\gamma_{h_b}]_0)\in K_{1} (C^\ast \Gg_\Ll|_{\Xi_{\Ll}^c})$ is nontrivial, then any lift of $h_b$ to $M_N(\CM) \otimes C^\ast \Gg_\Ll$ under the map $\mathfrak{p}$ in \eqref{Eq:Ext1} has spectrum inside the spectral  gap of $h$.
\end{proposition}

\begin{remark}\label{remark 223}{\rm The exactness of \eqref{Eq:k_bulkdefect} means that the class $[\gamma_{h_b}]_0$ has a pre-image in $K_0(C^\ast \Gg_\Ll)$ if and only if $\partial_0([\gamma_{h_b}]_0)$ is trivial. The occurrence of defect-localized states in the bulk gap protected by a K-theory class in $K_{1} (C^\ast \Gg_\Ll|_{\Xi_{\Ll}^c})$ is therefore exactly the obstruction to this lifting problem in K-theory. This characterization will be crucial when discussing higher-order correspondences.
}$\Diamond$
\end{remark}
\begin{remark}{\rm 	
An insightful way to write the connecting map in \eqref{Eq:k_bulkdefect} is as the Kasparov product 
$\partial_0([\gamma_{h_b}]_0)= [\gamma_{h_b}]_0 \times [ext({\Xi_\Ll^\infty})]_1$. It highlights that the connecting maps are not just plain homomorphisms of abelian groups but have very particular properties, such as compatibility with the associativity of the Kasparov product. Moreover, the algebraic structure on the extension classes encodes precise and useful topological information.
}$\Diamond$
\end{remark}
\begin{remark}{\rm For simplicity we ignored symmetries and just used complex K-theory in this section. If the short exact sequence is equivariant under a finite group $\Gamma$ then all of the above will of course also true for equivariant K-theory if the Hamiltonians are assumed to be symmetric.}
\end{remark}

To conclude, the ordinary bulk-defect correspondences stem from filtrations $\Xi_\Ll^\infty \subset \Xi_\Ll$ of length 1 of the transversals of the patterns with elementary geometric defects. The higher-order bulk-boundary correspondences will be associated with filtrations of the unit space that have lengths strictly larger than 1.

\section{Higher-order correspondences: Building the $C^\ast$-models}
\label{Sec:ModelAlg}

In this section, we supply the technical background for Propositions~\ref{Prop:ModelAlg} and \ref{Prop:Filtrations} and demonstrate how they play out for the three concrete cases of quarter, square and cube geometries.

\subsection{Technical statements}\label{Sec:TS}
\begin{proposition}\label{Prop:GGlobal} Let $(\Ll^\lambda)_{\lambda\in \Lambda}$ be a finite family of $U$-discrete subsets of $\Cc(G)$ and define their global transversal as
\begin{equation}
\Xi : = \mathsmaller{\bigcup}\nolimits_{\lambda\in \Lambda} \,  \Xi_{\Ll^\lambda}.
\end{equation}
	Then $\Gg_\Xi := \bigcup_{\lambda\in \Lambda} \Gg_{\Ll^\lambda}$ can be given the structure of a lcsc \'etale groupoid with unit space $\Xi$ and the same algebraic relations and topology as in Proposition~\ref{Prop:GL}, with $\Xi_\Ll$ replaced by $\Xi$. It is the co-limit under the diagrams induced by the inclusion maps
	\begin{equation}\label{Eq:ColimitDiag}
		\begin{tikzcd}
			\Gg_{\Ll^\lambda}   & \Gg_{\Ll^{\lambda}} \cap \Gg_{\Ll^{\lambda'}} & \Gg_{\Ll^{\lambda'}}
			\arrow[leftarrowtail ,from=1-1, to=1-2]
			\arrow[rightarrowtail ,from=1-2, to=1-3]
		\end{tikzcd}, \quad (\lambda,\lambda') \in \Lambda \times \Lambda,
	\end{equation}
	i.e. it is the smallest topological groupoid such that each of the diagrams
	\begin{equation}\label{Eq:GSquare}
		\begin{tikzcd}
			\Gg_\Xi & \Gg_{\Ll^{\lambda'}} \\
			\Gg_{\Ll^\lambda}   & \Gg_{\Ll^\lambda} \cap \Gg_{\Ll^{\lambda'}}
			\arrow[leftarrowtail,from=1-1, to=1-2]
			\arrow[rightarrowtail,from=2-1, to=1-1]
			\arrow[leftarrowtail ,from=2-1, to=2-2]
			\arrow[rightarrowtail ,from=2-2, to=1-2]
		\end{tikzcd}
	\end{equation}
	commutes. Furthermore, $\Gg_\Xi$ is amenable.
\end{proposition}
\begin{proof} Consider a pair $\lambda$ and $\lambda'$. Then $\Oo(\Ll^\lambda)$ and $\Oo(\Ll^{\lambda'})$ either coincide or are disjoint. In the first case, $\Xi_{\Ll^\lambda} = \Xi_{\Ll^{\lambda'}}$ while in the second case $\Xi_{\Ll^\lambda} \cap \Xi_{\Ll^{\lambda'}}$ is a closed subset invariant under the actions of both groupoids, which can be very well the empty set. In both cases,
\begin{equation}
\Gg_{\Ll^{\lambda}} \cap \Gg_{\Ll^{\lambda'}} =\left . \Gg_{\Ll^{\lambda}}\right |_{\Xi_{\Ll^{\lambda}} \cap \Xi_{\Ll^{\lambda'}}}=  \left . \Gg_{\Ll^{\lambda'}}\right |_{\Xi_{\Ll^{\lambda}} \cap {\Xi_{\Ll^{\lambda'}}}},
\end{equation}
hence $\Gg_{\Ll^{\lambda}} \cap \Gg_{\Ll^{\lambda'}}$ is a (full) subgroupoid, and
\begin{equation}\label{Eq:G2}
\Gg_{\Ll^\lambda}^2 = \big (\left . \Gg_{\Ll^\lambda}\right |_{\Xi_{\Ll^\lambda} \setminus (\Xi_{\Ll^\lambda} \cap \Xi_{\Ll^{\lambda'}})} \big )^2 \cup \big (\left . \Gg_{\Ll^\lambda}\right |_{ \Xi_{\Ll^\lambda} \cap \Xi_{\Ll^{\lambda'}}} \big )^2.
\end{equation}
A statement similar to \eqref{Eq:G2} holds for $\Ll^{\lambda'}$. Now, co-limits and pushouts in the category of groupoids are described in \cite{SchubertBook} and in more detail in \cite{MacDonaldCMB2009}. By following these descriptions and by taking into account the above facts, one concludes that the co-limit under \eqref{Eq:ColimitDiag} in the category of (algebraic) groupoids is just the union of the groupoids. Similarly, in the category of topological spaces, the co-limit under \eqref{Eq:ColimitDiag} is the union of topological spaces $\Gg_{\Ll^\lambda}$. Therefore, the co-limit in the category of topological groupoids must coincide with $\bigcup_{\lambda\in \Lambda} \Gg_{\Ll^\lambda}$, if the latter is a topological groupoid, that is, if the algebraic structure is compatible with the topology, and this is the case. Indeed, by construction, the elements of the so constructed algebraic groupoid are pairs $(g,\Ss)$ with $\Ss \in \Xi$ and $g \in \Ss$, and the inversion and composition of such pairs are exactly as described in  Proposition~\ref{Prop:GL}, if we replace $\Xi_\Ll$ by $\Xi$. Also, since each $\Gg_{\Ll^\lambda}$ inherits its topology from $G \times \Cc(G)$, their union also share this atribute. The conclusion is that the topology of $\Gg_\Xi$ is also as described in Proposition~\ref{Prop:GL}, and this topology is automatically compatible with the algebraic structure.

We now show that $\Gg_\Xi$ is \'etale, i.e. that the range map $\mathfrak r$ is a local homeomorphism.  Since every pattern from $\Xi$ is $U$-uniformly separated for a fixed open subset $U \in \mathcal{C}(G)$, the statement from Lemma~3.9 from \cite{Enstad1Arxiv2022} continues to apply without modifications. This statement assures us that the map $V \times \Xi \to \Cc(G)$ given by $(g, \Ss) \mapsto g \cdot \Ss$ is homeomorphism onto its image, for any open subset $V \in G$ such that $V \cap V^{-1} \subset U$. From here, we can follow the arguments from Proposition~3.10 from \cite{Enstad1Arxiv2022}. By definition, the sets $U_{\gamma,V,\Gamma} = ( \gamma \cdot V \times \Gamma) \cap \Gg_{\Xi}$ form a basis for the topology of $\Gg_{\Xi}$, where $\gamma$ ranges over $G$, $V$ ranges over all symmetric neighborhoods of the identity with $V^2 \subseteq U$ and $\Gamma$ ranges over all open sets in $\Xi$. Note that $(\gamma \cdot V)^{-1} (\gamma \cdot V) = V^2 \subseteq U$, hence we are in the conditions of \cite{Enstad1Arxiv2022}[Lemma~3.9] and the map $\gamma \cdot V \times \Xi \to \Cc(G)$ given by $(g,\Ss) \mapsto g \cdot \Ss$ is a local homeomorphism. By restricting this map to $U_{\gamma,V,\Gamma}$, we obtain a homeomorphism onto an image contained in $\Xi$, which coincides with the restriction of the range map on $U_{\gamma,V,\Gamma}$. Hence, $\Gg_\Xi$ is \'etale.

For the remaining statement, we will use an equivalent characterization of $\Gg_\Xi$. Defining a hull 
\begin{equation}
\Omega_{\Xi}^\times = \overline{\{g \cdot\Ll: \ g \in G, \Ll \in \Xi\}}\setminus \emptyset,
\end{equation}
one again obtains a locally compact $G$-invariant subset of $\Cc(G)$ and can therefore define a transformation groupoid $\tilde{\Gg}_\Xi := \Omega_\Xi^\times \rtimes G$. Restricting it to the invariant subset $\tilde{\mathfrak{s}}^{-1}(\Xi) \cap \tilde{\mathfrak{r}}^{-1}(\Xi)$, one obtains precisely $\Gg_\Xi$.  Now, \cite[Prop. 3.8]{Enstad1Arxiv2022} can be used without alteration to prove that $\Xi$ is an abstract transversal of $\Omega_\Xi^\times \rtimes G$, therefore the amenability of $\Gg_\Xi$ follows from the amenability of $G$, which is assumed throughout.
\end{proof}

Since we are dealing with a lcsc \'etale groupoid, there is a natural $C^\ast$-algebra:
\begin{corollary} 
\label{Cor:Pull_backs}	
To any finite set of patterns $(\Ll^\lambda)_{\lambda\in \Lambda}$ with transversal $\Xi=\bigcup_{\lambda\in \Lambda} \Xi_{\Ll^\lambda}$, we can associate the groupoid $C^\ast$-algebra $C^\ast\Gg_\Xi$ corresponding to the groupoid $\Gg_\Xi$ constructed in Proposition~\ref{Prop:GGlobal} and to its Haar system of counting measures. The $C^\ast$-algebras $C^\ast\Gg_{\Ll^\lambda}$ fit into commutative diagrams dual to \eqref{Eq:GSquare}
	\begin{equation}
		\begin{tikzcd}
			C^*(\Gg_{\Ll^{\lambda}}\cap\Gg_{\Ll^{\lambda'}}) & C^*\Gg_{\Ll^{\lambda'}} \\
			C^*\Gg_{\Ll^\lambda}   &  C^*\Gg_\Xi
			\arrow[twoheadleftarrow,from=1-1, to=1-2]
			\arrow[twoheadrightarrow,from=2-1, to=1-1]
			\arrow[twoheadleftarrow ,from=2-1, to=2-2]
			\arrow[twoheadrightarrow ,from=2-2, to=1-2]
		\end{tikzcd}
	\end{equation}
 with the homomorphisms induced by the surjections $C_c(\Gg_{\Xi})\to C_c(\Gg_{\Ll^\lambda})$.
\end{corollary}

\begin{remark}{\rm By construction $\Xi_{\Ll^\lambda}\subset \Xi $ is a closed invariant subset, then the left regular representations of the $C^\ast\Gg_\Xi$ supply representations of $C^*\Gg_{\Ll^\lambda}$ on the Hilbert spaces $\ell^2(\Ll^\lambda)$. As explained in the introduction, if we interpret $C^*\Gg_{\Xi}$ as an algebra which encodes observations made by multiple observers placed on different locations of a crystal in the infinite-size limit, then the sole consistency condition necessary is that all their observations come from representations of the same element of $C^\ast \Gg_\Xi$. 
}$\Diamond$
\end{remark}

The main motivation of the construction is that it can be used to implement group actions which permute different patterns. Indeed, as a consequence of Corollary~\ref{Cor:GAction1}, we have:

\begin{corollary}
	If $\sigma\in Aut(G)$ maps $\Xi$ into itself, then it gives rise to an automorphism of $C^\ast \Gg_\Xi$.
\end{corollary} 

Let us also comment briefly on ideals in those algebras:
\begin{proposition}[\cite{WilliamsBook2}, Thm.~5.1]
\label{def:ideal_abbreviation}
If $\tilde{\Xi}\subset \Xi$ is a closed and $\Gg_\Xi$-invariant subset then restriction of the unit space to the open subset
$$C^*\Gg_{\Xi\setminus \tilde{\Xi}}:= C^*\Gg_{\Xi}\rvert_{\Xi\setminus \tilde{\Xi}}$$
gives a closed ideal in $C^*\Gg_{\Xi}$ such that $C^*\Gg_{\Xi}/C^*\Gg_{\Xi\setminus \tilde{\Xi}}\simeq C^*\Gg_{\tilde{\Xi}}$.
\end{proposition}
Since the unit space $\Xi\setminus\tilde{\Xi}$ is open, $C^*\Gg_{\Xi\setminus \tilde{\Xi}}$ will, by definition of the groupoid algebra as $C^*$-completion of the convolution algebra $C_c(\Gg_{\Xi\setminus \tilde{\Xi}})$, only contain those elements of $C^*\Gg_{\Xi}$ that asymptotically vanish in regions that look more and more like patterns contained in $\tilde{\Xi}$, hence by choosing appropriate subsets of $\Xi$ one can isolate elements localized to certain boundaries or other geometric defects.

In the remainder of this section, we will consider examples how those constructions play out to demonstrate how the unit spaces glue naturally for typical crystals and also determine filtrations of the unit spaces by closed invariant subsets.

\subsection{Quarter geometry}\label{Sec:Quarter}

\begin{figure}[t]
\center
\includegraphics[width=\textwidth]{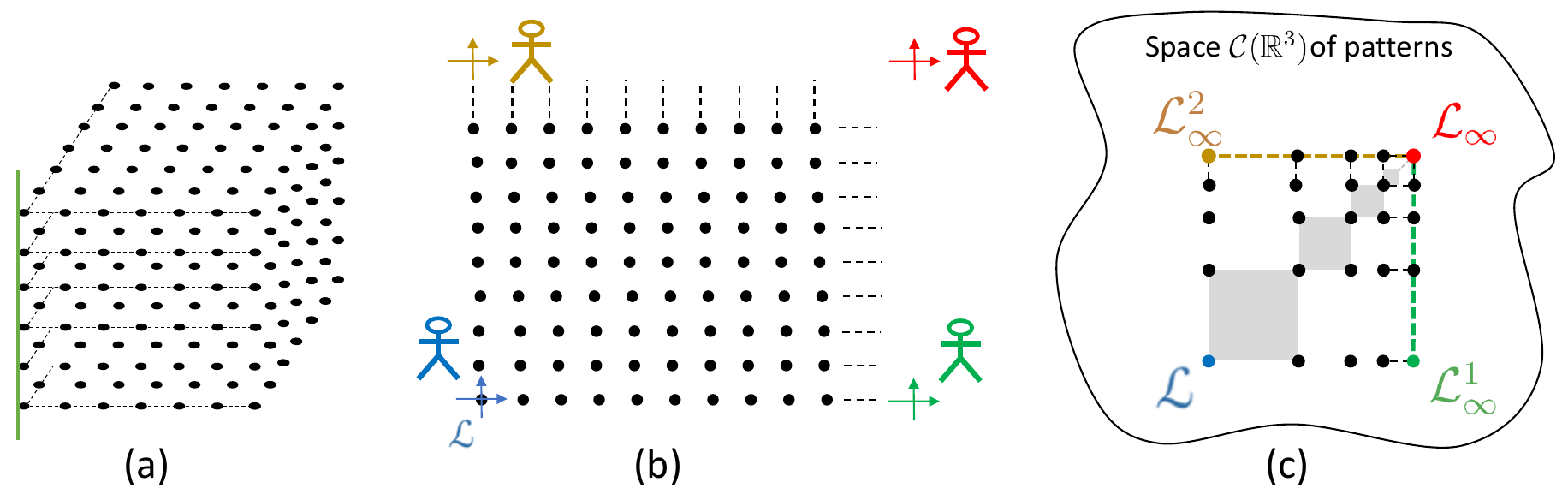}\\
  \caption{\small (a) The pattern $\Ll = \NM \times \NM \times \ZM \subset \RM^3$ considered in subsection~\ref{Sec:Quarter}. (b) A top view  of the pattern indicating directions in which an observer will find different asymptotic patterns. (c) Visual representation of the unit space $\Xi_\llcorner$ inside the space $\Cc(\RM^3)$.
}
 \label{Fig:T2}
\end{figure}

We analyze here the pattern $\Ll = \NM \times \NM \times \ZM^{d-2} \subset G=\RM^d$, shown in Fig.~\ref{Fig:T2}(a) for the case $d=3$. The corresponding groupoid $\Gg_\Ll$ will form a building block of the more complicated later examples. We will denote the canonical groupoid for this pattern by $\Gg_\llcorner$. It is easy to compute the transversal:

\begin{proposition} The transversal of the quarter pattern is the disjoint union 
\begin{equation}\label{Eq:XiCorner}
\Xi_\llcorner=\Oo(\Ll)\cup \Oo(\Ll^1_\infty) \cup \Oo(\Ll^2_\infty)\cup \{\Ll_\infty\}
\end{equation}
where, $\Ll_\infty = \ZM^d$ and $\Ll^1_\infty= \NM\times \ZM\times \ZM^{d-2}$ and $\Ll^2_\infty=\ZM \times \NM \times \ZM^{d-2}$ are the limits in the Fell topology of the translates $\Ll-n e_i$ for $n\to \infty$, with $e_i$ being the unit vectors for the first two directions.
\end{proposition}

One can again visualize the different limits by imagining the different patterns as seen by an observer moving along the quarter-space in different directions, see (see Figure~\ref{Fig:T2}). 
 
This transversal has several closed and invariant proper subsets such as $\{\Ll_\infty\}$, $\Xi_i^\infty=\Oo(\Ll^i_\infty)\cup \{\Ll_\infty\}$ and, more importantly,
 \begin{equation}
 \Xi_{\llcorner}^\infty=\Oo(\Ll^1_\infty)\cup\Oo(\Ll^2_\infty)\cup \{\Ll_\infty\} = \Xi_{\Ll_\infty^1} \cup \Xi_{\Ll_\infty^2}.
 \end{equation}
Those subsets are illustrated in Fig.~\ref{Fig:T2Deco}. As a result, we can generate several filtrations of $\Xi_\llcorner$ by closed invariant subsets. 

Let us now introduce some symmetry. The only symmetry which leaves the quarter-space invariant and does not involve the last $d-2$ directions is the mirror operation along the diagonal hyper-plane which interchanges the first two coordinates, implemented by $ \Sigma \subset SO(d)$ isomorphic to $\ZM_2$.
 
\begin{proposition} The space of units has a unique filtration
\begin{equation} 
 \{\Ll_\infty\} = \Xi_0 \subset \Xi_1 \subset  \Xi_2 = \Xi_{\Ll}
\end{equation}
of length $2$ by closed proper subsets that are invariant under the actions of both the groupoid  and $\Sigma$, namely $\Xi_1 =\Xi^\infty_\llcorner$. In the dual picture, this supplies an equivariant cofiltration
\begin{equation}\label{Eq:Ext17}
\bar \Qq \stackrel{\bar {\mathfrak p}^2}{\twoheadrightarrow} \bar \Pp \stackrel{\bar {\mathfrak p}^1}{\twoheadrightarrow} \Bb
\end{equation}
with $\bar \Qq : = C^\ast \Gg_\Ll$, $\bar \Pp: =  C^\ast \Gg_{\Xi^\infty_\llcorner} $ and $\Bb:=C^\ast \Gg_{\Ll_\infty}$. \end{proposition}

To prepare the ground for section~\ref{Sec:Mechanism}, we introduce and characterize useful ideals in the $C^\ast$-algebras listed above. The kernel of the epimorphism $\bar {\mathfrak p}^1$ is the $C^\ast$-algebra 
\begin{equation}
\bar \Cc : = \left . C^\ast \Gg_\llcorner\right |_{\Xi_\llcorner \setminus \Xi_\llcorner^\infty} = C^\ast \Gg_{\Xi_\llcorner \setminus \Xi_\llcorner^\infty},
\end{equation} 
which we will refer to as the corner algebra, since it relates to physical observations made around the corner of the quarter-space (for simplicity we will also call this a corner for $d>2$, even though it corresponds to a hinge in $d=3$ and is more generally called a ridge in polyhedral geometry). The kernel of the epimorphism $\bar {\mathfrak p}^0$ is the $C^\ast$-algebra 
\begin{equation}
\bar \Ff : = \left . C^\ast \Gg_\llcorner\right |_{\Xi^\infty_\llcorner \setminus \{\Ll_\infty\}} = C^\ast \Gg_{\Xi^\infty_\llcorner \setminus \{\Ll_\infty\}},
\end{equation} 
which we will call the face algebra because it relates to the physical observations around the facets and away from the corner. Additionally, as in Example~\ref{Ex:T1}, we can define $\Hh_i : = \left . C^\ast \Gg_\llcorner\right |_{\Xi^\infty_i}$ and $\Ff_i : = \left . C^\ast \Gg_\llcorner\right |_{\Xi^\infty_i \setminus \{\Ll_\infty\}}$, which are the half-space and face $C^\ast$-algebras for the two distinct faces. It is easy to see that $\bar \Ff = \Ff_1 \oplus \Ff_2$.

\begin{figure}[t]
\center
\includegraphics[width=0.8\textwidth]{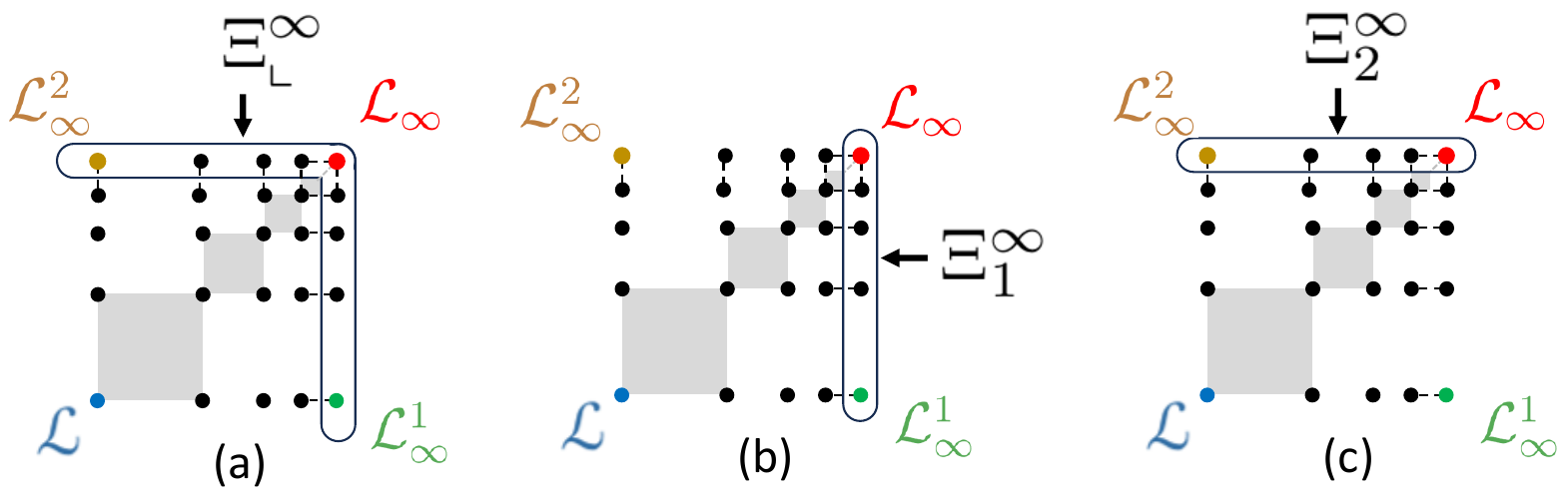}\\
  \caption{\small Illustrations in the space $\Cc(\RM^d)$ of patterns, showing the relevant closed and invariant subspaces of the unit space $\Xi_\llcorner$ for the quarter geometry described in subsection~\ref{Sec:Quarter}. 
}
 \label{Fig:T2Deco}
\end{figure}

\begin{remark}\label{Re:Park}{\rm For $d=2$, the $C^\ast$-algebras encountered above can be identified with $C^\ast$-algebras from the theory of quarter-space Toeplitz operators. Note that, instead of a simple quarter-space at right-angles, we could have more generally defined $\Ll$ to have angles $\alpha$, $\beta$ w.r.t. the horizontal, which would merely modify the asymptotic half-space patterns $\Ll^1_\infty, \Ll^2_\infty$. In the notation of \cite{ParkJOT1990}, $\bar \Qq$ then coincides with the algebra $\mathfrak{T}^{\alpha,\beta}$, $\bar \Pp$ coincides with the algebra $\mathcal{S}^{\alpha,\beta}$ and $\Hh_i$'s coincide with $\mathfrak{T}^\alpha$ and $\mathfrak{T}^\beta$. Furthermore, one has the short exact sequences
\begin{equation}\label{Eq:Ext5}
	\Ff_2 = C^\ast \Gg_{\Xi^\infty_\llcorner \setminus \Xi^\infty_1} \rightarrowtail \bar \Pp = C^\ast \Gg_{\Xi^\infty_\llcorner} \twoheadrightarrow \Hh_1 = C^\ast \Gg_{\Xi^\infty_1} ;
\end{equation}
\begin{equation}\label{Eq:Ext6}
	\Ff_1 = C^\ast \Gg_{\Xi^\infty_\llcorner \setminus \Xi^\infty_2} \rightarrowtail \bar \Pp = C^\ast \Gg_{\Xi^\infty_\llcorner} \twoheadrightarrow \Hh_2 =  C^\ast \Gg_{\Xi^\infty_2}.
\end{equation}

While the transversal of the quarter-space is generated by a single pattern, the same is not true for all of its closed subgroupoids. In particular, $\bar \Pp  =  C^\ast \Gg_{\Xi^\infty_\llcorner} $ is the groupoid algebra associated to the transversal $\Xi^\infty_\llcorner = \Xi_{\Ll^1_\infty} \cup \Xi_{\Ll^2_\infty}$ and thus by Corollary~\ref{Cor:Pull_backs} it is the pull-back of groupoid algebras
\begin{equation}
	\begin{tikzcd}
		\Bb & \Hh_1 \\
		\Hh_2  &  {\bar \Pp}
		\arrow[leftarrow,from=1-1, to=1-2]
		\arrow[rightarrow,from=2-1, to=1-1]
		\arrow[twoheadleftarrow ,from=2-1, to=2-2]
		\arrow[twoheadrightarrow ,from=2-2, to=1-2]
	\end{tikzcd}
\end{equation}
which is how $\bar \Pp$ was defined in \cite{ParkJOT1990}.} $\Diamond$
\end{remark}

\begin{figure}[t]
\center
\includegraphics[width=\textwidth]{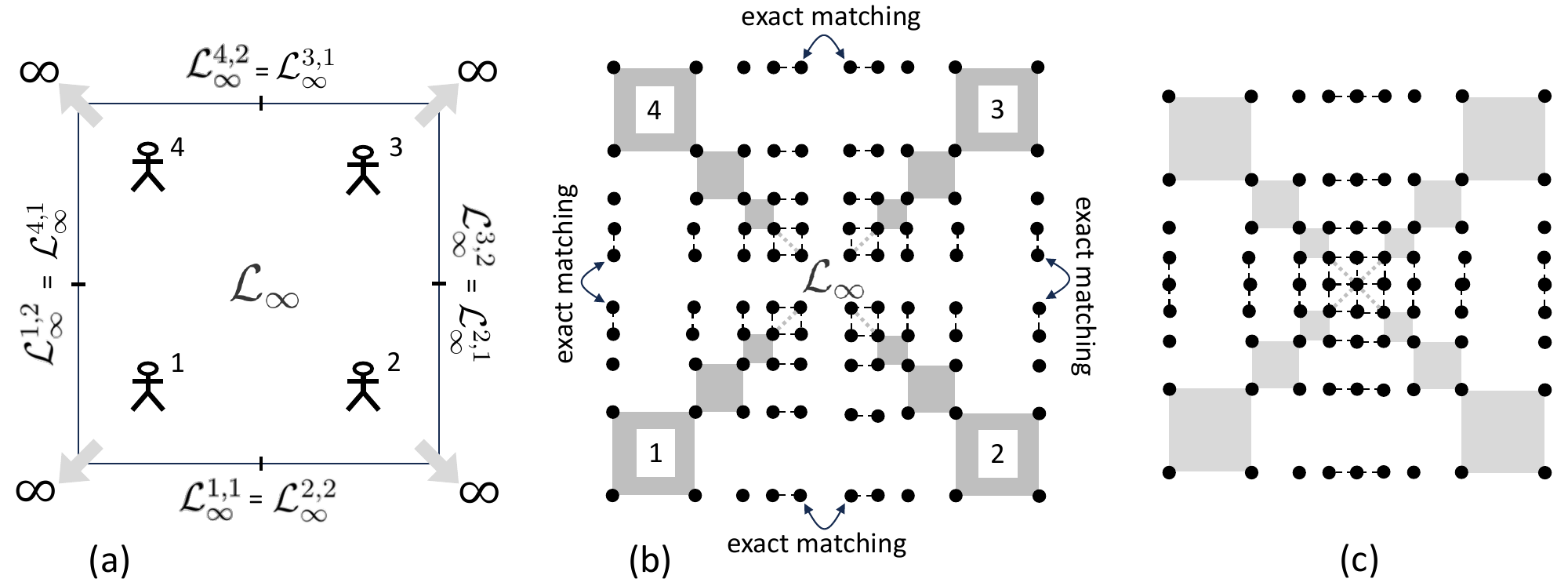}\\
  \caption{\small (a) A wire with a square cross section is growing laterally while being observed by four experimenters with a corner in their field of view; (b) The transversals $\Xi_{\Ll^\lambda}$ reported by the four observers; (c) The transversal $\Xi_\square$ of the whole wire in the infinite-size limit.
}
 \label{Fig:InfSq}
\end{figure} 

\subsection{Square geometry}\label{Sec:Wire}
We consider here a geometry with multiple corners obtained as a scaling limit of a mesh shaped like $([-L,L]^2\cap \ZM^2)\times \ZM^{d-2}$. In two dimensions this is the infinite-volume limit of a square, whereas in three dimensions it is the infinite-diameter limit of a wire with square cross-section. As illustrated in Fig.~\ref{Fig:InfSq} we consider the infinite-volume limits as seen by four observers sitting at different corners, which yields four quarter-spaces oriented in different directions.

It is easy to label them by rotations if we fix the quarter-space pattern of the previous section $\Ll^0=\NM\times \NM\times \ZM^{d-2}$  and then denote its rotations by $\frac{\pi}{2}\lambda$ as $(\Ll^\lambda)_{\lambda\in \ZM_4}$. The global transversal of the square geometry in the infinite-size limit is
\begin{equation}
\Xi_\square : = \mathsmaller{\bigcup}\nolimits_{\lambda \in \ZM_4} \, \Xi_{\Ll^\lambda},
\end{equation}
which is depicted in Fig.~\ref{Fig:InfSq}(c). There are many closed invariant subsets and we single out $\Xi_{\lambda,\#}^\infty$, $\# \in \{ 1,2,\llcorner\}$ equal to those of the previous section except for the rotation, fixing the labels as in Figure~\ref{Fig:InfSq}(a). A remarkable fact here is that the transversals $\Xi_{\Ll^\lambda}$ match along their boundaries as indicated in Fig.~\ref{Fig:InfSq}(b). The interpretation in terms of observers is consistent with this, since, for example, as observer 1 walks to the right and observer 2 walks to the left, they will eventually see the same half-space pattern, indicated as $\Ll_{\infty}^{1,1} = \Ll_{\infty}^{2,2}$ in Fig.~\ref{Fig:InfSq}(a).

\begin{figure}[t]
\center
\includegraphics[width=0.35\textwidth]{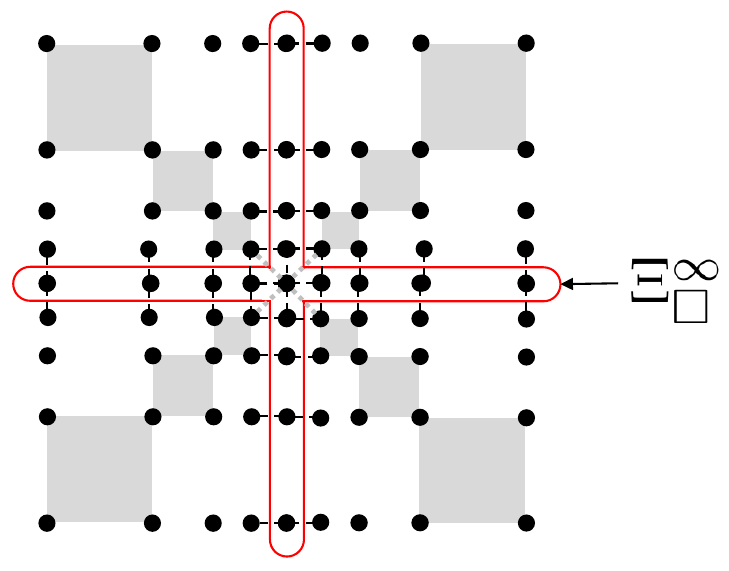}\\
  \caption{\small The largest closed and invariant proper subset of the space of units $\Xi_\square$ studied in section~\ref{Sec:Wire}. 
}
 \label{Fig:SqDeco}
\end{figure}

We denote by $\Gg_\square$ the topological groupoid for  transversal $\Xi_\square$. Note that the group of point symmetries $\Sigma\subset SO(2)$ of a finite square acts on $\Lambda$ and $\Xi_\square$ via automoprhisms. The groupoid $\Gg_\square$ has a large number of closed and invariant proper subsets and, as such, there are many filtrations which have lengths smaller or equal to 3. However, combined with the symmetry we find:

\begin{proposition} The space of units $\Xi_\square$ has a unique filtration of length 2,
\begin{equation}\label{Eq:SqFiltration}
\{\Ll_\infty\} \subset \Xi_\square^\infty \subset \Xi_\square,
\end{equation}
by closed proper subspaces that are invariant to the actions of $\Gg_\square$ and $\Sigma$. Here, $\Xi_\square^\infty =\bigcup_{\lambda\in \ZM_4} \Xi^\infty_{\lambda,\llcorner}$. In the dual picture, this supplies the equivariant cofiltration
\begin{equation}\label{Eq:Ext18}
\Qq \stackrel{\mathfrak p^2}{\twoheadrightarrow} \Pp  \stackrel{\mathfrak p^1}{\twoheadrightarrow} \Bb
\end{equation}
with $\Qq : = C^\ast \Gg_\square$, $\Pp : = C^\ast \Gg_{\Xi^\infty_\square} $ and $\Bb : = C^\ast \Gg_{\Ll_\infty}$.
\end{proposition}
$\Qq$ is the $C^\ast$-algebra encoding the experimental observations on a generalized wire with an inconceivably large square cross section. The groupoid $C^\ast$-algebra corresponding to the complement $\Xi_\square \setminus \Xi_\square^\infty$ coincides with ${\rm Ker}\, \mathfrak p^1$ and relates to the physical observations made around the corners of the sample, hence, we call it again the corner algebra. We denote it by $\Cc$ and $\Cc = \bigoplus_{\lambda \in \ZM_4} \bar \Cc_\lambda$ for four isomorphic copies of the corner algebra $\bar \Cc$ of the previous section. The groupoid $C^\ast$-algebra corresponding to the complement $\Xi_\square^\infty \setminus \{\Ll_\infty\}$ coincides with ${\rm Ker}\, \mathfrak p^0$ and will be denoted by $\Ff$. Clearly, $\Ff = \oplus_{\lambda \in \Lambda} \Ff_\lambda$, $\Ff_\lambda : = C^\ast \Gg_{\Ll_\infty^{\lambda,1}}$, and, as such, it will be called the face algebra.

\subsection{Cube geometry}\label{Sec:Crystal} We consider here a crystal of cubic shape cut out of $\ZM^3$ mesh, as shown in Fig.~\ref{Fig:CrystGeom}, and we take as the group $\Sigma$ of symmetries the full point symmetry group of the cubic lattice. The analysis may feel repetitive, but this is exactly the purpose of this exercises, to help us reveal the hierarchical structure of the space of units described in Propositions~\ref{Prop:ModelAlg} and \ref{Prop:Filtrations}.

\begin{figure}[t]
\center
\includegraphics[width=\textwidth]{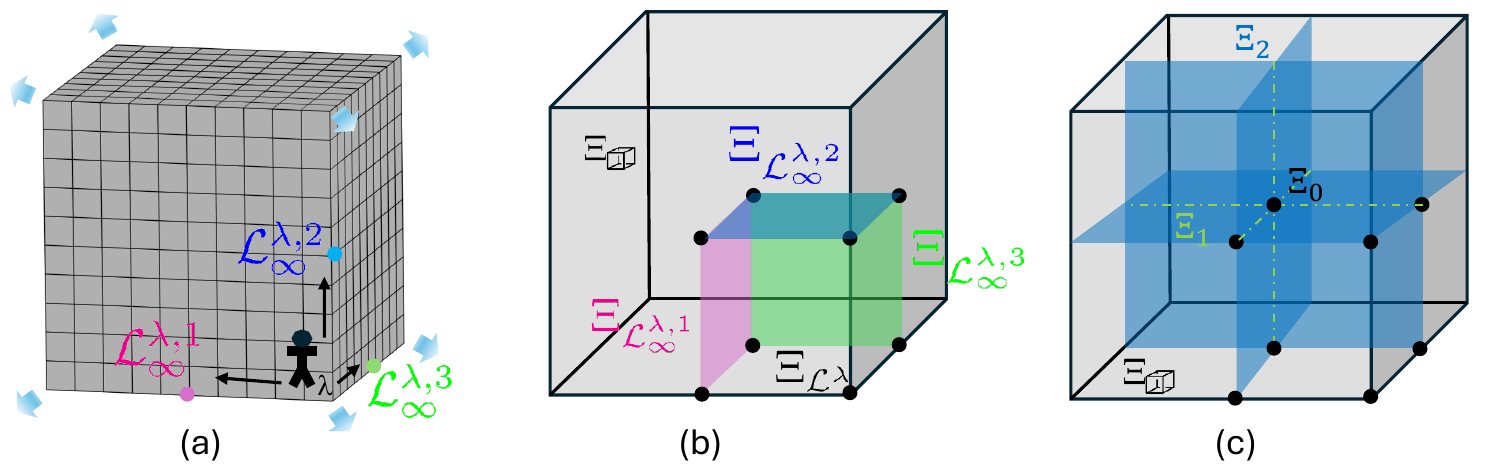}\\
  \caption{\small (a) Schematic of a growing process of a crystal with cubic symmetry. Also indicated are three asymptotic directions in which observer $\lambda$ experiences quarter geometries. (b) The transversal $\Xi_{\Ll^\lambda}$ consists of the union of the open subset $\Oo(\Ll^\lambda)$ and the transversals (shown in matching colors) of the quarter patterns experienced by the experimenter from the points indicated in panel (a). (c) The union of $\Xi_{\Ll^\lambda} \setminus \Oo(\Ll^\lambda)$ over $\lambda$ supplies the set $\Xi_2$, which is the largest closed and invariant proper subset of $\Xi_{\mbox{\small \mancube}}$. The characterization from Fig.~\ref{Fig:T2Deco} supplies the additional closed proper subset $\Xi_1$ of $\Xi_2$, which is also invariant against $\Gg_{\mbox{\small \mancube}}$ and $\Sigma$ actions.
}
 \label{Fig:CrystGeom}
\end{figure}

Following our general strategy, fixing observers at the corners of the cube while growing it in different directions we recover eight different patterns $(\Ll^\lambda)_{\lambda\in \Lambda}$ labeled by the respective corners $\Lambda=\{1,...,8\}$. The transversal of each $\Ll^\lambda$ contains three different quarter patterns $\Ll_\infty^{\lambda,i}$, labeled by the different hinges adjacent to the corner:
\begin{equation}\label{Eq:XiLl}
\Xi_{\Ll^\lambda} = \Oo(\Ll^\lambda) \cup \Xi_{\Ll_\infty^{\lambda,1}} \cup \Xi_{\Ll_\infty^{\lambda,2}} \cup \Xi_{\Ll_\infty^{\lambda,3}}
\end{equation}
Each of the quarter-space transversals $\Xi_{\Ll_\infty^{\lambda,3}}$ decomposes further as in Section~\ref{Sec:Quarter}, thus the global transversal of the cube geometry
\begin{equation}\label{Eq:CubeUnit}
\Xi_{\mbox{\small \mancube}} : = \mathsmaller{\bigcup}\nolimits_{\lambda \in \Lambda} \, \Xi_{\Ll^\lambda},
\end{equation}
contains also the transversals of the $6$ different half-spaces corresponding to faces of the cube. All of this is schematically shown in Fig.~\ref{Fig:CrystGeom}. The set $\Lambda$ has a natural action by the symmetry group of a finite cube $\Sigma\subset SO(3)$. Since the latter also acts on the transversals, $\sigma ( \Xi_{\Ll^\lambda}) = \Xi_{\Ll^{\sigma \cdot \lambda}}$, \eqref{Eq:XiLl} together with the analysis of the quarter patterns from subsection~\ref{Sec:Quarter} give the complete picture of the transversals and their invariant subsets. In particular, $\Xi_{\mbox{\small \mancube}}$ is by construction invariant under the action of the point group symmetry of the lattice.

\begin{proposition} For the assumed group of point symmetries $\Sigma$, the transversal $\Xi_{\mbox{\small \mancube}}$ has a unique filtration of length 3
\begin{equation}\label{Eq:CubeFilt}
\{\Ll_\infty\}= \Xi_0 \subset \Xi_1 \subset \Xi_2 \subset \Xi_3 =  \Xi_{\mbox{\small \mancube}},
\end{equation}
 by closed proper subsets that are invariant under the actions of the groupoid $\Gg_{\mbox{\small \mancube}}$ and the group $\Sigma$. In the dual picture, the cube geometry carries a canonical cofiltration of the algebra of physical observations
\begin{equation}\label{Eq:CoFilt3}
C^\ast \Gg_{\mbox{\small \mancube}}  \twoheadrightarrow C^\ast \Gg_{\Xi_{2}} \twoheadrightarrow C^\ast \Gg_{\Xi_1} \twoheadrightarrow C^\ast \Gg_{\Ll_\infty},
\end{equation}
which is the equivariant cofiltration listed in Proposition~\ref{Prop:Filtrations}.
\end{proposition} 
\begin{proof}
Take as $\Xi_r\subset \Xi_{\mbox{\small \mancube}}$ the set of all patterns which are invariant under translation by $(3-r)$ or more linearly independent lattice directions.
\end{proof}

\begin{remark}\label{Re:FiltChoice}{\rm The uniqueness of the filtrations presented so far spurs from the transitive character of the action of symmetry group on the set of observers. Transitivity, however, is lost for simpler symmetry groups, such as those containing just two elements, in which case there are more than one option for proper symmetry-adapted filtrations. However, the filtration by the codimension of the respective boundaries is still unique.
}$\Diamond$
\end{remark}

\subsection{Fundamental symmetries}\label{Sec:FSymm} As we shall see from the examples supplied in section~\ref{Sec:Examples}, the symmetries relevant to establish interesting higher-order bulk-boundary correspondences may involve not just crystalline symmetries but also fundamental symmetries, such as time-reversal, particle-hole or chiral symmetries. We briefly describe how they are implemented on our operator algebras by following the standard procedure devised in \cite{FreedAHP2013}. More details can be found in our Appendix~\ref{Sec:Appendix}. 

Time-reversal and particle-hole exchange are implemented as $\ZM_2$-actions and as such the finite group of crystalline symmetries $\Sigma$ is enhanced to the extended symmetry group $\bar \Sigma : = \Sigma \times \ZM_2 \times \ZM_2$. As explained in \cite{FreedAHP2013}, physics constrains time reversal and particle-hole exchange to be represented by anti-unitary operators $T$ and $P$, hence the generators of the two $\ZM_2$ subgroups must act anti-linearly on the algebras of observables. These complications are dealt with by passing to an extension of $\bar \Sigma$ by $\TM$. Up to isomorphisms, such extensions are enumerated by group morphisms $\phi : \bar \Sigma \to {\rm Out}(\TM) = {\rm Aut}(\TM) \simeq \ZM_2$ and by two-cocycles $\tau \in H^2_\phi(\bar \Sigma, \TM)$ and, as such, it is natural to denote the mentioned group extension by $\bar \Sigma^\phi_\tau$. As a set, $\bar \Sigma^\phi_\tau=\TM \times \bar \Sigma$ with multiplication determined by $\phi$ and $\tau$. The twist $\phi$ fixes which group elements will act anti-linearly in representations and the twisting cocycle $\tau$ further specifies whether the (matrix) representatives of $T$ or $P$ should square to $1$ or $-1$. In addition, the elements of $\bar \Sigma$ are graded by a homomorphism $c : \bar \Sigma \to \ZM_2$, which assigns $c(P)=1$, $c(T)=0$ and $c(\sigma)=0$ for all $\sigma\in \Sigma$, as one will eventually want particle-hole exchange to be a graded symmetry which reverses the sign of a symmetric Hamiltonian.

Now, let $C^\ast \Gg_\Xi$ be one of the groupoid $C^\ast$-algebras introduced in the previous subsection equipped with a $\Sigma$-action $\alpha$. It is extended uniquely to a complex linear $\bar \Sigma$-action $\alpha^*$ by letting $\ZM_2\times \ZM_2$ act trivially. Any element of $C^\ast \Gg_\Xi$ can be written uniquely as a continuous function in $C_0(\Gg_\Xi, \CM)$, since $\Gg_\Xi$ is \'etale and amenable. One lets $\bar \Sigma_\tau^\phi$ act on $f\in C_0(\Gg_\Xi,\CM)$ by setting
\begin{equation}\label{Eq:ExtAct}
(t,\bar \sigma) \cdot f = \begin{cases}
t \, \overline{\alpha^*_{\bar \sigma}(f)} & \text{if }\phi(\bar \sigma)=1\\
t \,\alpha^*_{\bar \sigma}(f) & \text{if }\phi(\bar \sigma)=0
\end{cases}
\end{equation}
with the complex conjugation on $\CM$, which is an $\RM$-linear automorphism of $C^\ast \Gg_\Xi$.  The morphism $\phi$ determines whether a group element $\bar \sigma$ acts linearly or anti-linearly on $C^\ast \Gg_\Xi$. 

Given a set of data $(\phi,c,\tau)$, we now specify the symmetry relations imposed upon Hamiltonians. In appendix~\ref{Sec:TwistRep}, we recall the definition of $(\phi,c,\tau)$-twisted $\bar \Sigma$-representations. The on-site degrees of freedom carried by each atom are encoded in a finite dimensional vector space $\Vv$. The algebra $ B(\Vv)$ of linear maps is equipped with a grading automorphism $C$ and inherits $\phi$-twisted $c$-graded $\bar \Sigma$-action stemming from a graded algebra morphism 
\begin{equation}
U:\RM \bar \Sigma_\tau^\phi \to B_\RM(\Vv), \quad U \circ c = C \circ U
\end{equation}
which represents $\TM$ as scalar multiplication. The Hamiltonians generating the low energy dynamics of the electrons are produced from the graded $\bar \Sigma_\tau^\phi$-$C^\ast$-algebra $B(\Vv) \otimes C^\ast \Gg_\Xi$.\footnote{Since $C^\ast \Gg_\Xi$ is trivially graded, the graded and un-graded tensor products coincide.} A specific Hamiltonian will be a self-adjoint element $h$ of this algebra which may or may not (anti-)commute with the representatives of $\bar\Sigma$, but if it does satisfy
\begin{equation}
[U_{\bar \sigma}, h] = 0, \quad \forall \bar \sigma \in \bar \Sigma_\tau^\phi,
\end{equation}
where $[\cdot,\cdot]$ is the graded commutator, then it is called $c$-twisted invariant under $\bar\Sigma$ (see Appendix~\ref{Ap:TwEq}). In most symmetry classes we will impose those relations to only hold for a specific subgroup $\Gamma\subset\bar \Sigma$.

\begin{example}\label{Ex:Chiral1}{\rm In systems with chiral symmetry, the time-reversal and particle-hole symmetries are broken, but the combination $S=TP$ persists, thus we have the symmetry group $\Gamma =\langle S\rangle \subset \overline{\Sigma}$. As a combination of two anti-unitary transformations, $S$ is unitary, hence the twists $(\phi,\tau)$ is trivial, but the grading $c(S) = 1$ is odd. The minimal example is when $B(\Vv_S) = M_2(\CM)$, which we view as being generated by Pauli matrices $\{\alpha_i\}_{i={1,2,3}}$. It is equipped with the outer grading that changes the signs of the Pauli matrices, hence the standard outer grading on the odd Clifford algebra, and also with the graded representation of the $\Gamma$ given by $S \mapsto Ad_{\alpha_3}$. The symmetric Hamiltonians anti-commute with the representative $\alpha_3$ of $S$ and therefore take the off-diagonal form
\begin{equation}\label{Eq:h38}
h = \begin{pmatrix} 0 & w^\ast \\ w & 0 \end{pmatrix} = \tfrac{1}{2}\alpha_1 \otimes (w + w^\ast) + \tfrac{\imath}{2}\alpha_2  \otimes (w - w^\ast), \quad w \in C^\ast \Gg_\Xi.
\end{equation} 
If there are additional (crystalline) symmetries then they should commute with $S$ and the blocks $w,w^*$ of a symmetric $h$ should be separately invariant.
}$\Diamond$
\end{example}

\section{Mechanism of higher-order correspondences}\label{Sec:Mechanism}

In this section, we first analyze higher-order bulk-boundary correspondences for the geometries studied in section~\ref{Sec:ModelAlg}. In the process, we formalize what it means for a model to be gapped at specific boundaries and then we reveal a pattern in the formulation of the higher-order bulk-boundary principle. Based on these findings, we provide a general picture of the principle and demonstrate how it fits into the spectral sequences induced by the cofiltrations from Proposition~\ref{Prop:Filtrations}.

Before we start, we need to fix a relation between gapped Hamiltonians and $K$-theory classes, which is essential in the field of topological condensed matter systems. For this, let $C^\ast \Gg$ be any of the groupoid $\Sigma$-$C^\ast$-algebras introduced so far.  For some subgroup $\Gamma$ of the extended symmetry group $\bar \Sigma=\Sigma \times \ZM_2\times \ZM_2$, we will want to use a $\Gamma$-equivariant K-functor together with its suspensions $(\KK_q)_{q\in \ZM}$ to assign K-theory classes to gapped Hamiltonians. To stay concrete, one may just want to fix the twisted equivariant K-functor $\KK_{q} = {}^{\phi} K^\Gamma_{0+q, c,\tau}$ or $\KK_{q} = {}^{\phi} K^\Gamma_{-1+q, c,\tau}$ with the two versions of the twisted equivariant K-functor as defined in appendix~\ref{Ap:TwEq}, and twisting data $(\phi, c, \tau)$ obtained from restriction of respective data on $\bar \Sigma$. Other similarly defined functors will also work or are already implicitly included in this generality, see Remark~\ref{rem:picturesk_theory}. Depending on the type of K-functor and precise picture of K-theory not all $\KK_q$ are naturally represented in terms of physically relevant Hamiltonians and are instead drawn from suspensions (though this can usually be remedied by going to K-theory for graded algebras where suspensions can be replaced by graded tensor products with Clifford algebras in exchange for other technical complications). However, as seen in remark~\ref{rem:picturesk_theory}, there are indeed many relevant cases where each $\KK_q$ can be represented in terms of Hamiltonians symmetric under some subgroup $\Gamma_q\subset \bar\Sigma$ that depends on $q$ with either $2$- or $8$-fold periodicity.

In any case, we focus on one particular value $q=\ast$, for which Hamiltonians can in principle define classes in $\KK_\ast$ without explicit suspensions. Throughout, a Hamiltonian will be a self-adjoint element drawn from one of the algebras $B(\Vv)\otimes C^*\Gg$ where $\Vv$ is a finite-dimensional graded vector space furnished with a twisted representation of the symmetry group $\Gamma\subset \bar \Sigma$ carrying twisting data $(\phi,c,\tau)$. We will call a Hamiltonian $h\in B(\Vv)\otimes C^*\Gg$ {\it symmetric} if it satisfies the symmetry requirements to canonically define a class in $\KK_\ast$, but possibly does not have a spectral gap. For the functor $\KK_\ast = {}^{\phi} K^\Gamma_{0, c,\tau}$ this means explicitly that $h$ is $c$-twisted invariant under $\Gamma$ for the fixed $(\phi,c,\tau)$-twisted representation on $\Vv$. If $h$ has a spectral gap then $[\gamma_h]_\ast=[\sgn(h)]_0$ with the functional calculus in $B(\Vv)\otimes C^*\Gg$ (we assume throughout that the Fermi energy is fixed at zero and thus when we speak of spectral gaps we will therefore always mean open intervals in the resolvent set which contain $0$). In a unitary picture such as $\KK_\ast =  {}^{\phi} K^\Gamma_{-1, c,\tau}$ or complex equivariant $K$-theory $\KK_\ast=K^\Gamma_{1}$ one will additionally assume that $\Vv$ has a balanced grading, i.e. $\Vv=\Ww\oplus \Ww$ and $\gamma=\one_\Ww\oplus(-\one_\Ww)$, and symmetric Hamiltonians are not only ($c$-twisted) invariant under $\Gamma$ but also odd w.r.t. the grading. Then $h\mapsto [\gamma_h]_\ast$ for invertible $h$ maps to the class defined by either of the two off-diagonal parts of $\sgn(h)$ (which of the two components one picks is a matter of convention). Comparing to \eqref{Eq:h38}, this corresponds precisely to Hamiltonians which have a chiral symmetry.

In any equivariant picture of K-theory one can, while preserving the $\KK_\ast$-theoretic class, amplify gapped Hamiltonians $h\in B(\Vv)\otimes C^*\Gg$ to ones in $B(\Vv\oplus \Ww)\otimes C^*\Gg$ with $\Ww$ any admissible finite-dimensional (graded) representation space of the symmetry group by adding to $h$ in a direct sum a gapped symmetric Hamiltonian in $B(\Ww)$ which represents the trivial class in $\KK_\ast$.
\begin{definition}\label{def: homotopic}Two gapped symmetric Hamiltonians $h_i \in B(\Vv_i) \otimes C^\ast \Gg$, $i=1,2$ are {\it stably symmetric-preserving homotopic} if there are  amplifications $\Ww_1$ and $\Ww_2$  so that $\Vv_1 \oplus \Ww_1 \simeq \Vv_2 \oplus \Ww_2 \simeq \bar \Vv$ (as graded representation spaces) and if the corresponding amplifications of $h_1$, $h_2$ in $B(\bar \Vv) \otimes C^\ast \Gg$ can be norm-continuously deformed into each other within the self-adjoint invertible symmetric operators. \end{definition}

The standing assumption (which is true for the explicitly given K-functors above) is that $[h] \mapsto [\gamma_h]_\ast \in \KK_\ast$ is a one-to-one correspondence with the stable homotopy equivalence classes of gapped symmetric Hamiltonians.

\subsection{Quarter geometry}
\label{Sec:quarter_mechanism} Using the notations from subsection~\ref{Sec:Quarter}, we will consider model Hamiltonians $h=h^\ast \in B(\Vv) \otimes \bar \Qq$, symmetric under a subgroup $\Gamma\subset{\bar \Sigma}$. For now we will leave the K-functor and symmetry group unspecified. The Hamiltonian determines a symmetric bulk Hamiltonian $h_b \in B(\Vv) \otimes \Bb$ via the composition $\bar {\mathfrak p}^1 \circ \bar {\mathfrak p}^2$ of the surjections introduced in equation~\eqref{Eq:Ext17}.  The bulk is assumed to be insulating, hence $h_b$ is assumed to have a gap in its spectrum, which we refer to as the bulk gap. We first address the question of what it means for $h$ to be gapped at the $(d-1)$-dimensional facets of the sample. As the words suggest, if one probes the quarter sample near the faces and moves farther and farther away from the corners, one should find that the Hamiltonian increasingly resembles a gapped operator. Using the physical interpretation of the left-regular representations of $\bar \Qq=C^\ast \Gg_\Ll$ (see Remark~\ref{Re:LRReps}), we can express this in precise mathematical terms by stating that $\pi_\Ss(h)$ is a spectrally gapped operator for all $\Ss \in \Xi_{\Ll_\infty^1} \cup \Xi_{\Ll_\infty^2}$, i.e. the transversals of the asymptotic half-spaces. Given the definition of $\bar \Pp$, we arrive at the following definition:

\begin{definition}\label{Def:Gappable1} A symmetric Hamiltonian $h\in B(\Vv) \otimes \bar \Qq$ is gapped at the codimension 1 boundaries (hence faces) if $\bar {\mathfrak p}^2(h) \in B(\Vv) \otimes \bar \Pp$ has a spectral gap contained inside the spectral gap of $h_b$.
\end{definition}

\begin{definition} Let $h_b=h_b^\ast\in B(\Ww)\otimes \Bb$ be a symmetric gapped bulk Hamiltonian. We say that $h_b$ is gappable at the codimension 1 boundaries if there exists a symmetric Hamiltonian $h'_b\in B(\Vv) \otimes \Bb$ which is in the same $\KK_\ast$-theoretic class $[\gamma_{h_b}]_\ast=[\gamma_{h_b'}]_\ast$ and which lifts under  $\bar{\mathfrak{p}}^1 \circ \bar{\mathfrak{p}}^2$ to a symmetric Hamiltonian $h \in B(\Vv) \otimes \bar \Qq$ that is gapped at the  codimension 1 boundaries. 
\end{definition}
When we speak of gappable Hamiltonians in the following we always assume they are symmetric unless stated otherwise. 
\begin{remark} {\rm For us, to be gappable at some boundary means precisely that K-theory does not provide an obstruction to the existence of a spectral gap at that boundary. The goal of this paper is to find all topological obstructions that are the result of the bulk K-theory class, which naturally means that we classify Hamiltonians up to stable equivariant homotopy. Accordingly, to construct a gapped lift of a bulk Hamiltonian one is by the definition above allowed to stabilize by adding as direct summands topologically trivial gapped Hamiltonians of the respective symmetry class.}
\end{remark}
The following equivalent characterization highlights that one does not need to think in terms of concrete Hamiltonians at all:

\begin{proposition}\label{Prop:Gapable1} $h_b\in B(\Vv) \otimes \Bb$ is gappable at the codimension 1 boundaries if and only if its class $[\gamma_{h_b}]_\ast \in \KK_\ast(\Bb)$ has a pre-image in $\KK_{\ast}(\bar \Pp)$ under the map $\bar {\mathfrak p}^1_\ast$. 
\end{proposition} 

\begin{proof} If $h_b$ is gappable then by assumption there exists some $h_b'\in \Bb(\Vv)\otimes \bar \Bb$ with $[\gamma_{h_b}]_\ast= [\gamma_{h_b'}]_\ast$ with a lift $h\in \Bb(\Vv)\otimes\bar \Qq$ that is gapped at the codimension 1 boundaries. Therefore, $\bar{\mathfrak{p}}^1_*([\gamma_{\bar{\mathfrak{p}}^2(h)}]_\ast)=[\gamma_{h_b'}]_\ast$ provides a pre-image of $[\gamma_{h_b}]_\ast$. Conversely, any class in $(\bar{\mathfrak{p}}^1_*)^{-1}([\gamma_{h_b}]_\ast) \subset \KK_\ast(\bar \Pp)$ can be represented by a gapped symmetric Hamiltonian $\tilde{h} \in B(\Vv)\otimes \bar \Pp$ and one can always lift that to a symmetric Hamiltonian $h\in B(\Vv)\otimes \bar \Qq$ by picking any self-adjoint lift and averaging it over a twisted representation of the subgroup  $\Gamma\subset \bar{\Sigma}$ which implements ${\bf K}_*$. One can then set $h_b' = (\bar{\mathfrak{p}}^1\circ \bar{\mathfrak{p}}^2)(h)$.
\end{proof}

The long exact sequence in K-theory induced by the equivariant epimorphism $\bar {\mathfrak p}^1 : \bar \Pp \twoheadrightarrow \Bb$,
\begin{equation}\label{Eq:KSeq01}
\KK_{\ast}(\bar \Ff) \stackrel{\bar i^1_\ast}{\longrightarrow} \KK_{\ast}(\bar \Pp) \stackrel{\bar {\mathfrak p}^1_\ast}{\longrightarrow} \KK_\ast(\Bb) \stackrel{\bar \partial^1_\ast}{\longrightarrow} \KK_{\ast-1}(\bar \Ff) \stackrel{\bar i^1_{\ast-1}}{\longrightarrow} \KK_{\ast-1}(\bar \Pp) \stackrel{\bar {\mathfrak p}^1_{\ast-1}}{\longrightarrow}  \KK_{\ast-1}(\Bb),
\end{equation}
gives an equivalent characterization, since exactness at $\KK_\ast(\Bb)$ in conjunction with Proposition~\ref{Prop:Gapable1} literally means:
\begin{corollary} Up to stable symmetry-preserving homotopies, the symmetric bulk Hamiltonians that are gappable at the codimension 1 boundaries correspond precisely to the classes in
\begin{equation}
{\rm Im}\, \bar {\mathfrak p}^1_\ast = \Ker \, \bar \partial^1_\ast \subset \KK_\ast(\Bb).
\end{equation}
\end{corollary}
We now turn our attention to the exact sequence derived from the $\bar {\mathfrak p}^2 : \bar \Qq \to \bar \Pp$ epimorphism:
\begin{equation}\label{Eq:KSeq02}
\KK_{\ast}(\bar \Cc) \stackrel{\bar i^2_\ast}{\longrightarrow} \KK_{\ast}(\bar \Qq) \stackrel{\bar {\mathfrak p}^2_\ast}{\longrightarrow} \KK_\ast(\bar \Pp) \stackrel{\bar \partial^2_\ast}{\longrightarrow} \KK_{\ast-1}(\bar \Cc) \stackrel{\bar i^2_{\ast-1}}{\longrightarrow}  \KK_{\ast -1}(\bar \Qq) \stackrel{\bar {\mathfrak p}^2_{\ast-1}}{\longrightarrow} \KK_{\ast -1}(\bar \Pp).
\end{equation}
Since it is exact at $\KK_{\ast}(\bar \Pp)$ a symmetric gapped Hamiltonian $h \in B(\Vv)\otimes \bar \Pp$ has (possibly after stabilization) a symmetric gapped lift to $\Bb(\Vv)\otimes \bar \Qq$ if and only if  $\bar \partial_*^2([\gamma_{h}]_\ast) \in \KK_{*-1}(\bar \Cc)$ is trivial. Conversely, if that class is non-trivial then any lift must be non-invertible due to in-gap corner modes which are characterized the topological invariant $\bar \partial_*^2([\gamma_{h}]_\ast) \in \KK_{*-1}(\bar \Cc)$.  Taking the analogue of Proposition~\ref{Prop:Gapable1} as a definition, a gapped bulk Hamiltonian shall be called gappable at the corner if and only if its K-theory class lifts from $\KK_\ast(\Bb)$ to $\KK_\ast(\bar \Pp)$ and then further to $\KK_\ast(\bar \Qq)$. If a bulk model is gappable at the edges, but not gappable at the corner, then this means that every realization of that model on the quarter-space which is gapped at the edges will display protected corner modes inside the bulk gap. In this case we will speak of an order-2 bulk-boundary correspondence, since the existence of some corner modes is enforced by the K-theory class of the bulk material, no matter which lift to $\Bb(\Vv)\otimes \bar \Pp$ one chooses, as long as it is gapped and symmetric. 

Those instances can be detected by a group homomorphism:

\begin{proposition}\label{Prop:BC1} There exists a well defined bulk-corner map
\begin{equation}
	\label{eq:bulk-corner-map}
	\delta^{\Bb \bar \Cc}_\ast: {\rm Im} \, \bar {\mathfrak p}^1_\ast =\Ker \, \bar \partial^1_\ast \subseteq \KK_\ast(\Bb) \to    \frac{\mathrm{Im} \, \bar \partial^2_\ast}{\bar \partial^2_\ast ({\rm Ker}\, \bar {\mathfrak p}^1_\ast) }=\frac{	\mathrm{Im} \, \bar \partial^2_\ast}{\mathrm{Im} \, (\bar \partial^2_\ast \circ \bar i^1_\ast)}  \subseteq \frac{\KK_{\ast-1}(\bar \Cc)}{\bar \partial^2_\ast ({\rm Ker}\, \bar {\mathfrak p}^1_\ast) },
\end{equation}
 such that $x\in \Ker \, \bar \partial^1_\ast$ has a pre-image in $\KK_\ast(\bar \Qq)$ if and only if $\delta^{\Bb \bar \Cc}_\ast(x)=0$. 
\end{proposition}

\begin{proof} By definition, a class $x\in {\rm Im} \, \bar {\mathfrak p}^1_\ast \subset \KK_\ast(\Bb)$ has a pre-image $\tilde{x}\in \KK_\ast(\bar \Pp)$, but it need not be unique.  However, any two pre-images $\tilde{x}, \tilde{x}'$ differ only by an element $\tilde{x}-\tilde{x}'\in \Ker\,\bar{\mathfrak{p}}^1 = i^1_\ast(\KK_\ast(\bar \Ff))$, with the latter following from equality being the exactness of \eqref{Eq:KSeq01}. Therefore, $\bar\partial_\ast^2(\tilde{x})$ and $\bar\partial_\ast^2(\tilde{x}')$ differ only by an element of $\bar\partial_\ast^2({\rm Ker}\, \bar {\mathfrak p}^1_\ast)$, which makes $$\delta^{\Bb \bar \Cc}_\ast:x\mapsto \bar\partial_\ast^2(\tilde{x}) + \bar\partial_\ast^2({\rm Ker}\, \bar {\mathfrak p}^1_\ast)$$ for arbitrary choice of lift $\tilde{x}$ a well-defined homomorphism of abelian groups. One can lift $x$ to $\KK_\ast(\bar \Qq)$ if and only if there is a pre-image $\tilde{x}\in \KK_\ast(\bar \Pp)$ with $\bar\partial_\ast^2(\tilde{x})=0$, hence if and only if $\delta^{\Bb\bar \Cc}_\ast(x)=0$.
%
\end{proof}
In light of the above discussion this allows us to identify the K-theory classes of symmetric bulk Hamiltonians that exhibit non-trivial order-2 bulk-boundary correspondences. The quotient by $\bar \partial^2_\ast({\rm Ker}\, \bar {\mathfrak p}^1_\ast)$ has a simple interpretation: A surface insulator, i.e. a gapped Hamiltonian over $\bar \Pp$ which is in the kernel of $\bar{\mathfrak{p}}^1$, can under the boundary map $\partial^2_\ast$ still lead to a non-trivial class in $\KK_{\ast-1}(\bar \Cc)$. The image $\bar \partial^2_\ast({\rm Ker}\, \bar {\mathfrak p}^1_\ast)\subset \KK_{\ast-1}(\bar \Cc)$ characterizes precisely the possible protected corner modes which can be carried by such surface layers. This ambiguity is always present in the choice of lift, since one can always add a surface layer without changing the bulk. The crucial observation is that due to  the long exact sequences, this group already enumerates {\it all} ambiguities on the level of K-theory and hence after taking the quotient the result depends only on the bulk K-theory class. 

\begin{remark} {\rm Generally it would be difficult to directly enumerate the classes in $\KK_\ast(\Bb)$ that lift to $\KK_\ast(\bar \Pp)$ but not to  $\KK_\ast(\bar \Qq)$, since the latter two groups are in practice not known explicitly. That is precisely the point of rewriting all components of \eqref{eq:bulk-corner-map} in terms of more computable expressions using exactness. For example, it is usually quite feasible to compute $\KK_\ast(\bar \Ff)$ and then $\mathrm{Im} \, (\bar \partial^2_\ast \circ \bar i^1_\ast)$. The domain and range of $\delta^{\Bb \bar \Cc}_\ast$ can then be determined by computing the boundary maps $\bar \partial_*^1$ and $\bar\partial_*^2$ for sufficiently many Hamiltonians.} $\Diamond$
\end{remark}

\begin{example}
\label{Ex:Chiral2}
{\rm In the remainder of this section we study an example of a non-trivial second-order bulk-boundary correspondence on the quarterspace in two dimensions $d=2$. The symmetry group $\Sigma =\{e,\sigma_m\} \simeq \ZM_2$ shall act by the diagonal mirror $(x_1,x_2)\in \ZM^2\mapsto (x_2,x_1)$. We do not involve anti-unitary symmetries and therefore use the ordinary $\ZM_2$-equivariant $K$-theories, i.e. the K-functor above becomes $\KK_q=K_q^{\ZM_2}$, which is to be distinguished from the non-equivariant complex K-functor $K_q$ which also plays a role in the following. Classes in $K^{\ZM_2}_0$ are defined by symmetric matrix-valued projections, in particular the spectral projections of gapped symmetric Hamiltonians, and classes in $K^{\ZM_2}_1$ by symmetric unitary matrices, in particular the polar decompositions of the off-diagonal parts $w$ of oddly graded gapped symmetric Hamiltonians as in \eqref{Eq:h38}.

For $S_1, \ S_2$ the unitary generators of the group $C^*$-algebra $\Bb=C^\ast \ZM^2$, the $\ZM_2$-action $\sigma_m$ interchanges $S_1$ and $S_2$. We fix the directions of the shifts by imposing that projecting them to $\ell^2(\NM\times\NM)$ makes them isometries. Since $\bar{\mathcal{C}}\simeq \mathbb{K}(\NM\times\NM)$, the compact operators on $\ell^2(\NM\times\NM)$, we have $ {K}_1^{\ZM_2}(\bar{\mathcal{C}})=0$ and
  \begin{equation}
    K_{0}^{\ZM_2}(\bar{\mathcal{C}})= K_0(C^\ast \ZM_2) \otimes K_0(\bar{\mathcal{C}})=\ZM[\chi_+ \otimes E_0]_0\oplus \ZM[\chi_- \otimes E_0]_0
\end{equation}
with representatives in $\Bb(\Vv)\otimes \bar{\Cc}$. Here $\Vv=\CM e \oplus \CM \sigma_m$ is a two-dimensional vector space, $\chi_\pm = \frac{1}{2}(e \pm \sigma_m)$ are  the central projections supporting the trivial/non-trivial irreducible representations of $\ZM_2$, and $E_0 = |\delta_{0}\otimes \delta_{0}\rangle \langle \delta_0 \otimes \delta_{0}|$ is the rank-one projection onto the site located exactly at the corner. The equivariant $K$-theories of the face algebra are also straightforward. Indeed, we have $\bar \Ff = \Ff_1 \oplus \Ff_2$, $\Ff_i \simeq \KM(\NM) \otimes C^\ast \ZM$, and $\ZM_2$ acts by an isomorphism between $\Ff_1$ and $\Ff_2$. Then Proposition~\ref{prop:reduction_equivariant} gives
\begin{equation}
\begin{split}
K_0^{\ZM_2}(\bar \Ff)&=\ZM[\chi_+ \otimes(P_1 \oplus P_2)]_0,\\ K_1^{\ZM_2}(\bar\Ff)&=\ZM[\chi_+ \otimes \left( (S_1 P_2+ P_2^\perp) \oplus (S_2 P_1+ P_1^\perp)\right)]_1
\end{split}
\end{equation}
  where $P_i$ are the images of $|\delta_0\rangle \langle \delta_0| \otimes 1 \in \KM(\NM) \otimes C^\ast \ZM$ through the stated isomorphisms. Since $K_1^{\ZM_2}(\bar \Cc)$ is trivial, one can only possibly observe a non-trivial bulk-corner correspondence if one considers $\ast =1$ in \ref{Prop:BC1}. This calls for a chiral symmetry on the level of self-adjoint Hamiltonians.

\begin{proposition} We have
\begin{equation}
K^{\ZM_2}_0(\mathcal{B}) = \ZM[\chi_+ \otimes 1_\Bb]_0 \oplus \ZM[\chi_- \otimes 1_\Bb]_0, \quad K^{\ZM_2}_1(\mathcal{B}) = \ZM[u_\Ff]_1 \oplus \ZM[u_\Cc]_1,
\end{equation} 
where
\begin{equation}\label{Eq:K1BGen}
u_{\mathcal{F}}=\chi_+ \otimes S_1S_2+\chi_- \otimes 1_\Bb,\quad u_\mathcal{C}=\frac{1}{2}\begin{pmatrix}
        -S_1^\ast-S_2^\ast & S_1^\ast-S_2^\ast\\
        S_1-S_2 & S_1+S_2
    \end{pmatrix},
\end{equation}
with the latter written as an element of $M_2(\CM) \otimes \Bb$ with the $\ZM_2$ action on $\CM^2$ given by $M  = {\scriptsize \begin{pmatrix} 1 & 0 \\ 0 & -1 \end{pmatrix}}$, such that $(Ad_M \otimes \sigma_m)(u_\Cc) = u_\Cc$.
\end{proposition}

\begin{proof} The equivariant $K^\ast_{\ZM_2}$-groups of the 2-torus can be easily computed from the symmetry-adapted CW-filtration $X_0 \subset X_1 \subset \TM^2$, defined as follows: If $\TM^2$ is seen as a square with identified opposite edges, then $X_0$ is any of the corners and  $X_1$ is the union of left and down edges, as well as the mirror-invariant diagonal. An elementary computation with the Atiyah-Hirzebruch spectral sequence shows that $K^{\ZM_2}_0(\Bb)\simeq K_0(C^*\ZM_2)$ and $K^{\ZM_2}_1(\Bb)\simeq \ZM^2$, with the two generators distinguished by the following properties: As unitary functions on the torus, the first, $u_\Ff$, has winding number $1$ on both edges of the square and the second, $u_\Cc$, becomes a diagonal matrix when restricted to the diagonal of the square and there has winding numbers $1$ and $-1$ respectively in the two eigenspaces of $M$. \end{proof}

\begin{proposition} The domain of $\delta^{\Bb \bar \Cc}_1$ is ${\rm Im} \, \bar {\mathfrak p}^1_1 = \ZM[u_\Cc]_1$.
\end{proposition}

\begin{proof} As a self-adjoint invertible element with chiral symmetry, $u_\Cc$ is represented by 
\begin{equation}
h_b = \begin{pmatrix} 0 & u_\Cc \\ u_\Cc^\ast & 0 \end{pmatrix}\in B(\Vv_S) \otimes M_2(\CM) \otimes \Bb,
\end{equation}
where $B(\Vv_S)$ and its structure are described in Example~\ref{Ex:Chiral1}. By separating the generators $S_i$, we have the decomposition $h_b = h_1(S_1) + h_2(S_2)$. A useful feature to notice is that $h_1(X)$ anti-commutes with $h_2(Y)$ whenever $X$ and $X^\ast$ both commute with $Y$ and $Y^\ast$, but not necessarily with themselves. Now, the pair $\hat h=(h_1(\hat S_1) + h_2(S_2), h_1(S_1) + h_2(\hat S_2))$ supplies a symmetric lift of $h_b$ to $M(\CM_4) \otimes \bar \Pp$, where the hat indicates the standard Toeplitz extension and $\bar \Pp$ is viewed as the pullback explained in Remark~\ref{Re:Park}. Since $h_1(S_1)^2 = h_2(S_2)^2 = \frac{1}{2}$, we have 
\begin{equation}
\hat h^2 = (h_1(\hat S_1)^2 + h_2(S_2)^2, h_1(S_1)^2 + h_2(\hat S_2)^2)=\tfrac{1}{2} \, 1_{\bar \Pp} + (h_1(\hat S_1)^2,  h_2(\hat S_2)^2),
\end{equation} 
hence $\hat h$ is invertible. As such, $\hat h$ supplies a class in $K_1^{\ZM_2}(\bar \Pp)$, which proves that $[u_\Cc]_1$ belongs to the range of $\bar {\mathfrak p}^1_1$. The other generator $u_\Ff$ has non-zero weak odd Chern numbers in both directions, hence it leads to non-trivial classes under $\bar \partial_1^1$ (see Section~\ref{Sec:NoEqui} below for the more general case).
\end{proof}

\begin{proposition} The image of $\delta^{\Bb \bar \Cc}_1$ is isomorphic to $\ZM_2$ and the generator $[u_\Cc]_1$ of the domain of  $\delta^{\Bb \bar \Cc}_1$ is mapped into the generator $\ZM_2$.
\end{proposition}

\begin{proof} Our first task is to evaluate $\bar \partial_1^2( [\hat u_c]_1)$, where $\hat u_c$ is the unitary operator in $M_2(\CM) \otimes \bar \Pp$ obtained from the spectral flattening of $\hat h$. A concrete expression of the connecting map for the non-equivariant case was given in \cite[Prop.~4.3.2]{ProdanSpringer2016}, the class in $K_0(\bar \Cc)$ of $\bar \partial_1^2( [\hat u_c]_1)$ is represented by the projection
\begin{equation}\label{Eq:Index1}
\bar P = e^{-\imath \frac{\pi}{2} \varphi(\bar h) } {\rm diag}(1_{M_2(\CM) \otimes \bar \Qq},0) e^{\imath \frac{\pi}{2} \varphi(\bar h)} \in M_4(\CM) \otimes \bar \Cc,
\end{equation}
where $\bar h$ is any chirally symmetric lift of $\hat h$ to $M_4(\CM) \otimes \bar \Qq$ and $\varphi: \RM \to [-1,1]$ a continuous non-decreasing odd function with variation inside the spectral gap of $\hat h$. This expression will also represent the element in the equivariant $K_0^{\ZM_2}(\bar \Cc)$, if we choose a $\ZM_2$-invariant lift $\bar h=h_1(\bar S_1) + h_2(\bar S_2)$, where $\bar S_1, \bar S_2$ are shift operators on $\ell^2(\NM\times \NM)$. Furthermore, we are in the conditions of Proposition~4.3.3 in \cite{ProdanSpringer2016}, hence  the projection in \eqref{Eq:Index1} is equivalent to $J \bar P_0 + {\rm diag}(0_{2},1_{{M_2(\CM)}})$, where $\bar P_0$ is the spectral projection of $\bar h$ onto $\{0\}$ and $J=1_{M_2(\CM)} \oplus (-1_{M_2(\CM)})$ is the operator implementing chiral symmetry. Using the anti-commuting structure mentioned above, it is easily verifiable that the kernel of $\bar h$ has dimension one and that the corresponding projector coincides with $\chi_+ \otimes {\rm diag}(E_0,0)$. Following a similar procedure for the generator of $K_1^{\ZM_2}(\bar \Ff)$, one finds 
\begin{equation}
\bar \partial_1^2 \Big([\chi_+ \otimes \left( (S_1 P_2+ P_2^\perp) \oplus (S_2 P_1+ P_1^\perp)\right)]_1 \Big)= -2[\chi_+ \otimes E_0]_0,
\end{equation} 
and the statement follows.
\end{proof}

These calculations demonstrate that our framework not only enable us to identify a non-trivial bulk model $h_b$, but also to conclude that, for a quarter geometry and diagonal mirror symmetry, there are no other (topologically distinct) models that can produce bulk-boundary correspondences of order-2.
}$\Diamond$
\end{example}
\begin{remark}{\rm
Had we used in the previous example complex non-equivariant $K$-theory, we would have found that the natural map $K_1(\bar \Ff)\to K_1(\bar \Pp)$ is an isomorphism, which readily implies $\delta^{\Bb \bar\Cc}_1=0$. This absence of second-order bulk-boundary correspondence means that without the crystalline symmetry fixing any particular $K$-theory class in the bulk never constrains the corner states in any way. Analyzing the commutative diagram
\begin{equation}
	\begin{tikzcd}
		K_{1}^{\ZM_2}(\bar\Pp) \arrow[d] \arrow[r,"{\bar \partial}^2_1"] & K_{0}^{\ZM_2}(\bar\Cc) \arrow[d] \\ 
		K_{1}(\bar\Pp) \arrow[r,"{\bar \partial}^2_1"] & K_{0}(\bar \Cc)
	\end{tikzcd}
\end{equation}
tells us, however, that if we start with a class in $K_{1}^{\ZM_2}(\bar\Pp)$ which has the non-trivial corner mode parity we will have a protected corner mode even if the symmetry is broken at the corner. In other scenarios, however, a symmetry-protected corner mode can be trivial in non-equivariant $K$-theory and then one can remove it by breaking the symmetry at the corner.}$\Diamond$
\end{remark}

\subsection{Square geometry} Once a filtration of the unit space is fixed, all statements from the previous subsection remain valid for the wire geometry, and they can be formulated in exactly the same form but with the bar removed from above the symbols and the $K$-theoretic functor properly adjusted. Examples of non-trivial higher-order bulk-boundary correspondences for this geometry will be supplied in section~\ref{Sec:Examples}.

\subsection{Cube geometry} 
\label{Sec:cube_mechanism}
This geometry displays faces, hinges and corners, which we refer to as boundaries of codimension 1, 2 and 3, respectively.  As such, the crystal geometry can support non-trivial correspondences of order 1, 2 and 3. An order-3 bulk boundary correspondence is carried by the corners of the cube and one will want to use a filtration of appropriate length:

\begin{proposition}\label{Pro:UFiltration} There exists a unique filtration 
	\begin{equation}\label{Eq:CubeFilt13}
		\Xi_0 \subset \Xi_1 \subset \Xi_2 \subset \Xi_3,
	\end{equation} 
starting at $\Xi_0=\{\Ll_\infty\}$ and ending at $\Xi_3 =\mathsmaller{\bigcup}_{\lambda \in \Lambda} \ \Xi_{\Ll^\lambda}$
	such that each $C^*\Gg_{\Xi_r \setminus \Xi_{r-1}}$, $1\leq r \leq 3$, consists out of those observables localized to the boundaries of codimension $r$. 
\end{proposition}
\begin{proof}
One can assign to each pattern in $\Xi_3$ uniquely a codimension, which is $2$ for all quarter-space patterns and $1$ for all half-space patterns. One needs to take for $\Xi_r$ precisely the set of all patterns with codimension $r$ or less, as characterized by their number of translation-invariant lattice direction, to achieve the mentioned localization property for $C^*\Gg_{\Xi_r \setminus \Xi_{r-1}}$ (see the comment below Definition~\ref{def:ideal_abbreviation}).
\end{proof}

The symmetry group $\Sigma$ of a finite cube acts on the groupoids $\Gg_{\Xi_r}$ and we again fix a subgroup $\Gamma\subset \bar \Sigma$ with twisting data $(\phi, c, \tau)$ together with some twisted $\Gamma$-equivariant K-functor $\KK_\ast$. The filtration \eqref{Eq:CubeFilt13} is by construction term-wise $\Gamma$-invariant and a symmetric model Hamiltonian $h \in B(\Vv) \otimes C^\ast \Gg_{\Xi_3}$ can therefore again be projected to a symmetric bulk Hamiltonian $h_b \in B(\Vv) \otimes C^\ast \Gg_{\Ll_\infty}$ using the maps in the equivariant cofiltration
\begin{equation}\label{Eq:CoFilt33}
	C^\ast \Gg_{\Xi_3}  \stackrel{\mathfrak p^3}{\twoheadrightarrow} C^\ast \Gg_{\Xi_{2}} \stackrel{\mathfrak p^2}{\twoheadrightarrow}  C^\ast \Gg_{\Xi_1} \stackrel{\mathfrak p^1}{\twoheadrightarrow}  C^\ast \Gg_{\Ll_\infty}.
\end{equation}

To observe protected topological face, hinge or corner modes, the bulk Hamiltonian $h$ should in the first place be an insulator, i.e. have a spectral gap. To classify protected corner modes one should further not already have face or hinge modes inside the bulk gap. Given the definition and physical interpretation of the left regular representations, this is the same as saying that $\pi_\Ss(h)$ has a spectral gap for all $\Ss \in \Xi_1$ respectively $\Ss \in \Xi_2$. These arguments leads to:

\begin{definition}\label{Def:Gappable2} The symmetric Hamiltonian $h \in B(\Vv) \otimes C^\ast \Gg_{\Xi_3}$ is gapped at the boundaries of codimension $r$ if the projection $(\mathfrak p^{r+1}\circ ... \circ \mathfrak p^{3}) (h)$ in $B(\Vv) \otimes C^\ast \Gg_{\Xi_{r}}$ has a spectral gap.  
	
A symmetric gapped bulk Hamiltonian $h_b \in B(\Ww) \otimes C^\ast \Gg_{\Ll_\infty}$ is said to be gappable at the boundaries of codimension $r$ if there is a symmetric Hamiltonian $h_b'\in B(\Vv) \otimes C^\ast \Gg_{\Ll_\infty}$ with the same bulk K-theory class $[\gamma_{h_b}]_\ast=[\gamma_{h_b'}]_\ast\in \KK_\ast(C^\ast \Gg_{\Ll_\infty})$ which admits a symmetric lift $h\in B(\Vv) \otimes C^\ast \Gg_{\Xi_3}$ that is gapped at the boundaries of codimension $r$.
\end{definition}

\begin{remark}\label{Re:ResolventSet}{\rm The spectrum only becomes smaller under morphisms, hence a Hamiltonian which is gapped or gappable at the codimension $r$ boundaries is also gapped or gappable at all boundaries with smaller codimension.
}$\Diamond$
\end{remark}
\begin{proposition}\label{Prop:Gapable2} The symmetric gapped bulk Hamiltonian $h_b$ is gappable at the boundaries of codimension $r$ if and only if the class $[\gamma_{h_b}]_\ast \in \KK_\ast(\Bb)$ can be lifted to an element of $\KK_\ast(C^\ast \Gg_{\Xi_{r}})$.
\end{proposition}

\begin{proof} The argument is practically identical with the one for Proposition~\ref{Prop:Gapable1}.
\end{proof}

\begin{corollary} Up to stable symmetry-preserving deformations, the symmetric bulk Hamiltonians that are gappable at the boundaries of codimension 1 and 2 are respectively listed by
\begin{equation}
{\rm Im} \, \mathfrak p^1_\ast \subseteq \KK_\ast(\Bb).
\end{equation}
and 
\begin{equation}
{\rm Im} \, \mathfrak p^1_\ast \circ \mathfrak p^2_\ast \subseteq \KK_\ast(\Bb).
\end{equation}
\end{corollary}


The ideal $\Cc={\rm Ker}\, \mathfrak p^3$ supplies the algebra of observations around the corners and the connecting map $\partial^3_\ast$ induced by $\mathfrak p^3$ sends classes in $\KK_\ast(C^\ast \Gg_{\Xi_2})$ of symmetric lifts of bulk Hamiltonians to classes in $\KK_{\ast-1}(\Cc)$. A symmetric bulk Hamiltonian that is gappable at the codimension-2 boundaries (hinges) can be (up to stable equivalence) lifted to $C^\ast \Gg_{\Xi_2}$ under $\mathfrak{p}^1\circ \mathfrak{p}^2$, but such a lift to a gapped symmetric Hamiltonian from $C^\ast \Gg_{\Xi_3}$ under $\mathfrak{p}^1\circ \mathfrak{p}^2 \circ \mathfrak{p}^3$ will exist if and only if there is a pre-image of $[\gamma_{h_b}]_\ast$ in $\KK_\ast(C^\ast \Gg_{\Xi_2})$ that is mapped by $\partial^3_\ast$ to the trivial class of $\KK_{\ast-1}(\Cc)$, equivalently if and only if there exists a choice of symmetric boundary condition for which there are no protected corner modes. If conversely there are always protected corner modes, then we say that we have a non-trivial order-3 bulk-boundary correspondence. The ambiguity in the choice of lift to $\KK_\ast(C^*\Gg_{\Xi_2})$ is due to the long exact sequence of K-theory enumerated by $\Ker(\mathfrak p^1_\ast \circ \mathfrak p^2_\ast) \subset \KK_\ast(C^*\Gg_{\Xi_2})$ therefore these instances can be detected using a group homomorphism like in Proposition~\ref{Prop:BC1}: 

\begin{proposition}\label{Prop:BC2} There exists a well defined bulk-corner map
\begin{equation}
	\label{Eq:BC2}
	\delta^{\Bb \Cc}_\ast: {\rm Im} \, (\mathfrak p^1_\ast \circ \mathfrak p^2_\ast) \subseteq \KK_\ast(\Bb) \to \frac{\mathrm{Im} \, \partial^3_\ast}{\partial^3_\ast({\rm Ker} \, \mathfrak p^1_\ast \circ \mathfrak p^2_\ast)} \subseteq \frac{\KK_{\ast-1}(\Cc)}{\partial^3_\ast({\rm Ker} \, \mathfrak p^1_\ast \circ \mathfrak p^2_\ast)},
\end{equation}
obtained by computing $\partial^3_\ast$ for any lift from $ {\rm Im} \, \mathfrak p^1_\ast \circ \mathfrak p^2_\ast$ to $\KK_\ast(C^*\Gg_{\Xi_2})$. One has $\delta^{\Bb \Cc}_\ast(x)=0$ if and only if $x$ has a pre-image in $\KK_\ast(\Qq)$.
\end{proposition} 
\begin{remark}{\rm
The classes $x\in {\rm Im} \, (\mathfrak p^1_\ast \circ \mathfrak p^2_\ast)$ with $\delta^{\Bb \Cc}_\ast(x)\neq 0$ are precisely the classes of symmetric bulk Hamiltonians that generate non-trivial order-3 bulk-boundary correspondences supported by $\Xi_3$, since by definition one can then never eliminate all their corner states by a change of symmetry-preserving boundary conditions (precisely those ambiguities are divided out). Any symmetric lift of such a bulk Hamiltonian to the cube which is gapped at the faces and hinges must then display corner modes in the bulk gap protected by a non-trivial class in $\KK_{*-1}(\Cc)$. However, the latter may depend on the specific lift (i.e. boundary condition) up to an element of ${\partial^3_\ast({\rm Ker} \, \mathfrak p^1_\ast \circ \mathfrak p^2_\ast)}\subset\KK_{*-1}(\Cc)$, which enumerates again the corner modes that can be carried by surface layers that are trivial in the bulk.}
$\Diamond$	
\end{remark}

It happens frequently that a bulk Hamiltonian will  display bulk-boundary correspondences of mixed order, for example, there can be protected surface modes at some faces while other faces can still be gapped and then there may be additional second-order hinge modes or third-order corner modes at the hinges respectively corners that do not border the gapless faces. We will see examples of this in section \ref{Sec:CubeExamples}. To resolve all phenomena of mixed higher bulk-boundary correspondence one must potentially investigate {\it all} possible filtrations adapted to the chosen symmetry group $\Gamma$, not just \eqref{Eq:CubeFilt13}.  Generally, to detect a mixed order-$r$ bulk-boundary correspondence induced by a specific class in $\KK_\ast(\Bb)$ one should adapt the filtration to exclude so many boundaries as to make it no longer be subject to bulk-boundary correspondences of order $(r-1)$ and lower, starting with removing all faces that potentially carry first-order boundary states. There are some subtleties involved; a change of boundary condition can move the higher boundary modes between different parts of the boundary and one can further have both too many or too few boundaries left to stabilize higher-bulk boundary correspondences, which means that one may need to sample several filtrations until one finds the correct combinations of gappable boundaries that stabilize interesting higher-order boundary modes. Since the relevant boundary states should still be localized to faces, hinges or corners one can restrict oneself to filtrations of the form
\begin{equation}\label{eq:filtrations_choice}
(\tilde{\Xi} \cap \Xi_0) \subset (\tilde{\Xi} \cap \Xi_1)\subset (\tilde{\Xi} \cap \Xi_2) \subset (\tilde{\Xi} \cap \Xi_3)
\end{equation}
enumerated by closed invariant transversals $\tilde{\Xi}\subset \Xi_{\mbox{\small \mancube}}$ that select parts of the cube geometry. This choice ensures that the boundary ideals $C^*\Gg_{(\tilde{\Xi} \cap \Xi_r)\setminus (\tilde{\Xi} \cap \Xi_{r-1})}$ still localize precisely to the boundaries of codimension $r$ contained in $\tilde{\Xi}$. A particular example of this form is the transversal $\Xi_\square$ of section~\ref{Sec:Wire} in three dimensions, which can be seen as a closed invariant subset of $\Xi_{\mbox{\small \mancube}}$ that excludes all corners and all but four of the hinges and faces each, thereby isolating mixed second-order bulk-boundary correspondences which localize to the selected hinges. Other than using a different transversal respectively filtration, the mathematical formalism remains unchanged.


\subsection{A unifying picture}\label{Sec:Unifying} We now consider a crystal geometry in $d$ dimensions displaying boundaries of codimension between $1$ and $d$ and with a global transversal $\Xi$ (which may be only a subset of all possible patterns constructed in the scaling limit of an actual crystal such as the wire geometry of section~\ref{Sec:Wire}). 

In accordance with the standing assumptions of this section, we fix again a finite group of the form $\Gamma \subset \bar \Sigma=\Sigma \times \ZM_2 \times \ZM_2$ together with a twist $(\phi,c,\tau)$ such that the twisted $\Gamma$-equivariant K-functor $\KK_\ast$ classifies the corresponding stable homotopy classes of gapped symmetric Hamiltonians.  We assume that $\Gamma$ acts on $C^*\Gg_\Xi$ via (anti-)linear automorphisms and on the unit space $\Xi$ via homeomorphisms. We assume that we have a filtration 
\begin{equation}
\label{Eq:SatFiltGen}
\{\Ll_\infty\} = \Xi_0 \subset  \ldots \subset \Xi_{d-1} \subset \Xi_d = \Xi,
\end{equation} 
by $\Gamma$-symmetric closed invariant subsets, which mathematically defines which patterns $\Xi_r \subset \Xi$ we consider to have only boundaries of codimension $r$ or less. Then the algebra of physical observables $C^\ast \Gg_{\Xi}$ has a $\Gamma$-equivariant cofiltration
\begin{equation}\label{Eq:CoFiltD}
 C^\ast \Gg_{\Xi_{d}} \stackrel{\mathfrak p^{d}}{\twoheadrightarrow} C^\ast \Gg_{\Xi_{d-1}}\stackrel{\mathfrak p^{d-1}}{\twoheadrightarrow} \cdots \stackrel{\mathfrak p^{2}}{\twoheadrightarrow} C^\ast \Gg_{\Xi_1} \stackrel{\mathfrak p^{1}}{\twoheadrightarrow} C^\ast \Gg_{\Xi_0}= C^\ast \Gg_{\Ll_\infty},
\end{equation}
and each ${\rm Ker}\, \mathfrak p^{r} = C^\ast \Gg_{\Xi_{r} \setminus \Xi_{r-1}}$ has a sensible interpretation as an algebra of observations near the boundaries of codimension precisely $r$. 

Regarding the correspondence itself, the arguments are very similar to the ones from the previous subsections: 
\begin{definition}
The gapped symmetric Hamiltonian $h_b\in B(\Vv)\otimes C^*\Gg_{\Xi_0}$ is called gappable at the codimension $r$ boundaries if and only if $[\gamma_{h_b}]_\ast\in \KK_\ast(C^*\Gg_{\Xi_0})$ admits a pre-image under $s^r: = \mathfrak p^1 \circ \cdots \circ \mathfrak p^r: C^*\Gg_{\Xi_{r}}\to C^*\Gg_{\Xi_{0}}$ in $\KK_\ast(C^*\Gg_{\Xi_{r}})$. We also say that $h_b$ exhibits an order-$r$ bulk boundary correspondence if $h_b$ is gappable at the codimension $r-1$ boundaries but not at the codimension $r$ boundaries.
\end{definition}
Let $\partial^{r}_\ast: \KK_\ast(C^*\Gg_{\Xi_{r-1}})\to \KK_{*-1}(C^*\Gg_{\Xi_{r}\setminus \Xi_{r-1}})$ be the connecting map induced by the epimorphism $\mathfrak p^r$. Any particular lift to $\KK_\ast(C^\ast \Gg_{\Xi_{r-1}})$ can be further lifted to $\KK_\ast(C^\ast \Gg_{\Xi_{r}})$ if and only if it is in the kernel of $\partial^{r}_\ast$. If the image under $\partial^{r}_\ast$ is conversely non-trivial for all lifts of a fixed class in $\KK_\ast(C^\ast\Gg_{\Xi_0})$ then we have found an order-$r$ bulk-boundary correspondence. The ambiguities for the lifts from $\KK_\ast(C^\ast\Gg_{\Xi_0})$ to $\KK_\ast(C^\ast \Gg_{\Xi_{r-1}})$, induced by choosing different symmetric boundary conditions, are enumerated by ${\rm Ker} \, (s^{r-1}_*)\subset \KK_\ast(C^*\Gg_{\Xi_{r-1}})$. Thus, after these are quotiented out, we obtain an enumeration of the order-$r$ correspondences supported by the boundaries selected by $\Xi_r$:

\begin{proposition} The stable equivariant homotopy classes of symmetric bulk Hamiltonians which are gappable at the boundaries of codimensions less than $r$ are listed by ${\rm Im}\, s^{r-1}_\ast \subset \KK_\ast(\Bb)$. One obtains a well-defined homomorphism
\begin{equation}
	\label{Eq:BCr}
	\delta^r_\ast: {\rm Im}\, s^{r-1}_\ast \subseteq \KK_\ast(C^*\Gg_{\Xi_0}) \to \frac{\mathrm{Im} \, \partial^{r}_\ast}{\partial^{r}_\ast({\rm Ker} \, s^{r-1}_\ast)}
\end{equation}
by computing the connecting map for any lift from $\KK_\ast(C^*\Gg_{\Xi_0})$ to $\KK_\ast(C^*\Gg_{\Xi_{r-1}})$ and then taking the stated quotient to make it independent of the lift. 

If a bulk Hamiltonian $h_b$ is gappable at the codimension $r-1$ boundaries then the range of $\partial^r_\ast$ evaluated on lifts of $[\gamma_{h_b}]_\ast$ to $\KK_\ast(C^*\Gg_{\Xi_{r-1}})$ is precisely given by the coset $\delta^r_\ast([\gamma_{h_b}]_\ast)$ as a subset of $\KK_{\ast-1}(C^*\Gg_{\Xi_{r}\setminus\Xi_{r-1}})$. 

In particular, if $\delta^r_\ast([\gamma_{h_b}]_\ast)$ evaluates to a non-trivial value (i.e. a coset that does not contain the neutral element) then $h_b$ is not gappable at the codimension $r$ boundaries and therefore exhibits an order-$r$ bulk-boundary correspondence. 
\end{proposition} 
In this way we can enumerate both the Hamiltonians which exhibit order-$r$ bulk-boundary correspondences as well as their possible boundary states at the order $r$ boundaries under the assumption that the order $r-1$ boundaries are gapped.

\begin{remark}{\rm 
The higher boundary maps $(\delta^s_\ast)_{s\leq r}$ together identify the bulk classes in $\KK_\ast(C^*\Gg_{\Xi_0})$ which cannot be lifted to $\KK_\ast(C^*\Gg_{\Xi_r})$ and assign an order to the obstruction, which is the lowest order boundary which must be un-gapped. For bulk models which are already {\bf not} gappable at the codimension $r-1$ boundaries one does not obtain any information about possible codimension $r$ boundary states which may coexist with the lower order boundary states. As sketched in the previous section, it is largely a modeling problem to choose the filtrations (respectively geometries) in such a way that the constructed boundary maps help one to identify and enumerate those phenomena of higher-order bulk-boundary correspondence one is actually interested in.}
\end{remark}

We now prepare to place the higher-order bulk-boundary maps $\delta^r_\ast$ in their appropriate framework: 

\begin{theorem}[\cite{SavinienBellissard}]\label{th:specsequence} For a cofiltration of $C^\ast$-algebras,
\begin{equation}
\Aa_d \twoheadrightarrow \Aa_{d-1}\twoheadrightarrow ... \twoheadrightarrow \Aa_0,
\end{equation} 
there exists a spectral sequence $(E^r_{p,q},d^r_{p,q})$ converging to $\KK(\Aa_d)$.
\end{theorem}

Although standard (see e.g. \cite{McCleary}), we need to describe this spectral sequence in some detail. First of all one extends the cofiltration to $(\Aa_n)_{n\in \ZM}$ by setting $\Aa_n=\Aa_d$ for $n\geq d$ and $\Aa_n=0$ for $n<0$. The first page of the spectral sequence is then the bigraded complex
\begin{equation}
E^1=\bigoplus_{p,q} E^1_{p,q}, \quad E^1_{p,q}: =\KK_{-p+q}(\Ee_p), \quad \Ee_p = \Ker(\Aa_p\twoheadrightarrow \Aa_{p-1}),
\end{equation}
and one considers the auxiliary bigraded complex $D^1=\bigoplus_{p,q} D^1_{p,q}$ with $D^1_{p,q} : =\KK_{-p+q}(\Aa_p)$. They form an exact couple, i.e. a diagram
\begin{equation}
\begin{tikzcd}
	D^1 \arrow[r, "\alpha"] & D^1 \arrow[dl, "\beta", shift left=1.5ex] \\
	E^1 \arrow[u, "\gamma"]
\end{tikzcd}
\end{equation}
where every map has as its kernel the image of the previous one. The morphism $\alpha$ maps each  $D^1_{p,q}$ to $D^1_{p-1,q-1}$ and is given by  $\KK_{-p+q}(\Aa_{p}) \to \KK_{-p+q}(\Aa_{p-1})$, induced from $\Aa_p \twoheadrightarrow \Aa_{p-1}$. The morphism $\beta$ maps each $D^1_{p,q}$ to $E^1_{p+1,q}$ and is given by the connecting map induced by $\Aa_{p+1}\twoheadrightarrow \Aa_p$. The morphism $\gamma$ maps each $E^1_{p,q}$ to $D^1_{p,q}$ and is given by $\KK_{-p+q}(\Ee_{p})\rightarrow \KK_{-p+q}(\Aa_{p})$. This first page is just a reformulation of the long exact sequence of $K$-theory.

To any exact couple, one can canonically associate a new derived exact couple $(E^2,D^2)$ and iteration results in a spectral sequence $E^r_{p,q}$ with differentials 
\begin{equation}
d^r_{p,q}: E^r_{p,q}\to E^r_{p+r,q+r-1},
\end{equation}
defined in terms of combinations of $\alpha$, $\beta$ and $\gamma$. Concretely (e.g. \cite[Proposition 2.8]{McCleary})
\begin{equation}
	\label{eq:er_page}E^r =  \frac{\gamma^{-1}\alpha^{\circ(r-1)}(D^1)}{\beta( \Ker \, \alpha^{\circ(r-1)})},
\end{equation}
which can be identified with the homology $H(E^{r-1},d^{r-1})$ w.r.t. the differentials 
\begin{equation}
d^r = \beta (\alpha^{-1})^{\circ(r-1)}\gamma,
\end{equation}
where the inverse of $\alpha$ is well-defined since it does not depend on the choice of lift. The construction is iterative and uses that $D^r=\alpha^{r-1}(D^1)$ forms another exact couple with $E^r$ and $\alpha^{(r)}=\alpha$, $\beta{(r)}=\beta (\alpha^{-1})^{\circ(r-1)}$ and $\gamma^{(r)}=\gamma$.

\begin{remark}\label{rem:schochetSS} {\rm There is a similar spectral sequence due to Schochet \cite{Schochet} also constructed from an exact couple: Let $\Ii_{-1}= \Ii_0 \subset \Ii_1 \subset ... \subset \Ii_d = \Aa$ be a filtration of a $C^*$-algebra $\Aa$ by closed ideals. Then there is an exact couple with
\begin{equation}		
\bar E^1_{p,q}= \KK_{p+q}(\Ii_p/\Ii_{p-1}), \quad \bar D^1_{p,q}= \KK_{p+q}(\Ii_p).
\end{equation}
One can write any filtration of a $C^*$-algebra equivalently as a cofiltration and then both approaches yield equivalent spectral sequences that converge to $\KK_\ast(\Aa)$ \cite{SavinienBellissard}, hence the higher-order bulk-boundary maps can be constructed using either version.
	}$\Diamond$
\end{remark}

The following statement provides the connection between the boundary maps constructed above with the spectral sequence and therefore completes the proof of our main Theorem~\ref{Th:Main}:

\begin{proposition}
	\label{prop:higher_boundary_maps}
	Let $(E^r_{p,q},d^r_{p,q})$ be the spectral sequence corresponding to cofiltration~\eqref{Eq:CoFiltD}. Then the corestriction of the $r$-th differential
\begin{equation}
d^r_{0,q}: \Ker(d^{r-1}_{0,q}) \to \mathrm{Im}(d^r_{0,q})
\end{equation}
for $q=\ast$ coincides with the corestriction of $\delta^r_\ast: \Ker(\delta^{r-1}_\ast) \to \mathrm{Im}(\delta^r_{\ast})$.
\end{proposition}

\begin{proof} One can write \eqref{eq:er_page} as
\begin{equation}
		E^r_{p,q} = \frac{\gamma^{-1}\mathrm{Im}(\KK_{-p+q}(\Aa_{p+r-1})\to \KK_{-p+q}(\Aa_{p}))}{\beta \Ker(\KK_{-p+q+1}(\Aa_{p-1})\to \KK_{-p+q+1}(\Aa_{p-r}))},
\end{equation}
which, for $p=0$, specializes to
\begin{equation}
E^r_{0,q} = \mathrm{Im}(\KK_{q}(\Aa_{r-1})\to \KK_q(\Aa_{0})) = {\rm Im} \, s_\ast^{r-1}.
\end{equation}
The codomain of $d^r_{0,q}$ is $E^r_{r,q+r-1}$, hence one takes the quotient by
$$\beta \Ker(\KK_{q}(\Aa_{p-1})\to \KK_{q}(\Aa_{p-r}))= \partial^{r}_q(\Ker s_*^{r-1}).$$
%
Since $d_{0,\ast}^r$ and $\delta_\ast^r$ have the same domain, are both defined in the same way (by choosing a lift from $\KK_\ast(\Aa_0)$ to $\KK_{\ast}(\Aa_{p-1})$, computing the connecting map to $\KK_{\ast-1}(\Ee_p)$ and then taking the quotient by the same group) they corestrict to the same map.
\end{proof}
One main advantage of using the spectral sequence is that it breaks the computation into smaller pieces; the iterative construction shows that the subquotients can be expanded as follows:
\begin{corollary}
The codomain of $d^r_{0,\ast}$ is equivalently given by
$$\frac{\partial^{r}_\ast ((s_*^{r-1})^{-1}\KK_\ast(\Aa_0))}{\partial^{r}_\ast({\rm Ker} \, s^{r-1}_\ast)}=\frac{\partial^{r}_q ((s_*^{r-1})^{-1}\KK_\ast(\Aa_0))}{\mathrm{Im}(d^{r-1}_{1,1+\ast})+\mathrm{Im}(d^{r-2}_{2,2+\ast})...+\mathrm{Im}(d^{1}_{r-1,r-1+\ast})}.$$
\end{corollary}
The image $\mathrm{Im}(d^{s}_{r-s,r-s+\ast})$ enumerates the possible boundary states at the codimension $r$ boundaries which are protected by surface topological insulators at the codimension $(r-s)$-boundaries through an order-$s$ correspondence. Hence, to compute the codomain of $d^r_{0,\ast}$ one should first classify those lower-dimensional higher-order topological insulators of order less than $r$. This can be done through an iterative process, in fact it happens automatically if one computes one page after another of the spectral sequence. Those groups then enumerate exhaustively the possible dependence of the boundary states of any fixed bulk Hamiltonian on the choice of boundary condition, which turns out to be in K-theory completely equivalent to investigating the effects of decorating the surface with additional layers.

The framework of spectral sequences supplies other powerful tools. For example, we can recognize when two higher-order bulk-boundary correspondences are identical from a topological point of view: 

\begin{proposition}
\label{prop:homomorphism_spectral_sequence}
Let there be two cofiltrations 
\begin{equation}
...\twoheadrightarrow \Aa_n \twoheadrightarrow \Aa_{n-1}\twoheadrightarrow ... \twoheadrightarrow \Aa_0 \twoheadrightarrow 0, \quad \tilde{\Aa}_n \twoheadrightarrow \tilde{\Aa}_{n-1}\twoheadrightarrow ... \twoheadrightarrow \tilde{\Aa}_0 \to 0
\end{equation}
	and homomorphisms $\psi_p: \KK_q(\Aa_p)\to \KK_q(\tilde{\Aa}_p)$, $\varphi_n: \KK_q(\Ee_p)\to \KK_q(\tilde{\Ee}_p)$ such that the diagram
\begin{equation}
\begin{tikzcd}[column sep=small]
		\arrow[r, ""]	& \KK_{q}(\Aa_p) \arrow[d, "\psi_p"] \arrow[r, ""] & 	 \KK_{q}(\Aa_{p-1}) \arrow[d, "\psi_{p-1}"] \arrow[r, "\partial"] & \KK_{q-1}(\Ee_p)  \arrow[d, "\varphi_p"] \arrow[r, ""] & \KK_{q-1}(\Aa_{p-1}) \arrow[d, "\psi_{p-1}"] \arrow[r, ""] & \ \\
		 \arrow[r, ""] & \KK_{q}(\tilde{\Aa}_p) \arrow[r, ""] & 	\KK_{q}(\tilde{\Aa}_{p-1})  \arrow[r, "\partial"] & \KK_{q-1}(\tilde{\Ee}_p)  \arrow[r, ""] & \KK_{q-1}(\tilde{\Aa}_{p-1}) \arrow[r, ""] & \ 
	\end{tikzcd}
\end{equation} 
commutes, then the higher order boundary maps satisfy $\tilde{d}^p_{0,q} \circ \varphi_p = \varphi_p \circ d^p_{0,q}$. If all homomorphisms $\varphi_p$ are isomorphisms then the induced spectral sequences are term-wise isomorphic.
\end{proposition}
\begin{proof}
For a commutative diagram of exact couples (in unraveled form)
\begin{equation}
\begin{tikzcd} 
	\ \arrow[r, "\gamma"]	& D^1 \arrow[d, "\psi"] \arrow[r, "\alpha"] & 	D^1 \arrow[d, "\psi"] \arrow[r, "\beta"] & E^1  \arrow[d, "\varphi"] \arrow[r, "\gamma"] & D^1 \arrow[d, "\psi"] \arrow[r, "\alpha"] & \ \\
	\ \arrow[r, "\tilde{\gamma}"] & \tilde{D}^1 \arrow[r, "\tilde{\alpha}"] & 	\tilde{D}^1  \arrow[r, "\tilde{\beta}"] & \tilde{E}^1  \arrow[r, "\tilde{\gamma}"] & \tilde{D}^1 \arrow[r, "\tilde{\alpha}"] & \
\end{tikzcd}
\end{equation}
one naturally obtains induced homomorphisms relating the derived pages $D^r$, $E^r$ with $\tilde{D}^r$, $\tilde{E}^r$, again in a commutative diagram of the same shape \cite[I.5]{Massey}. For the last statement, note that if the maps $\varphi_p: \KK_q(\Ee_p)\to \KK_q(\tilde{\Ee}_p)$ are isomorphisms, then the maps $\psi_p: \KK_q(\Aa_p)\to \KK_q(\tilde{\Aa}_p)$ are also isomorphisms since $\Aa_0=\Ee_0$ and one can iteratively apply the five lemma from this base case.
\end{proof}

\begin{corollary}\label{Cor:IsoSpec} Let there be a commutative diagram of equivariant morphisms
\begin{equation}
	\begin{tikzcd} 
		C^*\Gg_{\Xi_d} \arrow[d, "\psi_d"] \arrow[r] & C^*\Gg_{\Xi_{d-1}}  \arrow[d, "\psi_{d-1}"] \arrow[r] & ...  \arrow[d] \arrow[r] & C^*\Gg_{\Xi_0}  \arrow[d, "\psi_0"]\\
		C^*\Gg_{\tilde{\Xi}_d} \arrow[r] & C^*\Gg_{\tilde{\Xi}_{d-1}}  \arrow[r] & ...   \arrow[r] & C^*\Gg_{\tilde{\Xi}_0} 
	\end{tikzcd}
\end{equation}
between cofiltrations corresponding to different transversals. If all of the homomorphisms $\psi_n$ induce isomorphisms in $K$-theory, then the induced spectral sequences are term-wise isomorphic.
\end{corollary} 

\begin{example}{\rm Consider for example a situation where two lattices have the same symmetry group but different unit cells (e.g. one is a decoration of the other). Then Corollary~\ref{Cor:IsoSpec} can show that the two systems support the same higher-order bulk-boundary correspondences.
}$\Diamond$
\end{example}

\begin{remark}{\rm Corollary~\ref{Cor:IsoSpec} can more generally be used to obtain partial (complete) information about the bulk-boundary maps by relating to simpler (equivalent) groupoids for which it may be more feasible to compute the boundary maps (see subsetion~\ref{Sec:Adiabatic}).
}$\Diamond$
\end{remark}


\begin{remark}
{\rm For any polyhedron representing a $d$-dimensional crystal, one can construct a groupoid algebra $\Gg_{\Xi}$ as in Section~\ref{Sec:ModelAlg} by gluing together patterns which model its corners. A natural filtration is then given by grouping together boundaries of the same linear dimension. Such a filtration makes the groupoid solvable in the sense that there is an equivariant isomorphism (or at least Morita equivalence)
\begin{equation}
\label{eq:building_blocks}
\Ee_r = C^*(\Gg_{\Xi_r\setminus \Xi_{r-1}})\simeq \bigoplus_{m=1}^{N_r} C(\TM^{d-r})\otimes \KM(\ell^2(\Ll_{r,m}))\end{equation}
with compact operators on some discrete $r$-dimensional pattern. Each direct summand represents one disjoint piece of a boundary, e.g. a single facet, hinge or corner. The equivalence arises since one can choose a large enough unit cell compatible with the boundaries and then Fourier transform the remaining translation-invariant directions. Any element of the symmetry group acts either as an automorphism of one or an isomorphism between two of the direct summands. Using Proposition~\ref{prop:reduction_equivariant} one can compute the $K$-theory of $\Ee_r$ by choosing orbit representatives and then computing the $H_{r,m}$-equivariant $K$-theories of each summand $C(\TM^{d-r})\otimes \KM(\ell^2(\Ll_{r,m}))$ where $H_{r,m}$ is its stabilizer group. In many cases the individual pieces are just groups $K^{H_{r,m}}_q(C(\TM^{d-r}))$, i.e. (possibly twisted) equivariant $K$-groups describing $d-r$-dimensional topological crystalline insulators with a point group appropriate for $d-r$ dimensions. In conclusion, the first page of the spectral sequence is usually feasible to compute and consists of finite groups, which is very important since all subsequent pages are then subquotients that admit concrete expressions. However, the higher derivations unfortunately elude systematic computation; spectral sequences are bookkeeping tools to systematically compile information about the connecting maps but those are difficult to compute using their definition alone. In practice one needs to construct a sufficiently large (but finite) collection of Hamiltonians for which one can exactly determine on which boundaries their gaps close and what the associated $K$-theoretic invariants are. We will revisit this issue when we compute examples in Section~\ref{Sec:Adiabatic}.} $\Diamond$
\end{remark}

\begin{remark}
{\rm While we consider only crystals in this work, quasi-crystals can possibly be described using the literally same formalism, simply by using different patterns $\Ll^{\lambda}$ to make the global transversal. This is of interest in particular since those allow point symmetries that are not possible in crystalline systems, such as five-fold rotational symmetry. Likewise, continuous models for topological phases can be realized by similar groupoid algebras. As the formalism of this section applies to general $C^\ast$-algebras, the only thing that changes is that the basic building blocks in the direct sum decompositions \eqref{eq:building_blocks} will be different algebras.} $\Diamond$
\end{remark}

\subsection{Comparison with ordinary bulk boundary correspondence}
In our formalism we have a cofiltration
$$C^*\Gg_{\Xi_d}\to C^*\Gg_{\Xi_{d-1}}\to ... \to C^*\Gg_{\Xi_0}$$
and  we say that a gapped bulk Hamiltonian exhibits higher-bulk boundary correspondence of order $r$ or lower if and only if its class in $\KK_\ast(C^*\Gg_{\Xi_0})$ cannot be lifted to $\KK_\ast(C^*\Gg_{\Xi_r})$. In $K$-theory this obstruction is precisely the boundary map of the exact sequence
$$0 \to C^*\Gg_{\Xi_r\setminus \Xi_0} \to C^*\Gg_{\Xi_r} \to C^*\Gg_{\Xi_0} \to 0,$$ 
thus, it seems that an alternative approach to classification would be to compute the associated boundary map 
\begin{equation}
	\label{eq:totalboundarymap}
	\mathring{\partial}^r_\ast: \KK_\ast(C^*\Gg_{\Xi_0})\to \KK_{\ast-1}(C^*\Gg_{\Xi_r\setminus\Xi_0})
\end{equation}
with all the higher-order boundary states corresponding to non-trivial classes in $\KK_{\ast-1}(C^*\Gg_{\Xi_r\setminus\Xi_0})$. From this point of view the higher-order bulk boundary correspondence is almost the same thing as an ordinary bulk-boundary correspondence. The crucial difference is that in higher-order bulk boundary correspondence we impose an additional gap condition at the codimension $(r-1)$-boundary which allows us to localize the obstruction to a subquotient of $\KK_{\ast-1}(C^*\Gg_{\Xi_r\setminus\Xi_{r-1}})$. Let us sketch how the boundary map above relates to our higher boundary maps and why the latter are more practical:
\begin{enumerate}
	\item[i)]  Without first understanding the K-groups of $C^*\Gg_{\Xi_r\setminus\Xi_0}$ one cannot feasibly compute the boundary map since one will be unable to pinpoint classes in $\KK_{\ast-1}(C^*\Gg_{\Xi_r\setminus\Xi_0})$. Unfortunately, this can be complicated. We know from examples that many different classes in $\KK_{\ast-1}(C^*\Gg_{\Xi_r\setminus\Xi_{r-1}})$ are related by a change of boundary conditions, which means that they must be one and the same class when included in $\KK_{\ast-1}(C^*\Gg_{\Xi_r\setminus\Xi_0})$. Furthermore, the relations to the boundary states $\KK_{\ast-1}(C^*\Gg_{\Xi_p\setminus\Xi_{p-1}})$ for $p\leq r$ are encoded in a spectral sequence which converges to the groups $\KK(C^*\Gg_{\Xi_r\setminus\Xi_0})$, namely the one associated to the cofiltration
	$$C^*\Gg_{\Xi_r\setminus\Xi_0}\to C^*\Gg_{\Xi_{r-1}\setminus\Xi_0} \to ...\to C^*\Gg_{\Xi_1\setminus\Xi_0} \to 0.$$
	The quotients in this cofiltration are the same ideals $\Ee_p=C^*\Gg_{\Xi_p\setminus\Xi_{p-1}}$ as before, except for the quotient at $p=0$ which is now trivial, and therefore the first page is almost the same. Thus our spectral sequence approach can also be used to compute $\KK_{\ast-1}(C^*\Gg_{\Xi_r\setminus\Xi_0})$ up to a finite number of group extension problems involving subquotients of $\KK_{\ast-1}(C^*\Gg_{\Xi_r\setminus\Xi_{r-1}})$.
	\item[ii)] Our higher-order boundary maps give us partial information: By construction, $\mathring{\partial}^r_\ast$ maps a bulk class in $\KK_\ast(C^*\Gg_{\Xi_0})$ to a non-trivial class in $\KK_{\ast-1}(C^*\Gg_{\Xi_r\setminus\Xi_0})$ if and only if one of the boundary maps $\delta^p_\ast$ of order $p\leq r$ maps it to a non-trivial value.
	\item[iii)] Imposing the gap condition at the codimension $(r-1)$-boundaries narrows us down to classes in $\KK_{\ast-1}(C^*\Gg_{\Xi_r\setminus\Xi_0})$ which are in the image of the natural map $\KK_{\ast-1}(C^*\Gg_{\Xi_r\setminus\Xi_{r-1}})\to \KK_{\ast-1}(C^*\Gg_{\Xi_r\setminus\Xi_0})$. Recalling that $\delta^r_\ast$ takes values in a subquotient of $\KK_{\ast-1}(C^*\Gg_{\Xi_r\setminus\Xi_{r-1}})$, every representative of the same coset in $\mathrm{Im}(\delta^r_\ast)$ gives rise to one and the same class in $\KK_{\ast-1}(C^*\Gg_{\Xi_r\setminus\Xi_{0}})$.
\end{enumerate}
The conclusion is that the boundary map \eqref{eq:totalboundarymap} is intimately related to the higher-order bulk-boundary correspondences and the tools required to understand it are essentially the same as we have developed to map the higher-order bulk-boundary correspondences. While $\mathring{\partial}^r_\ast$ is certainly non-trivial in the cases of interest, in our setup of higher-order bulk boundary correspondence where the codimension $(r-1)$ boundaries are assumed to be gapped it cannot provide any additional information over the boundary maps $\delta^p_\ast$, which are at this point both easier to compute and more directly allow us to derive the possible manifestations of the boundary states. The classification by \eqref{eq:totalboundarymap} would, however, be the relevant one if one imposes {\it only} a bulk gap assumption.
\begin{example}
{\rm In the setting of example~\ref{Ex:Chiral2}, $\mathring{\partial}^2_1$ corresponds to a boundary map of the exact sequence $0\to \Ker(\bar{\Qq}\to \Bb) \to \bar{\Qq}\to \Bb \to 0$. The computations given there and the spectral sequence mentioned above yield $K_0^{\ZM_2}(\Ker(\bar{\Qq}\to \Bb))\simeq \ZM \oplus \ZM_2$ as the unique solution to a group extension problem. Under $\mathring{\partial}^2_1$ the bulk classes $[u_\Ff]_1$ and $[u_\Cc]_1$ map to the first respectively second of those generators. This gives us information on possible experimental signatures of the corner states which may remain if the edges are not gapped. In contrast, for non-equivariant K-theory one finds that naturally $K_0(\Ker(\bar{\Qq}\to \Bb)))\simeq K_0(\bar \Ff)$ which shows again that the bulk K-theory can only enforce first-order boundary states at the edges but not corner states.}$\Diamond$
\end{example}

\section{Examples of Higher-Order Correspondences}\label{Sec:Examples}
In this section, we compute the higher-order bulk-boundary maps for additional symmetries and the generalized wire geometry of Section~\ref{Sec:Wire}. 

We will use the infinite square $C^*\Gg_\square$ algebra with $d-2$ infinite translation-invariant direction. In three dimensions our crystal model an infinitely long wire with infinite cross-section. From Section~\ref{Sec:Wire}, we have a cofiltration of $C^\ast$-algebras \eqref{Eq:Ext18}
\begin{equation}
\Qq  \stackrel{\mathfrak p^2}{\twoheadrightarrow} \Pp  \stackrel{\mathfrak p^1}{\twoheadrightarrow} \Bb \twoheadrightarrow 0
\end{equation}
with $\Bb= C^*\ZM^d \simeq C(\TM^{d})$ and kernels
\begin{equation}\label{Eq:FC}
\begin{aligned}
\Ff&=\Ker(\Pp\twoheadrightarrow \Bb) \simeq \bigoplus\nolimits_{\lambda\in \ZM_4} \Ff_\lambda, \quad
&\Ff_\lambda &\simeq C^*\ZM^{d-1} \otimes \KM(\ell^2(\NM))\\
\Cc&=\Ker(\Qq \twoheadrightarrow \Pp) \simeq \bigoplus\nolimits_{\lambda\in \ZM_4} \Cc_\lambda, \quad
&\Cc_\lambda &\simeq C^*\ZM^{d-2}\otimes \KM(\ell^2(\NM))
\end{aligned}
\end{equation}
As before, the equivariant $K$-functors will be specified by $\KK_\ast$, while the non-equivariant ones by $K_\ast$.

First, we demonstrate that there are no higher-order topological insulators in the non-equivariant case. We then introduce tools developed in our recent work~\cite{OSP2024} to compute the image of the second-order boundary maps $\delta^2_\ast$ for three interesting symmetries, namely, inversion, $C_2T$ and $C_4T$ symmetries. This fully classifies all possible topological hinge currents of three-dimensional topological insulators with those symmetries. On the bulk side, the group ${\rm Dom} \, \delta^2_\ast/ {\rm Ker\, \delta^2_\ast} \simeq {\rm Im}\, \delta^2_\ast$ divides the bulk gappable Hamiltonians in distinct classes, and all Hamiltonians from one such class generate identical boundary effects (up to stabilization). We provide representative models for each of these classes, which completes the classification of order-2 topological insulators for the square geometry and mentioned symmetries. A discussion of the cube geometry is supplied at the end of the section.



\subsection{Triviality of the non-equivariant case}\label{Sec:NoEqui}
The algebras $\Bb$, $\Ff_\lambda$, $\Cc_\lambda$ are Morita-equivalent to algebras of continuous functions on tori.
While it is well-known that
\begin{equation}
K_\ast(C^*\ZM^n)=K_\ast(\TM^n)\simeq  \ZM^{2^{n-1}}, \quad *=0,1
\end{equation}
choosing a consistent labeling of the $K$-group elements requires extra care, because of the different orientations of the facets and corners. Nevertheless, this can be achieved as follows. Since the groups are torsion-free, we can label all generators using the numerical pairings with cyclic cohomology. It is known that the cyclic cohomology of $C(\TM^n)$ is spanned by the so-called Chern cocycles, which we describe now. For $\Aa\in \{\Bb, \Ff_\lambda, \Cc_\lambda\}$, one has a natural (densely defined) trace $\Tt_\Aa$ induced from the Haar measure on $C(\TM^n)$ and the trace on the compact operators. There is an $\RM^d$-action on $C^\ast \Gg_\square$ acting via
\begin{equation}
(\Theta_x f)(g, S) = e^{2\pi \imath (x\cdot g)} f(g,S), \quad x\in \RM^d, \ g \in S, \ S\in \Xi_{\square}
\end{equation}
on elements of the convolution algebra $C_c\Gg_{\Xi_{\square}}$. Since all of the relevant groupoid algebras $\Qq$, $\Pp$ and $\Bb$ as well as their ideals $\Ff$, $\Cc$ come from invariant subgroupoids of $\Gg_\square$ this also defines $\RM^d$-actions on them in such a way that all natural homomorphisms between them are equivariant. This makes sure that spatial directions between the different algebras are labeled consistently. For any tuple $v=(v_1,...,v_n)$ made up of unit vectors in $\RM^d$ and $(n+1)$-tuple $(f_0,\ldots,f_n)$ made up of elements of  a suitable dense subalgebra of $\Aa$, one can define the Chern cocycle
\begin{equation}
\Ch_{\Aa, v}(f_0,...,f_n) 
\;=\; 
c_{\Aa,n} \,
\sum_{\rho \in S_n} (-1)^\rho\, \Tt_\Aa\big(f_0 \nabla_{v_{\rho(1)}} f_1 \ldots  \nabla_{v_{\rho(n)}} f_n\big)
\;,
\end{equation}
with some normalization constants and the densely defined derivation $\nabla_v$ in the direction $v$ of the action $\Theta$. Using a picture of $K$-theory where $K_0$ is represented by projections and $K_1$ by unitaries, one has well-defined pairings between $K_{q}(\Aa)=K_{q\,\mathrm{mod}\,2}(\Aa)$ and the cyclic cocycles of parity $n\,\mathrm{mod}\,2 = q\,\mathrm{mod}\, 2$ via
\begin{equation}
\langle [x]_q, [\Ch_{\Aa,v}]\rangle = \begin{cases}
	\Ch_{\Aa,v}(x,...,x) & q\text{ even}\\
	\Ch_{\Aa,v}(x^{-1},x,x^{-1},...,x,x^{-1}) & q \text{ odd}
\end{cases}
\end{equation}
We assume the choice of normalization constants is made as in \cite{ProdanSpringer2016} such that all non-trivial pairings will exactly have $\ZM$ as their range whenever $v$ made up of unit vectors (see \cite[Section 5.7]{ProdanSpringer2016}). We will now identify tuples of unit vectors $v_I = (e_{i_1},...,e_{i_n})$, $i_1<i_2<...<i_n$ with subsets $I=\{i_1,...,i_n\}$ of elements in $\{1,...,d\}$ and write $\Ch_{\Aa,I}=\Ch_{\Aa,v_I}$.  Let us then say a class $x_{\Aa,I}\in K_{q}(\Aa)$ is dual to $\Ch_{\Aa,I}$ if 
\begin{equation}
\langle [x_{\Aa,I}]_q, [\Ch_{\Aa,J}]\rangle =
	\delta_{I,J}
\end{equation}
for all subsets $I,J\subset \{1,...,d\}$  whose length has the same parity as $q$. For the torus such dual classes exist and are unique, hence we can label (assuming $d\geq 2$)
\begin{enumerate}
	\item[(i)] A basis of $K_{q}(\Bb)$ by the $2^{d-1}$ even respectively odd subsets of $\{1,...,d\}$.
	\item[(ii)] A basis of $K_{q}(\Ff_\lambda)$ by the $2^{d-2}$ even respectively odd subsets of $\{1,...,d\}$ that exclude the direction parallel to the normal vector $n_\lambda$ of the face $\lambda$ (see Fig.~\ref{fig:square}a).
	\item[(iii)] A basis of $K_{q}(\Cc_\lambda)$ by the $2^{d-3}$ even respectively odd subsets of $\{3,...,d\}$ if $d\geq 3$; for $d=2$ there is only the even subset $I=\emptyset$ and no odd subsets.
\end{enumerate}



We can now express the first-order bulk-boundary map $\partial^{\Bb\Ff}_\ast:K_\ast(\Bb)\to K_{\ast-1}(\Ff)$. Since $\Ff$ is a direct sum of four algebras one can treat them individually writing $\partial^{\Bb\Ff}_\ast=\oplus_\lambda \partial^{\Bb\Ff_\lambda}_\ast$ and one has: 
\begin{proposition}
\label{prop:bulk_face}
For any tuple $v$ of directions in $\RM^d$ with the opposite parity as $q\in \{0,1\}$, the Chern numbers relate under the boundary map by
\begin{equation} 
(-1)^q\langle \partial^{\Bb \Ff_\lambda}_q[x]_q, [\Ch_{{\Ff_\lambda}, v}]\rangle= \langle [x]_{q-1}, [\Ch_{\Bb, v\times n_\lambda}]\rangle,
\end{equation}
which uniquely determines $\partial^{\Bb \Ff_\lambda}_i$.
\end{proposition}
\begin{proof} The exact sequence connecting $\Bb$ and $\Ff_\lambda$ has the form
\begin{equation}
0 \to \Ff_\lambda \to \Hh_\lambda \simeq C^*\ZM^{d-1}\otimes \mathfrak{T} \to \Bb \to 0,
\end{equation}
where $\mathfrak{T}$ is the Toeplitz algebra. Under this isomorphism a translation in the direction $n_\lambda$ corresponds to the co-isometry which generates the Toeplitz extension. The identity for the Chern cocycles is then a well-known duality for cyclic cocycles under under this exact sequence (see \cite{KellendonkRMP2002,ProdanSpringer2016,SStSpringer2022} for treatments at various level of detail).
\end{proof}

The boundary maps $\partial^{\Ff\Cc}_\ast=\oplus_\lambda \partial^{\Ff\Cc_\lambda}_\ast$ are only slightly more complicated:
\begin{proposition}
With the same notation as Proposition~\ref{prop:bulk_face} one has
	\begin{equation}
		\label{eq:facecorner}
		\begin{split}
			(-1)^q\langle \partial_q^{\Ff \Cc_\lambda}[(x_1,x_2,x_3,x_4)]_q, [&\Ch_{{\Cc_\lambda}, v}]\rangle \\ =\langle [x_\lambda]_q, [&\Ch_{{\Ff_\lambda}, v\times n_{\lambda-1}}]\rangle + \langle [x_{\lambda-1}]_q, [\Ch_{{\Ff_{\lambda-1}}, v\times n_{\lambda}}]\rangle.
		\end{split}
	\end{equation}
\end{proposition}
\begin{proof}
The exact sequence of groupoid algebras relating each face algebra $\Ff_\lambda$ to the adjacent corner algebra $\Cc_\lambda$ is isomorphic to an exact sequence
\begin{equation}
	0 \to \Cc_\lambda \to C^\ast\ZM^{d-2}\otimes \KM(\ell^2(\NM)) \otimes \mathfrak{T} \to \Ff_\lambda \to 0.
\end{equation}
In this extension the translation in direction $n_{\lambda-1}$ plays the role of the co-isometry generating the Toeplitz extension which results analogously to Proposition~\ref{prop:bulk_face} in the expression
\begin{equation}
(-1)^q\langle \partial_q^{\Ff_\lambda\Cc_\lambda}([x_\lambda]_q), [\Ch_{\Cc_\lambda, v}]\rangle =\langle [x_\lambda]_q, [\Ch_{\Ff_\lambda, v\times n_{\lambda-1}}]\rangle
\end{equation}
for its boundary map. Denoting $\mathring{\Qq}=\Ker(\Qq\to \Bb)$ one has a commutative diagram
\begin{equation}
\begin{tikzcd}
0 \arrow[r] & \Cc_\lambda \arrow[r] \arrow[d]& C^*\ZM^{d-2}\otimes \KM(\ell^2(\NM)) \otimes \mathfrak{T} \arrow[r] \arrow[d]& \Ff_\lambda \arrow[r] \arrow[d]& 0\\
0 \arrow[r] & \Cc \arrow[r] & \mathring{\Qq} \arrow[r] & \Ff \arrow[r] & 0
\end{tikzcd}
\end{equation}
for each $\lambda$. One has a similar diagram which relates the other adjacent face algebra $\Ff_{\lambda-1}$ with the same corner $\Cc_{\lambda}$ resulting in a boundary map characterized by
\begin{equation}
(-1)^q\langle \partial_q^{\Ff_{\lambda-1}\Cc_\lambda}([x_{\lambda-1}]_q), [\Ch_{\Cc_{\lambda}, v}]\rangle =\langle [x_\lambda]_q, [\Ch_{\Ff_{\lambda-1}, v\times n_{\lambda}}]\rangle.
\end{equation}
The total boundary map is then due to additivity under direct sums given by $$\partial_q^{\Ff\Cc}=\oplus_\lambda \partial_q^{\Ff\Cc_\lambda}=\oplus_\lambda (\partial_q^{\Ff_\lambda\Cc_\lambda}+\partial_q^{\Ff_{\lambda-1}\Cc_\lambda}).$$
\end{proof}
One should note that if one wants to label the Chern cocycles in terms of sets of standard directions $I$, as we did above, some of the contributions would acquire minus signs due to the algebraic property $\Ch_{\Aa, v\times (-n)}=-\Ch_{\Aa, v\times n}$ of the Chern cocycles

\begin{figure}
	\label{fig:square}
	\centering
	\begin{tikzpicture}
		\begin{scope}[shift={(-3,0)},scale=0.5]
			{
				\coordinate (A) at (-2,-2);
				\coordinate (B) at (2,-2);
				\coordinate (C) at (2,2);
				\coordinate (D) at (-2,2);
				
				\draw[thin, dashed] ($(A)!-0.5!(B)$) -- ($(B)!-0.5!(A)$);
				\draw[thin, dashed] ($(B)!-0.5!(C)$) -- ($(C)!-0.5!(B)$);
				\draw[thin, dashed] ($(C)!-0.5!(D)$) -- ($(D)!-0.5!(C)$);
				\draw[thin, dashed] ($(D)!-0.5!(A)$) -- ($(A)!-0.5!(D)$);
				
				\fill[gray!30] (A) -- (B) -- (C) -- (D) -- cycle;
				
				\draw[->] (0,-2) -- +(0,-0.5) node[below] {$n_1$};
				\draw[->] (2,0) -- +(0.5,0) node[right] {$n_2$};
				\draw[->] (0,2) -- +(0,0.5) node[above] {$n_3$};
				\draw[->] (-2,0) -- +(-0.5,0) node[left] {$n_4$};
				
				\draw (A) node[below left] {$1$};
				\draw (B) node[below right] {$2$};
				\draw (C) node[above right] {$3$};
				\draw (D) node[above left] {$4$};

				\draw (-3,-3) node[below left] {$a)$};
			}
		\end{scope}
		\begin{scope}[shift={(3,0)},scale=0.5]
			{
				\coordinate (A) at (-2,-2);
				\coordinate (B) at (2,-2);
				\coordinate (C) at (2,2);
				\coordinate (D) at (-2,2);
				
				\draw[thin, dashed] ($(A)!-0.5!(B)$) -- ($(B)!-0.5!(A)$);
				\draw[thin, dashed] ($(B)!-0.5!(C)$) -- ($(C)!-0.5!(B)$);
				\draw[thin, dashed] ($(C)!-0.5!(D)$) -- ($(D)!-0.5!(C)$);
				\draw[thin, dashed] ($(D)!-0.5!(A)$) -- ($(A)!-0.5!(D)$);
				
				
				\draw[line width=0.5mm] (A) -- (B) {};
				\draw[line width=0.5mm] (B) -- (C) {};
				\draw[line width=0.5mm] (C) -- (D) {};
				\draw[line width=0.5mm] (D) -- (A) {};
				\draw (0,-2) node[below] {$w$};
				\draw (2,0)  node[right] {$x$};
				\draw (0,2) node[above] {$y$};
				\draw (-2,0)  node[left] {$z$};
				
				\draw (A) node[below left] {$w-z$};
				\draw (B) node[below right] {$x-w$};
				\draw (C) node[above right] {$y-x$};
				\draw (D) node[above left] {$z-y$};
				
				\draw (-3,-3) node[below left] {$b)$};
			}
		\end{scope}
	\end{tikzpicture}
	\caption{a) Labeling of the corners and half-spaces with their outward normals. b) Pictorial description of the face-corner correspondence: At the faces we mark the Chern numbers $\Ch_{{\Ff_\lambda},v\times n_{\lambda+1}}$ of some class in $K_i(\Ff)$ and at the corners the Chern numbers $(-1)^q \Ch_{{\Cc_\lambda},v}$ of its image under $\partial^{\Ff\Cc}$.
	}
\end{figure}

It is now easy to compute the range of the map $\partial^{\Ff \Cc}$:

\begin{proposition}
In complex (non-equivariant) $K$-theory the map $\partial^{\Ff \Cc}_q: K_q(\Ff)\to K_{q-1}(\Cc)$ has the range
\begin{equation}
\mathrm{Im}\, \partial^{\Ff \Cc}_q \simeq K_{q-1}(\Cc_1)\oplus K_{q-1}(\Cc_2) \oplus K_{q-1}(\Cc_3)
\end{equation}
and 
$K_q(\Cc)/\mathrm{Im} \, \partial^{\Ff \Cc}_q \simeq K_{q-1}(\Cc_4)$.
\end{proposition}

\begin{proof} Any class in $K_i(\Cc)$ is determined by $4$ sets of corner Chern numbers $\Ch_{\Cc_\lambda, I}$, $\lambda \in \ZM_4$, $I\subset \{3,...,d\}$. Any assignment of corner Chern numbers to the corners $1,2,3$ can be re-produced from a pre-image under $\partial^{\Ff\Cc}$ by choosing appropriate face Chern numbers. However, the sums
\begin{equation}
\sum_{\lambda=1}^4 \langle \partial^{\Ff\Cc_\lambda}[(x_1,x_2,x_3,x_4)]_q, [\Ch_{{\Cc_\lambda}, I}]\rangle
\end{equation}
are constrained to be zero for each $I\subset \{3,...,d\}$ (see Figure~\ref{fig:square}) which determines the Chern numbers of the remaining corner. \end{proof}

In principle the sums of the corner Chern numbers can therefore be independent of the lift of any fixed bulk class $[x]_q\in K_q(\Bb)\cap \Ker \, \partial^{\Bb\Ff}_q$. Nevertheless, we have the negative result:
\begin{proposition}
The bulk-corner map $\delta^{\Bb \Cc}_q$ {\rm (}$q=0,1${\rm )} is the zero-map.
\end{proposition} 
\begin{proof} The kernel of $\partial^{\Bb\Ff}_q$ is spanned by precisely those basis elements which are dual to the Chern cocycles $\Ch_{\Bb,I}$, where $I$ contains neither $1$ nor $2$. This span is nothing but the image of the inclusion $K_q(C(\TM^{d-2}))\to K_q(\Bb)$. Those generators can be represented by projections/unitaries that act trivially on the first factor of the decomposition $\ell^2(\ZM^d)=\ell^2(\ZM^2)\otimes \ell^2(\ZM^{d-2})$ and therefore their restrictions to any half- or quarterspace will still be a projections/unitaries. Hence they are in the kernel of $\delta^{\Bb\Cc}_q$.
\end{proof}

This is the reason why one needs to enhance the $K$-theory by spatial symmetries to obtain non-trivial higher-order bulk-boundary correspondence. We used here complex $K$-theory but the same is true for the real Altland-Zirnbauer classes. Those can be labeled by real K-theory groups $KO^i(\TM^d)$ and the kernel of the first boundary map is spanned by the image of $KO^i(\TM^{d-2})$. We omit the details.

The same applies to the cube geometry (see below) and it seems to us that for all polyhedral geometries the higher-order bulk-boundary maps are all zero unless one imposes an additional spatial symmetry. We will leave a more precise analysis to future work.
  
\subsection{The equivariant case: computational aspects}\label{Sec:Adiabatic}
We now impose invariance under an additional twisted action of finite subgroup $\Gamma\subset \overline{\Sigma}=\Sigma\times \ZM_2 \times \ZM_2$, where $\Sigma$ is the symmetry group of a square and two conjugate-linear $\ZM_2$-actions representing time-reversal and particle-hole symmetry, respectively. For simplicity, we are mostly interested in the case where $\Gamma$ acts freely on faces and corners in the sense that $\ZM^d\in \Xi_\square$ is the only element with non-trivial stabilizer group. The forgetful maps $K_\ast^\Gamma(\Ff)\to K_\ast(\Ff)$ and $K_\ast^\Gamma(\Cc)\to K_\ast(\Cc)$ then turn out to be injections due to Proposition~\ref{prop:reduction_equivariant}. The second-order boundary map $\partial^{\Bb \Cc}_\ast$ takes value in the quotient group $K_{\ast-1}^\Gamma(\Cc)/\partial^2_*({\rm Im}\, \partial^{\Ff\Cc}_*)$ which can then be seen as a sub-quotient of $K_{*-1}(\Cc)$.
As we shall see, computing that group will be very easy since one merely needs to enumerate the algebraic constraints that the symmetry imposes on the possible values of the corner Chern numbers.



The difficult task is to find Hamiltonians in the wire geometry whose hinge states are still non-trivial after taking the quotient. In the following, we briefly describe a powerful recent approach to the construction of Hamiltonians with analytically computable topological invariants: In the space-adiabatic approach, one considers Hamiltonians with slowly modulated domain walls as being well-described by adiabatic symbols, i.e. matrix-valued phase-space functions which depend on space and momentum. Such ideas are popular in physics \cite{TeoKane2010,ShiozakiEtAlPRB17} and one can make them rigorous, in particular for models on continuous space via pseudodifferential methods (see \cite{Bal22} for a relevant example) and, more importantly for us, also on the lattice $\ZM^d$ \cite{OSP2024}. 

The first step is to simplify the operator algebras from groupoid $C^*$-algebras to crossed product algebras in the square geometry: 
\begin{proposition}
\label{prop:crossed_products}
There exist locally compact $\ZM^d$-spaces $\mathcal{X}_1, \mathcal{X}_2 \in \mathcal{C}(\RM^d)$ such that 
each $\Gg_{\Xi_r}$ naturally includes into the transformation groupoid $\mathcal{X}_r\rtimes \ZM^d$. One then has a commutative diagram 

\begin{equation}\label{eq: diagram d}
   \begin{tikzcd} 
	C^*\Gg_{\Xi_2} \arrow[d, "\psi_2"] \arrow[r] & C^*\Gg_{\Xi_{1}}  \arrow[d, "\psi_{1}"] \arrow[r] & C^*\Gg_{\Xi_0}  \arrow[d, "\psi_0"]\\
	C_0(\mathcal{X}_2)\rtimes \ZM^d \arrow[r] &	C_0(\mathcal{X}_1)\rtimes \ZM^{d}  \arrow[r]  & \CM\rtimes \ZM^{d}  
\end{tikzcd} 
\end{equation}
 where the rows are cofiltrations and the vertical arrows are injections. If $\ZM^d$ is the only element $\Xi_{\square}$ with non-trivial stabilizer then this gives rise to an isomorphism of spectral sequences in $\Gamma$-equivariant K-theory.
\end{proposition}
\begin{proof}
Define $\mathcal{X}_r:=\overline{\{\Ss -x\,|\, \Ss\in \Xi_r, x\in \ZM^d\}}\subset \Cc(\RM^d)\cap \Cc(\ZM^d)$ with closure in the Fell topology. Since $\ZM^d$ is discrete it is easy to see out that $\Xi_r \subset \mathcal{X}_r$ is an open subset, then $\Gg_\Xi\subset \mathcal{X}_r\rtimes \ZM^d$ is an open subgroupoid of the transformation groupoid. Thus one gets an injective homomorphism $\psi_r\colon C^*\mathcal{G}_{\Xi_r}\to C_0(\mathcal{X}_r)\rtimes \ZM^d$ and arrives at the commutative diagram \eqref{eq: diagram d}.



To prove equivalence of the spectral sequences it is by Proposition~\ref{prop:homomorphism_spectral_sequence} enough to show that the inclusions between the ideals $C^*(\Gg_{\Xi_{r}\setminus \Xi_{r-1}})\to C_0(\mathcal{X}_r\setminus \mathcal{X}_{r-1})\rtimes \ZM^d$ induce isomorphisms in equivariant K-theory. If the action is free in the stated sense, then those K-groups decompose naturally like the right-hand side of \eqref{eq:building_blocks} and each summand has a trivial stabilizer group. As such, the inclusions are isomorphisms due to the stability of K-theory.
\end{proof}

\begin{remark}{\rm In real-space representations, the difference between the groupoid and crossed product algebras is that the former act on the Hilbert spaces $\ell^2(\Ll)$ for patterns $\Ll\in \Xi_2$, whereas the latter act on $\ell^2(\ZM^d)$, but are represented by operators whose matrix elements decay outside of $\Ll$, i.e. they are essentially only supported a quarter- or half-space but not via a sharp truncation. For questions of K-theory, this distinction is usually irrelevant due to stability.} $\Diamond$
\end{remark}

For the crossed products, we describe in \cite{OSP2024} a general method to construct non-trivial higher-order topological models by quantization of adiabatic symbol functions, which are elements of a commutative $C^*$-algebra. To make such an Ansatz, one includes the $\ZM^d$-space $\mathcal{X}_r$ into a larger locally compact $\ZM^d$-space $\mathcal{X}_r^{\mathrm{ad}}$ so that the $\ZM^d$-action is the restriction of an $\RM^d$-action and one has an additional scaling action of $\RM_+$. For the square geometry, there is a natural choice (see \cite[Section~7.4]{OSP2024}), where the one-point compactification of $\mathcal{X}_r^{\mathrm{ad}}$ is a CW-complex with a filtration by closed subsets of the form
\begin{align*}
	\label{eq:config_square}
	\mathcal{X}_2^{\mathrm{ad}}\;&=\; \bigsqcup_{\lambda\in \ZM_4} \RM^2_{\lambda,\lambda+1} \bigsqcup_{\lambda\in \ZM_4} \RM_{\lambda} \sqcup \{ +\infty\}\\
	\mathcal{X}_1^{\mathrm{ad}}\;&=\; \bigsqcup_{\lambda\in \ZM_4} \RM_{\lambda} \sqcup \{ +\infty\}\\
	\mathcal{X}_0^{\mathrm{ad}}\;&=\; \{ +\infty\}
\end{align*} 
with labeled copies of $\RM^2$ and $\RM$ which carry actions of $\RM^d$ as well as an added invariant point $+\infty$ representing the bulk.  Due to the scaling action one can connect the crossed product $C_0(\mathcal{X}^{\mathrm{ad}}_r)\rtimes \ZM^d$ with the tensor product $C_0(\mathcal{X}^{\mathrm{ad}})\otimes C(\TM^d)$, as fibers of a continuous field of $C^*$-algebras. One can then quantize self-adjoint symbol functions in $C_0(\mathcal{X}_2^{\mathrm{ad}})\otimes C(\TM^d)$ to self-adjoint Hamiltonians in $C_0(\mathcal{X}_2^{\mathrm{ad}})\rtimes\ZM^d$ in a way that preserves symmetries and spectral gaps \cite[Proposition 2.6]{OSP2024}, thus giving rise to canonical homomorphisms $K^\Gamma_\ast(C_0(\mathcal{X}_2^{\mathrm{ad}})\otimes C(\TM^d))\to K^\Gamma_\ast(C_0(\mathcal{X}_2^{\mathrm{ad}})\rtimes \ZM^d)$. Due to their naturalness they further lead to a homomorphisms of spectral sequences mapping from the spectral sequence obtained from the cofiltration $$C(\mathcal{X}_r^{\mathrm{ad}}\times \TM^d)\to C(\mathcal{X}_{r-1}^{\mathrm{ad}}\times \TM^d)\to... \to C(\mathcal{X}_0^{\mathrm{ad}}\times \TM^d)$$ 
to the two equivalent spectral sequences from Proposition~\ref{prop:crossed_products}, which is much more amenable to the description by the methods of equivariant topology. There is generally no expectation that one obtains an isomorphism in this way. Nevertheless, one can use that formalism to prove the existence of non-trivial examples for higher-order bulk-boundary correspondences as long as one knows how to explicitly relate topological invariants of the symbols with Chern numbers of their quantizations and this is accomplished in \cite[Theorem 1.10]{OSP2024}: The face and corner Chern numbers w.r.t. $m$ additional directions $v_1,...,v_m$ in our case correspond to $(m+1)$-dimensional Chern numbers of the symbol restricted to the one-cells $\RM_\lambda \times \TM^d$ respectively $(m+2)$-dimensional Chern numbers of the restriction to two-cells $\RM_{\lambda,\lambda+1}$. One can thus construct with manageable effort Hamiltonians with gaps at prescribed limits and known Chern numbers just by doing analysis of matrix-valued symbol functions. Each such example provides partial input to the computation of the higher-order boundary map $\delta^{\Bb\Cc}_*$ and in some symmetry classes adiabatic quantization in fact provides enough explicit Hamiltonians to completely determine it.

\subsection{Inversion-symmetry}
\label{ssec:inversion}
Here, we explicitly compute the bulk-corner map in the square geometry for the inversion-symmetry in dimension $d=3$, i.e. the order two symmetry generated by the involutive transformation
\begin{equation}
\sigma : x\in \RM^3 \mapsto -x \in\RM^3.
\end{equation}
This symmetry induces an involutive automorphism $\sigma$ on $C^\ast\Gg_\square$,  which we also consider as a $\ZM_2$-action denoted by the same letter. To distinguish from the other symmetries considered in this work, we denote the $\ZM_2$-equivariant $K$-groups w.r.t. inversion symmetry by $K_\ast^I(C^\ast \Gg_\square)\equiv K_\ast^{\ZM_2}(C^\ast \Gg_\square)$. 

We compute first $K^I_\ast(\Cc)/\mathrm{Im} \, \partial^{\Ff\Cc}_\ast$, which embeds the codomain of the second order boundary map: 
\begin{proposition}It holds that
\label{prop:inversion}
\begin{equation}
K_1^I(\Cc) \simeq \ZM^2 \supset \mathrm{Im} \, \partial^{\Ff\Cc}_0  \simeq \ZM \oplus (2\ZM)
\end{equation}
and 
\begin{equation}\label{Eq:P1}
\mathrm{Im} \, \delta_0^{\Bb \Cc} \subseteq K_1^I(\Cc)/\mathrm{Im} \, \partial^{\Ff\Cc}_0 \simeq \ZM_2.
\end{equation}
Above, $\ZM_2$ encodes the parity of the difference of the corner Chern numbers of the two not symmetry-related corners.
\end{proposition}

\begin{proof}
    The face algebra $\Ff = \oplus_{\lambda \in \ZM_4}\Ff_\lambda$ can be decomposed into two orbits under inversion-symmetry by writing
\begin{equation}
\Ff = \Ff_1 \oplus \Ff_2 \oplus \sigma(\Ff_1) \oplus \sigma(\Ff_2) \simeq (\Ff_1\oplus \Ff_2)\otimes C(\ZM_2),
\end{equation}
with $C(\ZM_2)$ the functions $\ZM_2\to \CM$. The isomorphism is equivariant when $C(\ZM_2)$ carries the regular representation. Hence, one has by Proposition~\ref{prop:reduction_equivariant}
\begin{equation}
K_\ast^I(\Ff)\simeq K_\ast(\Ff_1 \oplus \Ff_2).
\end{equation} 
Similarly one can decompose $\Cc$ into orbits and obtain
\begin{equation}
K_\ast^I(\Cc)\simeq K_\ast(\Cc_1 \oplus \Cc_2).
\end{equation}
Here we do not distinguish between $K^I_{-1}\simeq K^I_1$ and represent both by invariant unitaries (see Remark~\ref{rem:picturesk_theory}(i)).
The boundary maps in equivariant $K$-theory are now easy to compute by forgetting the equivariance
\begin{equation}
	\begin{tikzcd}
		K_\ast^I(\Ff) \arrow[d] \arrow[r,"{\partial}^{\Ff \Cc}_\ast"] & K_{\ast-1}^I(\Cc) \arrow[d] \\ 
		K_\ast(\Ff) \arrow[r,"{\partial}^{\Ff\Cc}_\ast"] & K_{\ast-1}(\Cc)
	\end{tikzcd}
\end{equation}
but keeping in mind that the image of $K_\ast^I(\Ff)\to K_\ast(\Ff)$ is generated by $\sigma$-invariant representatives. Therefore, it is enough to compute the boundary map for representatives $[x_1,x_2,x_3,x_4]_\ast\in \bigoplus_{\lambda=1}^4 K_\ast(\Ff_i)$ of the form $[x_1,x_2,\sigma(x_1),\sigma(x_2)]_\ast$. Here, $x_\lambda \in M_N(\CM)\otimes \Ff_\lambda$ and $\sigma$ acts on the first factor by some unitary representation of $\ZM_2$. This implies relations for the Chern numbers of the faces 
\begin{equation}
\begin{aligned}
\langle [x_\lambda]_0, [\Ch_{{\Ff_\lambda}, e_3 \times n_{\lambda+1}}]\rangle &= \langle \sigma([x_\lambda]_0), [\Ch_{{\Ff_{\lambda+2}}, (-e_3) \times (-n_{\lambda+1})}]\rangle \\
&= -\langle \sigma([x_\lambda]_0), [\Ch_{{\Ff_{\lambda+2}}, e_3 \times n_{\lambda+3}}]\rangle ,
\end{aligned}
\end{equation}
since $\sigma$ maps $\Tt_{\Ff_\lambda}$ to $\Tt_{\Ff_{\lambda+2}}$ but flips the signs of all derivations $\sigma \circ \nabla_i = -\nabla_i \circ \sigma$ due to $\Theta_x \circ \sigma = \sigma \circ \Theta_{-x}$.
Those Chern numbers are the only relevant ones for the computation of $\partial^{\Ff\Cc}$ via \eqref{eq:facecorner} since $K_1(\Cc)\simeq \ZM^4$ is labeled precisely by the four Chern cocycles $\Ch_{{\Cc_{\lambda}}, e_3}$. Due to $\sigma$-invariance, classes in  $K^I_1(\Cc)\simeq \ZM^2$ are then already determined precisely by the pairings with $\Ch_{{\Cc_{\lambda}}, e_3}$ for $\lambda \in \{1,2\}$.

One can now read off the possible values for the corner Chern numbers from Figure~\ref{fig:square}: The inversion symmetry implies the constraints $y=-w$ and $z=-x$ and therefore the difference of the two corner Chern numbers $w-z$ and $x-w=-z-w$ lies in $2\ZM$.
\end{proof}

Stated in terms of Hamiltonians this gives:

\begin{corollary}\label{Cor:c1} An element of the bulk $K$-group $x \in K_0^I(\Bb)$ maps under $\delta_0^{\Bb \Cc}$ to the non-trivial parity sector in \eqref{Eq:P1} if and only if there exists an inversion-symmetric Hamiltonian $h\in B(\Vv)\otimes C^*(\Gg_\square)$ such that $\mathfrak{p}^2_0(h)$ is gapped, $[\gamma_{(\mathfrak{p}^1\circ \mathfrak{p}^2)(h)}]=x$ and one has
		\begin{equation}
			\label{eq:hinge_modes}\langle \partial^2_0(\mathfrak{p}^2(h)), [\Ch_{\Cc_1,e_3}]\rangle+ \langle \partial^2_0(\mathfrak{p}^2(h)), [\Ch_{\Cc_2,e_3}]\rangle \;=\; 1 \mod 2.
		\end{equation}
If such element $x \in K_0^I(\Bb)$ exists, then $\mathrm{Im} \, \delta_0^{\Bb \Cc} \simeq \ZM_2$.
\end{corollary}

\begin{remark}{\rm In words, \eqref{eq:hinge_modes} says that, when model $h$ is deployed on a square geometry, out of any two adjacent hinges one has an odd and the other has an even Chern number. In that case a change of boundary condition that keeps the faces gapped cannot remove all hinge modes; at most it can move them between different sets of hinges. It would have been impossible to obtain such a constraint for a single hinge or even two inversion-related hinges. Here we get the payoff for constructing the novel groupoid algebra in Section~\ref{Sec:ModelAlg}: Only by taking the geometric relations between the hinges into account and gluing them together to provide an algebra for an infinite crystal with boundary we managed to obtain those non-trivial maps in equivariant $K$-theory.
}$\Diamond$
\end{remark}

Our final task is to show that a Hamiltonian as in Corollary~\ref{Cor:c1} exists. For this, consider the slight modification of the bulk Hamiltonian listed in \cite{HughesPRB2011}[Eq.~145]
\begin{equation}\label{Eq:Ham1}
h = \tfrac{1}{2 \imath} \sum_{i=1}^3 \Gamma_i \otimes (S_i - S_i^\ast)+ \Gamma_0 \otimes \Big(2 +\tfrac{1}{2}\sum_{i=1}^3 (S_i + S^\ast_i)\Big ) + \gamma \Gamma_B \otimes 1,
\end{equation}
where $S_i$'s are the generators of $C^\ast\ZM^3$, $\Gamma_1=\sigma_3\otimes \sigma_1$, $\Gamma_2=\one \otimes \sigma_2$, $\Gamma_3=\sigma_2\otimes \sigma_1$, $\Gamma_0=\one \otimes \sigma_3$ and $\Gamma_B = \frac{1}{2}(\sigma_1+\sigma_2+\sigma_3)\otimes (\one+\sigma_3)$, with the Pauli matrices $\sigma_i$. On the atomic orbitals space $M_4(\CM)$, we assume that the action of inversion is implemented by conjugation with the matrix $\one \otimes \sigma_3$. 

One can prove that this bulk Hamiltonian realizes the non-trivial parity in Proposition~\ref{prop:inversion} by extending its Fourier transform, a function in $M_4(\CM)\otimes C(\TM^3)$, to a matrix-valued adiabatic symbol function over $ (C_0(\mathcal{X}_2^{\mathrm{ad}}\times \TM^3)^\sim)$ (where $\sim$ denotes the unititization) which becomes gapped when restricted to $(C_0(\mathcal{X}_1^{\mathrm{ad}}\times \TM^3))^\sim$. The outcome is as follows:

\begin{proposition}[{\cite[Proposition~7.6]{OSP2024}}]
There exists a Hamiltonian $h \in B(\Vv) \otimes (C_0(\mathcal{X}_2)\rtimes \ZM^3)^\sim$, $\Vv=\CM^2\otimes \CM^2$, with the following properties:
	\begin{enumerate}
		\item[(i)] $h$ is invariant under the $\ZM_2$-action which implements the inversion symmetry.
		\item[(ii)] $h$ gets mapped to a gapped Hamiltonian $\mathfrak{p}^2(h)$ under the surjection $$\mathfrak{p}^2: B(\Vv) \otimes (C_0(\mathcal{X}_2)\rtimes \ZM^3)^\sim \to B(\Vv) \otimes (C_0(\mathcal{X}_1)\rtimes \ZM^3)^\sim$$
		 and thus defines a class $[\gamma_{\mathfrak{p}^2(h)}]_0\in K_0^I(C_0(\mathcal{X}_1)\rtimes \ZM^3)$.
		\item[(iii)] Under the surjection $$\mathfrak{p}^1\circ \mathfrak{p}^2: B(\Vv) \otimes (C_0(\mathcal{X}_2)\rtimes \ZM^3)^\sim \to B(\Vv)\otimes C^*(\ZM^3)$$ the image $\mathfrak{p}^1\circ \mathfrak{p}^2(h)$ is equal to the bulk Hamiltonian \eqref{Eq:Ham1}.

		\item[(iv)] The image of the boundary map
		$$\partial^2_0: K_0^I(C_0(\mathcal{X}_1)\rtimes \ZM^3) \to K_1^I(C_0(\mathcal{X}_2\setminus \mathcal{X}_1)\rtimes \ZM^3)\simeq K_1(\Cc_1)\oplus K_1(\Cc_2)$$ is again labeled by two Chern cocycles $\Ch_{{\Cc_1},e_3}$ and $\Ch_{{\Cc_1},e_3}$ for which the K-theoretic pairings take integer values. For the given $h$ one has
		\begin{equation*}\langle \partial^2_0([\gamma_{\mathfrak{p}^2(h)}]_0), [\Ch_{\Cc_1,e_3}]\rangle+ \langle \partial^2_0([\gamma_{\mathfrak{p}^2(h)}]_0), [\Ch_{\Cc_2,e_3}]\rangle \;=\; 1 \mod 2.
		\end{equation*}
	\end{enumerate}
\end{proposition}

In light of the above discussion, this implies:
\begin{corollary}
	The class in $K_0^I(C(\TM^3))$ represented by the gapped Hamiltonian \eqref{Eq:Ham1} maps to the non-trivial parity under $\delta^2_0: \Ker(\delta^1_0) \to K_1^I(\Cc)/\mathrm{Im} \, \partial^{\Ff\Cc}_0 \simeq \ZM_2$.
\end{corollary}
\begin{remark}{\rm Having found a non-trivial image-preimage pair under $\delta^2_0$ the remaining classification task would merely be to determine the kernel of $\delta^2_0$, which does not immediately relate to observable phenomena.
}    $\Diamond$
\end{remark}

\begin{remark}
	{\rm  It is standard to compute the  $K$-groups $K_*^I(C(\TM^3))$ using the Atiyah-Hirzebruch spectral sequence (AHSS) (see \cite[Page 22]{ShiozakiAHSSspacegroups}).  One obtains
			\begin{equation}
				K_0^I(C(\TM^3))\simeq \ZM^{12}, \quad K_1^I(C(\TM^3))=0.
			\end{equation}
		Importantly, the K-group does not have torsion components, hence $\partial^{\Bb\Cc}_0$ must end up being a mod $2$ linear combination of integer invariants. There is a strong candidate for that map: There are $8$ points $\xi_1,...,\xi_8\in \TM^3$ that are fixed points under inversion, called time-reversal-invariant momenta (TRIM), since time-reversal is just inversion composed with complex conjugation. To any invariant (virtual) vector bundle on the torus one can assign $8$ integers $n_a$, $a=1,...,8$, which are the multiplicities of the sign representation of $\ZM_2$ at the TRIM. This provides a homomorphism $K_0^I(C(\TM^3))\to \ZM^8$.	To a vector bundle on $\TM^3$ one can also associate the Chern-Simons invariant $$\theta_{\mathrm{CS}}= \frac{1}{4\pi}\int_{\TM^3} \mathrm{tr}(dA \wedge A + \frac{2}{3} A \wedge A \wedge A)$$ 
		where $A$ is the Berry connection on the vector bundle. Modulo $2\pi$ it is quantized to $0$ or $\pi$ in the presence of inversion symmetry and an invariant under stable equivariant homotopy. If all weak Chern numbers of the vector bundle are trivial then it is known to the level of rigor common in theoretical physics \cite{HughesPRB2011} that one has the relation
		\begin{equation}
			\label{eq:cs_parity}
			\frac{\theta_{\mathrm{CS}}}{\pi} = \frac{1}{2}\Big( \sum_a n_a \mod 4\Big)
		\end{equation}
	which can only take the values $0$ or $1$. In fact this relation can in principle be made rigorous post-hoc simply by checking it for each element of a basis of $K_0^I(C(\TM^3))$ which can be obtained by independent means. Non-trivial secondary characteristic classes such as the Chern-Simons invariant have been associated with higher-order boundary states in the literature \cite{SchindlerNeupert,UribeEtAl22}, however, there appears to be as of yet no structural argument why that might be the case. An obvious conjecture is therefore that the right-hand side of \eqref{eq:cs_parity} is exactly the sought map $\delta^{\Bb\Cc}_0$ and, indeed, we chose the Hamiltonian \eqref{Eq:Ham1} as a promising candidate for second-order boundary states precisely since it satisfies that non-trivial parity. }$\Diamond$
\end{remark}


\begin{figure}[t]
\center
\includegraphics[width=\textwidth]{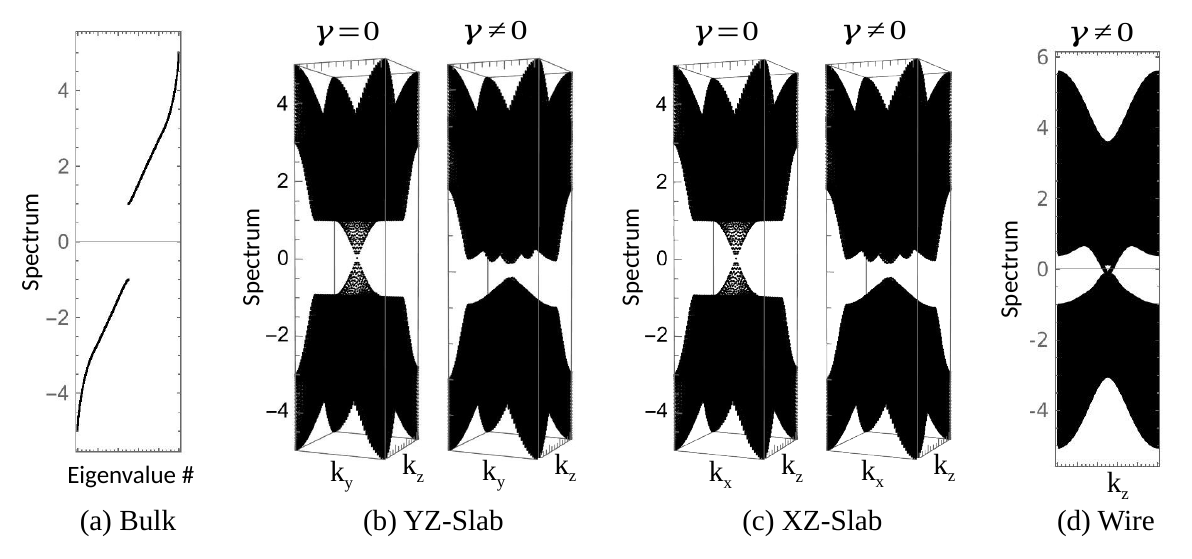}\\
  \caption{\small Spectral analysis of Hamiltonian~\eqref{Eq:Ham1} with different boundary conditions in a finite $30\times 30\times 30$ cube and for $\gamma=0$ as well as for $\gamma=0.5$. (a) the eigenvalues are rendered in increasing order. (b,c,d) the eigenvalues are resolved by the momenta in the directions with periodic boundary conditions while the remaining directions have open boundary conditions. Panels b) and c) show that the model is gappable at the boundaries of codimension 1. In panel d), near the origin of the vertical axis, we see a locus of non-degenerate eigenvalues. In real space, the corresponding states are localized at two opposite hinges (see Fig.~\ref{Fig:HingeModes}a).}
 \label{Fig:Model1}
\end{figure}

For completeness, we present in Fig.~\ref{Fig:Model1} a numerical analysis of model~\eqref{Eq:Ham1}. Note that, if the parameter $\gamma$ is set to $0$, the time reversal and the space inversion become separate symmetries of the model and the Hamiltonian~\eqref{Eq:Ham1} becomes a strong topological insulator from the class $AII$ of the classification table \cite{RyuNJP2010}. This means whenever a surface is cut into a bulk sample the bulk spectral gap as seen in Fig.~\ref{Fig:Model1}(a) will be filled with surface spectrum. For a flat surface, the latter can be resolved by the quasi-momenta $k_j \in \TM$ along the surfaces, by using the isomorphism $C^\ast\ZM^3 \simeq C^\ast\ZM \otimes C(\TM^2)$, $S_j \mapsto u_j(k_j)= e^{\imath k_j}$, where $j$ samples the cartesian directions parallel to surface. If done so, one can see in the momentum-resolved spectrum as in Fig.~\ref{Fig:Model1}(b,c) the hallmark Dirac spectral singularities originating from protected surfaces states.\footnote{The $yz$-slab, for example, results from a confinement of the material in the $x$-direction. Thus, a slab displays two infinite surfaces, but these surfaces are separated by a distance large enough to reduce the effects of their interference below the resolution of our figures.} When $\gamma$ is set to a non-zero value, the time-reversal symmetry is broken but the space inversion symmetry persists. As seen in Fig.~\ref{Fig:Model1}(b,c), the Dirac singularities are lifted and the surfaces of the slabs become spectrally gapped, regardless of the orientations of the cuts. Yet, when we cut a wire with a large square section out of the bulk sample, the spectrum becomes again un-gapped due to two infinitely thin bands crossing the slabs' spectral gaps in Fig.~\ref{Fig:Model1}(d). This spectrum is supported by 1-dimensional wave channels that develop along the hinges of the wires.

\subsection{$C_2 T$-symmetry}
\label{ssec:c2t}
Here we consider an anti-linear symmetry. In three dimensions $C_2$ shall be rotation by $\pi$ in the $x_1$-$x_2$-plane, hence implemented by the operation
\begin{equation}
C_2: (x_1,x_2,x_3)\in \RM^3 \mapsto (-x_1,-x_2,x_3).
\end{equation}
As for inversion this is an order two symmetry, whose action on $C^\ast\Gg_\square$ we denote by the same symbol $C_2$. We compose it with time-reversal $T$ which acts trivially on $\RM^3$ but acts by complex conjugation on the convolution algebra $C^\ast\Gg_\square$, i.e. it is an involutive anti-automorphism. The same is true for the composition $C_2T$. One can therefore consider the twisted equivariant $K$-groups 
\begin{equation}
K_\ast^{C_2T}(C^\ast\Gg_\square) := {}^\phi K_{\ast,c,\tau}^{\ZM_2}(C^\ast\Gg_\square),
\end{equation} 
for $\ast\in \{0,-1\}$ and where $c=0$ and $\tau$, $\phi$ just express that the generator shall be represented anti-unitarily in any $(\phi,c,\tau)$-twisted representation of $\ZM_2$.

As for inversion-symmetry, opposite sides and corners of the square are conjugate under the symmetry:
\begin{proposition}It holds that
\begin{equation}
K_\ast^{C_2T}(\Ff)\simeq K_\ast(\Ff_1 \oplus \Ff_2)
\end{equation}
and
\begin{equation}
K_\ast^{C_2T}(\Cc)\simeq K_\ast(\Cc_1 \oplus \Cc_2).
\end{equation}
\end{proposition}

\begin{figure}[t]
\center
\includegraphics[width=\textwidth]{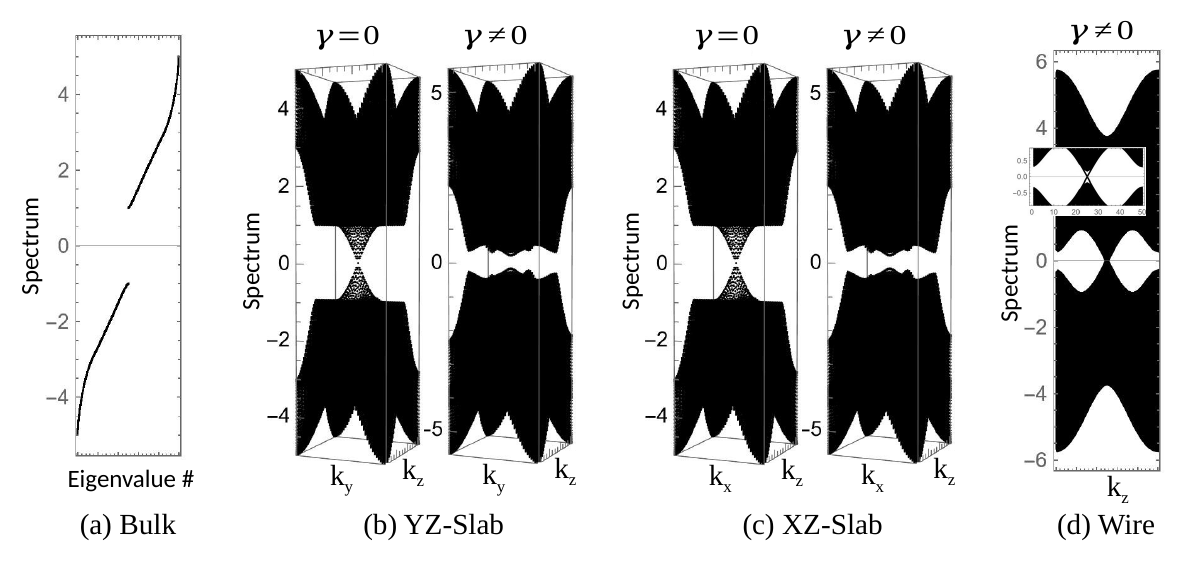}\\
  \caption{\small Same as Fig.~\ref{Fig:Model1} but for the Hamiltonian~\eqref{Eq:Ham2}. The more refined scale shown in the inset of panel (c) is needed to see the topological hinge modes. 
}
 \label{Fig:Model2}
\end{figure}
Here one has $K_1=K_{-1}$ due to Bott periodicity of complex K-theory. One can again think of $K_\ast^{C_2T}(\Ff)$ as being represented by elements $[x_1,x_2,x_3,x_4]_\ast\in \bigoplus_{\lambda=1}^4 K_\ast(\Ff_i)$ of the form $[x_1,x_2,C_2T (x_1),C_2T(x_2)]_\ast$. From a simple computation, one finds precisely the same constraints on the hinge Chern numbers as for inversion symmetry:
\begin{proposition}\label{Prop:ImC2T} It holds that
\begin{equation}
K_{-1}^{C_2T}(\Cc) \supset \mathrm{Im} \, \partial^{\Ff\Cc} _0 \simeq \ZM \oplus (2\ZM)
\end{equation}
	and 
\begin{equation}
\mathrm{Im} \, \delta_0^{\Bb \Cc} \subseteq K_{-1}^{C_2T}(\Cc)/\mathrm{Im} \, \partial^{\Ff\Cc}_0 \simeq \ZM_2.
\end{equation}
\end{proposition}
We now consider the Hamiltonian inspired by \cite{SchindlerSciAdv2018}[Eq.~1]
\begin{equation}\label{Eq:Ham2}
h = \tfrac{1}{2 \imath} \sum_{i=1}^3 \Gamma_i \otimes (S_i - S_i^\ast)+ \Gamma_0 \otimes \Big(2 +\tfrac{1}{2}\sum_{i=1}^3 (S_i + S^\ast_i)\Big ) + \gamma \Gamma_B \otimes 1,
\end{equation}
where $S_i$ are the generators of $C^\ast\ZM^3$, and $\Gamma_i=\sigma_1\otimes \sigma_i$, $\Gamma_0=\sigma_z \otimes 1$ and finally $\Gamma_B = 1\otimes (\sigma_1+\sigma_2)$. On the atomic orbitals space $M_4(\CM)$, the action of the time reversal is implemented by conjugation with $(1 \otimes \sigma_2 ) K$, while that of the 2-fold rotation by conjugation with $\one \otimes e^{\imath \frac{\pi}{2}\sigma_3}$. 

As in the previous subsection, we can prove that this bulk Hamiltonian realizes the non-trivial parity in
Proposition~\ref{Prop:ImC2T}:

\begin{proposition}
There exists a Hamiltonian $h \in B(\Vv) \otimes (C_0(\mathcal{X}_2)\rtimes \ZM^3)^\sim$, $\Vv=\CM^2\otimes \CM^2$, with the following properties:
	\begin{enumerate}
		\item[(i)] $h$ is invariant under the $C_2T$-action.
		\item[(ii)] $h$ gets mapped to a gapped Hamiltonian $\mathfrak{p}^2(h)$ under the surjection $$\mathfrak{p}^2: B(\Vv) \otimes (C_0(\mathcal{X}_2)\rtimes \ZM^3)^\sim \to B(\Vv) \otimes (C_0(\mathcal{X}_1)\rtimes \ZM^3)^\sim$$
		and thus defines a class $[\gamma_{\mathfrak{p}^2(h)}]_0\in K_0^I(C_0(\mathcal{X}_1)\rtimes \ZM^3)$.
		\item[(iii)] Under the surjection $$\mathfrak{p}^1\circ \mathfrak{p}^2: B(\Vv) \otimes (C_0(\mathcal{X}_2)\rtimes \ZM^3)^\sim \to B(\Vv)\otimes C^*(\ZM^3)$$ the image $\mathfrak{p}^1\circ \mathfrak{p}^2(h)$ is equal to the bulk Hamiltonian \eqref{Eq:Ham2}.

		\item[(iv)] The image of the boundary map
		$$\partial^2_0: K_0^I(C_0(\mathcal{X}_1)\rtimes \ZM^3) \to K_1^I(C_0(\mathcal{X}_2\setminus \mathcal{X}_1)\rtimes \ZM^3)\simeq K_1(\Cc_1)\oplus K_1(\Cc_2)$$ is again labeled by two Chern cocycles $\Ch_{{\Cc_1},e_3}$ and $\Ch_{{\Cc_1},e_3}$ for which the K-theoretic pairings take integer values. For the given $h$ one has
		\begin{equation*}\langle \partial^2_0([\gamma_{\mathfrak{p}^2(h)}]_0), [\Ch_{\Cc_1,e_3}]\rangle+ \langle \partial^2_0([\gamma_{\mathfrak{p}^2(h)}]_0), [\Ch_{\Cc_2,e_3}]\rangle \;=\; 1 \mod 2.
		\end{equation*}
	\end{enumerate}
 The class in $K_0^{C_2T}(C(\TM^3))$ represented by the gapped Hamiltonian \eqref{Eq:Ham2} therefore maps to the non-trivial parity under $\delta^2_0: \Ker(\delta^1_0) \to K_1^{C_2T}(\Cc)/\mathrm{Im} \, \partial^{\Ff\Cc}_0 \simeq \ZM_2$.
\end{proposition}
\begin{proof}
Existence can be proven verbatim as in \cite[Proposition~7.6]{OSP2024}, since \eqref{Eq:Ham1} and \eqref{Eq:Ham2} only differ by the choice of matrices $\Gamma_i$ and the term $\Gamma_B$ plays the same role as the gap-opening mass term.
\end{proof}

For completion, we summarize the spectral characteristics of the Hamiltonian~\eqref{Eq:Ham2} when restricted to different geometries in Fig.~\ref{Fig:Model2}. 

\begin{remark}{\rm The $C_2T$-invariant $K$-theory of the torus has been computed recently \cite{UribeEtAl22} and the only generator which can potentially exhibit non-trivial second-order boundary maps corresponds to a vector bundle that is invariant under $C_2T$- and inversion symmetry and has a non-trivial Chern-Simons parity. The Hamiltonian above is not invariant under inversion but must represent that same class.
}$\Diamond$
\end{remark}

\subsection{$C_4T$-symmetry}
\label{ssec:c4t}
We consider the fourfold rotation
\begin{equation}
C_4: (x_1,x_2,x_3)\in \RM^3 \mapsto (-x_2, x_1,x_3)
\end{equation}
composed with time-reversal $T$, i.e. complex conjugation, which defines an order four anti-linear automorphism $C_4T$ on $C^\ast\Gg_\square$. Moreover, the rotation shall be of so-called fermionic type, by which we mean that the representations of $C_4T$ shall be twisted by
\begin{equation}
\label{eq:order4twist}
U^4 = -\one
\end{equation}
where $U$ is the anti-unitary generator of the projective $\ZM_4$ representation. The relevant $K$-groups are 
\begin{equation}
K_\ast^{C_4T}(C^\ast\Gg_\square) := {}^\phi K^{\ZM_4}_{\ast,c,\tau}(C^\ast\Gg_\square)
\end{equation} 
for $\ast\in \{0,-1\}$ and where $c=0$, $\phi$ is determined by the fact that $C_4T$ and $(C_4T)^2$ are anti-linear automorphisms. The twist $\tau$ is the unique $2$-cocycle such that \eqref{eq:order4twist} is imposed for the generators of the $(\phi,c,\tau)$-twisted representations of $\ZM_4$.

In this case all four sides and corners form a single orbit each under rotations, which leads to:
\begin{proposition} It holds that
\begin{equation}
K_\ast^{C_4T}(\Ff)\simeq K_\ast(\Ff_1)
\end{equation}
	and
\begin{equation}
K_\ast^{C_4T}(\Cc)\simeq K_\ast(\Cc_1).
\end{equation}
\end{proposition}
Therefore $K_\ast^{C_4T}(\Ff)$ is represented by elements $[x_1,x_2,x_3,x_4]_\ast\in \bigoplus_{\lambda=1}^4 K_\ast(\Ff_\lambda)$ of the form $[x_1,C_4T(x_1),(C_4T)^2(x_1),(C_2T)^3(x_1)]_\ast$, where $x_1\in M_N(\CM)\otimes \Ff_1$ and $C_4T$ acts on $M_N(\CM)$ by some $(\phi,c,\tau)$-twisted representation. This implies the following relation for the relevant Chern numbers
\begin{equation}
\begin{aligned}
	\langle [C_4T(x_{\lambda})]_0, [\Ch_{{\Ff_{\lambda+1}}, e_3 \times n_{\lambda+2}}]\rangle &= \langle [x_\lambda]_0, [\Ch_{{\Ff_{\lambda}}, e_3 \times n_{\lambda+1}} \circ T]\rangle \\
	&= -\langle [x_\lambda]_0, [\Ch_{{\Ff_{\lambda}}, e_3 \times n_{\lambda+1}}]\rangle.
\end{aligned}
\end{equation}
For the Chern cocycle which determines $K_{-1}^{C_4T}(\Cc)\simeq \ZM$ one therefore has by \eqref{eq:facecorner}
$$\langle \partial^{\Ff\Cc_1}([x_1,C_4T(x_1),(C_4T)^2(x_1),(C_4T)^3(x_1)], [\Ch_{{\Cc_1},e_3}]\rangle = 2 \langle [x_1]_0, [\Ch_{{\Ff_1},e_3\times n_4}]\rangle.$$
Since this is any even integer we conclude
\begin{proposition}
\begin{equation}
K_{-1}^{C_4T}(\Cc) \supset \mathrm{Im} \, \partial^{\Ff\Cc} _0 \simeq 2\ZM
\end{equation}
	and 
\begin{equation}
\mathrm{Im} \, \delta_0^{\Bb \Cc} \subseteq K_{-1}^{C_4T}(\Cc)/\mathrm{Im} \, \partial^{\Ff\Cc}_0 \simeq \ZM_2.
\end{equation}
\end{proposition}

\begin{remark}{\rm Let us highlight something remarkable about this fact: The $\ZM_2$-invariant is the parity of {\it any} of the four corner Chern numbers. Hence if a bulk Hamiltonian maps to the odd parity sector then the corner Chern number of any single corner is non-trivial and has odd parity, even though it would not be possible to detect the global $C_4T$-symmetry of the crystal. Indeed, to derive the maps in equivariant $K$-theory which guarantee the topological protection of the hinge mode we had to adjoin three additional hinges. The only remnant of the symmetry at a single hinge is that the two asymptotic half-spaces adjacent to the hinge are related by a $C_4T$-transformation.}$\Diamond$
\end{remark}

Consider the bulk Hamiltonian listed in \cite{SchindlerSciAdv2018}[Eq.~1]
\begin{equation}\label{Eq:Ham3}
h = \tfrac{1}{2 \imath} \sum_{i=1}^3 \Gamma_i \otimes (S_i - S_i^\ast)+ \Gamma_0 \otimes \Big(2 +\tfrac{1}{2}\sum_{i=1}^3 (S_i + S^\ast_i)\Big ) + \frac{\gamma}{2} \Gamma_B \otimes (S_1+ S^\ast_1- S_2 - S^\ast_2),
\end{equation}
where $S_i$ are the generators of $C^\ast\ZM^3$, and $\Gamma_i=\sigma_1\otimes \sigma_i$, $\Gamma_B = \sigma_2\otimes 1$ and finally $\Gamma_0=\sigma_3 \otimes 1$. On the atomic orbitals space $M_4(\CM)$, the action of the time reversal is implemented by conjugation with $(1 \otimes \sigma_2 ) \Kk$, while that of the 4-fold rotation by conjugation with $1 \otimes e^{\imath \frac{\pi}{4}\sigma_3}$. 

\begin{proposition}
There exists a Hamiltonian $h \in B(\Vv) \otimes (C_0(\mathcal{X}_2)\rtimes \ZM^3)^\sim$, $\Vv=\CM^2\otimes \CM^2$, with the following properties:
	\begin{enumerate}
		\item[(i)] $h$ is invariant under the $C_4T$-action.
		\item[(ii)] $h$ gets mapped to a gapped Hamiltonian $\mathfrak{p}^2(h)$ under the surjection $$\mathfrak{p}^2: B(\Vv) \otimes (C_0(\mathcal{X}_2)\rtimes \ZM^3)^\sim \to B(\Vv) \otimes (C_0(\mathcal{X}_1)\rtimes \ZM^3)^\sim$$
		 and thus defines a class $[\gamma_{\mathfrak{p}^2(h)}]_0\in K_0^I(C_0(\mathcal{X}_1)\rtimes \ZM^3)$.
		\item[(iii)] Under the surjection $$\mathfrak{p}^1\circ \mathfrak{p}^2: B(\Vv) \otimes (C_0(\mathcal{X}_2)\rtimes \ZM^3)^\sim \to B(\Vv)\otimes C^*(\ZM^3)$$ the image $\mathfrak{p}^1\circ \mathfrak{p}^2(h)$ is equal to the bulk Hamiltonian \eqref{Eq:Ham3}.

		\item[(iv)] The image of the boundary map
		$$\partial^2_0: K_0^I(C_0(\mathcal{X}_1)\rtimes \ZM^3) \to K_1^I(C_0(\mathcal{X}_2\setminus \mathcal{X}_1)\rtimes \ZM^3)\simeq K_1(\Cc_1)\oplus K_1(\Cc_2)$$ is again labeled by two Chern cocycles $\Ch_{{\Cc_1},e_3}$ and $\Ch_{{\Cc_1},e_3}$ for which the K-theoretic pairings take integer values. For the given $h$ one has
		\begin{equation}\label{eq:hinge_c4t_ex}\langle \partial^2_0([\gamma_{\mathfrak{p}^2(h)}]_0), [\Ch_{\Cc_1,e_3}]\rangle \;=\; 1 \mod 2.
		\end{equation}
	\end{enumerate}
 The class in $K_0^{C_4T}(C(\TM^3))$ represented by the gapped Hamiltonian \eqref{Eq:Ham2} therefore maps to the non-trivial parity under $\delta^2_0: \Ker(\delta^1_0) \to K_1^{C_4T}(\Cc)/\mathrm{Im} \, \partial^{\Ff\Cc}_0 \simeq \ZM_2$.
\end{proposition}
\begin{proof}
Again, the proof can be given verbatim as in \cite[Proposition~7.6]{OSP2024}, since one only needs to replace the matrices $\Gamma_i$ by the different Clifford representation and the gap-opening mass term by $\Gamma_B$. The $C_4T$ symmetry implies that the four mass-terms on the four $1$-cells must have alternating signs and hence \eqref{eq:hinge_c4t_ex} will take the value $1$ or $-1$.
\end{proof}

For completion, we summarize the spectral characteristics of Hamiltonian~\eqref{Eq:Ham3} under various underlying atomic configurations in Fig.~\ref{Fig:Model3}. 


\begin{figure}[t]
\center
\includegraphics[width=\textwidth]{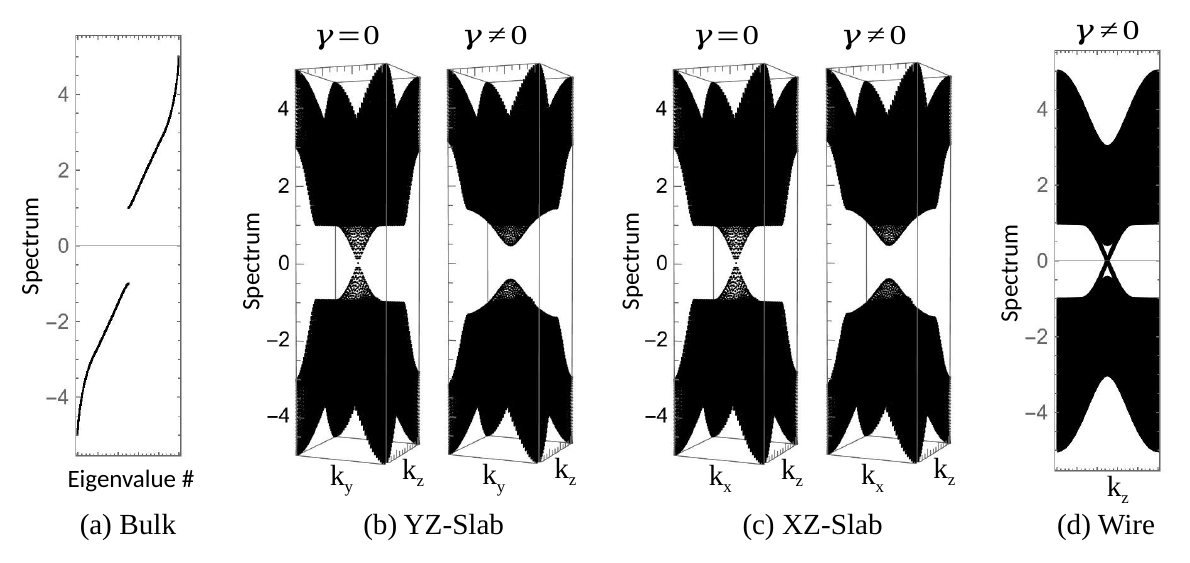}\\
  \caption{\small Same as Fig.~\ref{Fig:Model1} but for the Hamiltonian~\eqref{Eq:Ham3}. 
}
 \label{Fig:Model3}
\end{figure}

\subsection{The cube geometry}\label{Sec:CubeExamples}
Let us now consider the cube groupoid $\Gg_{\mbox{\small \mancube}}$ from Section~\ref{Sec:Crystal}. There is a natural filtration of the unit space by faces, hinges, corners such that the cofiltration \eqref{Eq:CubeFilt}, written schematically as $\Aa_3\to \Aa_2\to \Aa_1 \to \Aa_0 \to 0$, has ideals
\begin{equation}
\begin{aligned}
\Ee_0 &= \Bb \simeq C(\TM^3)\\
\Ee_1 &\simeq  (C(\TM^2)\otimes\KM(\ell^2(\NM)))^{\otimes 6}\\
\Ee_2 &\simeq (C(\TM)\otimes\KM(\ell^2(\NM\times\NM)))^{\otimes 12}\\
\Ee_3 &\simeq (\KM(\ell^2(\NM\times\NM\times\NM)))^{\otimes 8}.
\end{aligned}
\end{equation}
The first-order differentials can again be computed easily in terms of Chern numbers and Toeplitz extensions and one finds that there are no non-trivial higher-order bulk-boundary correspondences in complex K-theory. In fact, the kernel of the first differential $\partial^{1}_q$ is precisely the image of $K_q(\CM)\to K_\ast(\TM^3)$, i.e. only trivial elements remain. Enhancing the K-groups by crystalline symmetries, one can stabilize second-order hinge modes or third-order corner modes. Since the mechanism and formalism should be clear by now we do not need to give an explicit example of a third-order bulk-boundary correspondence.

\begin{figure}
	\centering
\includegraphics[width=\textwidth]{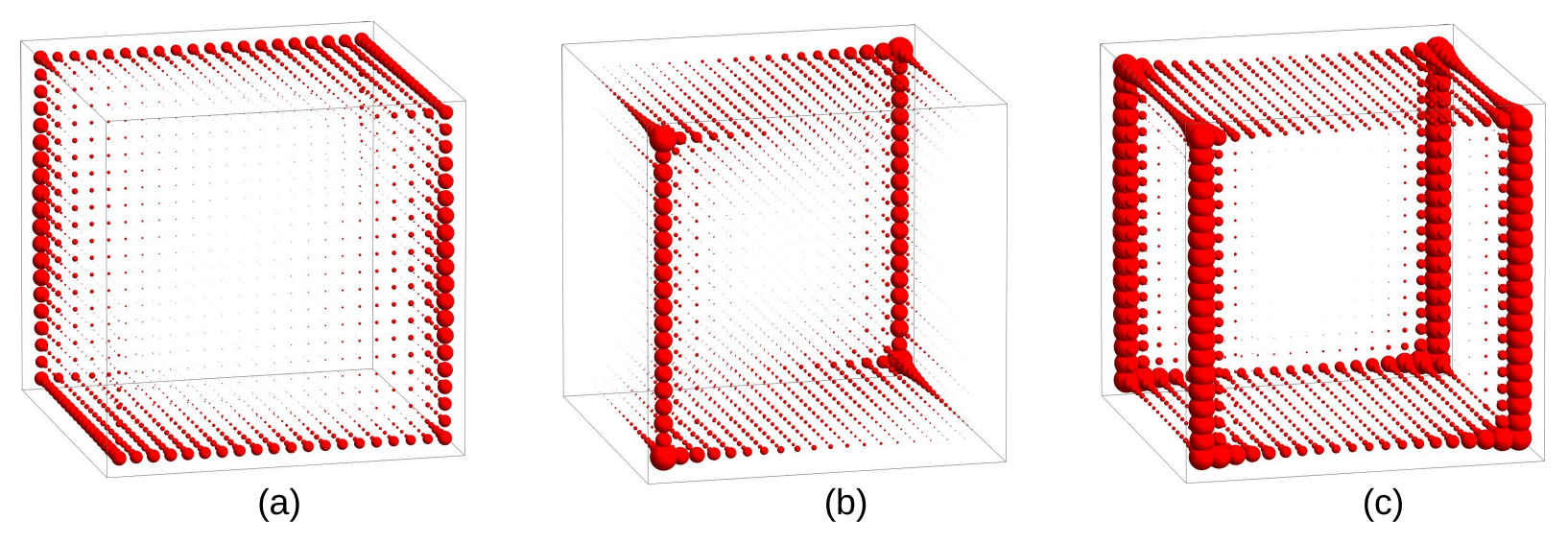}
	\caption{Renderings of the amplitudes $|\psi_0|^2$ corresponding to near zero-energy eigenfunctions  of Hamiltonians: (a)  \eqref{Eq:Ham1}, (b) \eqref{Eq:Ham2} and (c) \eqref{Eq:Ham3}, deployed on a 20$\times$20$\times$20 cube geometry. The amplitudes $|\psi_0(x)|^2$ are encoded in the size of the markers. The hinge modes predicted in the main text for inversion, $C_2T$ and $C_4T$ symmetries are clearly visibile along the verticle edges.}
	\label{Fig:HingeModes}
\end{figure}

Instead, of doing those computations we want to illustrate the point made in Section~\ref{Sec:cube_mechanism} about mixed-order bulk-boundary correspondence. If we put the example Hamiltonians \eqref{Eq:Ham1},\eqref{Eq:Ham2} and \eqref{Eq:Ham3} on a finite cube one finds boundary states supported on different parts of the boundary as seen in Figure~\ref{Fig:HingeModes}. Only in the inversion-symmetric case does one still have purely hinge modes and thus a second-order bulk-boundary correspondence, while in the other two cases one has surface states protected by a first-order bulk-boundary correspondence on the $C_2$- respectively $C_4$-invariant top and bottom surfaces. This is in fact easy to prove: Any possible surface state is characterized by a K-theory class in the image of a boundary map $\partial^p_*$. As a consequence of the long exact sequences of K-theory it is therefore in the kernel of the consecutive differential $\partial^{p+1}_{*-1}$. In the case of the boundary map $\partial_0^2$ the kernel of $\partial_{-1}^3$ consists precisely of those classes whose hinge Chern numbers satisfy the analogue of the Kirchhoff current law at each corner and in the present case there exist no $C_2T$ or $C_4T$-invariant configurations that do so. Thus some of the faces must be un-gapped already.

This is the main reason why we used the smaller filtration given by the infinite wire $\Xi_\square$ instead of maximal choice $\Xi_{\mbox{\small \mancube}}$ in this section. Again, to resolve all mixed-order bulk-boundary correspondences one may have to try different filtrations to find out which of the boundaries in real-space present K-theoretic obstructions at which order for a given bulk K-theory class.

\newpage
\section{Appendix}\label{Sec:Appendix}

\subsection{Twisted group representations}\label{Sec:TwistRep}

As is well-known since Wigner's analysis, the symmetries of quantum systems are given by groups acting on rays in a Hilbert space. They lift to projective representations which act as unitary or anti-unitary operators, sometimes called PUA (projective unitary/anti-unitary) representations. In a more modern form that also includes the possibility of introducing a grading on the groups they have been studied in the remarkable article \cite{FreedAHP2013} and for actions on $C^*$-algebras we refer to \cite{Kubota16,Kubota17,ThiangAHP16,GomiEtAl21}.

\begin{definition}
	A twist $(\phi, c,\tau)$ of a group $\Gamma$ consists of 
	\begin{itemize}
		\item homomorphisms $\phi,c: \Gamma \to \ZM_2$.
		\item A $2$-cocycle $\tau \in H^2_\phi(\Gamma, \TM)$ 	where $\Gamma$ acts on $\TM=\{t \in \CM: \, \abs{t}=1\}$ via $g \cdot t = {}^{\phi(g)}t$ where ${}^{\phi(g)}t=t$ if $\phi(g)=0$ and ${}^{\phi(g)}t=\overline{t}$ if $\phi(g)=1$.
		
		Those cocycles classify the $\phi$-twisted group extensions, i.e. exact sequences
		$$1\to \TM \to \Gamma^\phi_{\tau} \to \Gamma \to 1,$$
		where $\Gamma$ acts on the abelian subgroup $\TM$ via $g t g^{-1} = {}^{\phi(g)}t$ for all $g \in \Gamma$. 
	\end{itemize}
\end{definition}
For simplicity we will assume that all groups are finite, as is sufficient for the purposes of this work. Those groups can also act on operator algebras, in particular complex $C^*$-algebras, in a way that incorporates those twists. While those do not necessarily need to possess a (distinguished) Real structure, the actions may be anti-linear:
\begin{definition}
	A $\phi$-twisted $\Gamma$-action on a $C^*$-algebra $\Aa$ is an $\RM$-linear homomorphism $\alpha: \Gamma \to \mathrm{Aut}_\RM(\Aa)$ such that $\alpha_g$ is complex linear for $\phi(g)=0$ and conjugate linear for $\phi(g)=1$.
	
	A $(\phi,\tau)$-twisted $\Gamma$-action on a $C^*$-algebra $\Aa$ is a $\phi$-twisted $\Gamma^\tau$-action such that $\alpha\rvert_{\TM} = \idmap.$
\end{definition}
The grading and twist is furthermore incorporated by representations on graded Hilbert modules:
\begin{definition}
	Let $(\Aa,\alpha)$ be a $\phi$-twisted $\Gamma$-$C^*$-algebra. A $(\phi, c,\tau)$-twisted representation on a graded Hilbert $\Aa$-module $E$ is an $\phi$-linear homomorphism $U: \Gamma^\phi_\tau \to \Ll_\RM(E)$ ($\RM$-linear bounded operators that are not necessarily adjointable) where $U\rvert_\TM$ is realized by scalar multiplication, and one has
	$$\langle U(g) \xi, U(g)\eta\rangle_E = \alpha_g(\langle \xi, \eta\rangle_E)$$
	and
	$$\gamma_E(U(g)\xi)=(-1)^{c(g)} U(g) \gamma_E(\xi)$$
	for $\gamma_E$ the grading of $E$. In particular, $U(g)$ is odd if $c(g)=1$ and even otherwise, $U(g)$ is anti-linear if $\phi(g)=1$ and linear otherwise.
\end{definition}
For representations on graded Hilbert spaces, i.e. the case $\Aa=\CM$, the action $\alpha_g$ can only be trivial or complex conjugation, thus $\phi(g)$ decides if $U(g)$ is unitary or anti-unitary.

\subsection{Twisted equivariant $K$-theory}\label{Ap:TwEq}
The notion of twisted equivariant $K$-theory was developed for spaces in \cite{FreedAHP2013} and for operator algebras in \cite{Kubota16,Kubota17,GomiEtAl21}. Of course, equivariant $K$-theory itself even including twisted actions is significantly older, the difference is that these works incorporate anti-linear actions and also those of graded groups in a very natural way that closely aligns with the needs of solid state physics.

There is a natural notion of $(\phi, c, \tau)$-twisted equivariant $KK$-theory based on Hilbert $\Aa$-$\Bb$-modules of $\phi $-twisted $\Gamma$-$C^*$-algebras which carry $(\phi,c,\tau)$-twisted representations, from which one can define $${}^\phi K^\Gamma_{c,\tau}(\Aa):={}^\phi KK^\Gamma_{c,\tau}(\CM, \Aa)$$ where we think of $\CM$ as a $\phi$-twisted $\Gamma$-algebra on which elements $g\in \Gamma$ act trivially or by complex conjugation depending on $\phi(g)$.

For practical computations the definition as $KK$-groups are fairly inconvenient since they are closer to the Fredholm picture of $K$-theory than the standard picture. Kubota \cite{Kubota16} also defines a more useful van-Daele-like picture:

Let $(\Aa, \alpha)$ be a $\phi$-twisted ungraded $\Gamma$-$C^*$-algebra and assume for now it is unital. For any $(\phi, c, \tau)$-twisted finite-dimensional representation of $\Gamma$ on a finite-dimensional vector space $\Vv$ the algebra $\Aa \otimes B(\Vv)$ carries a $(\phi, \tau)$-twisted $\Gamma$-action which we denote by the same letter $\alpha$.

An element $a \in \Aa \otimes B(\Vv)$ is called $c$-twisted invariant if $\alpha_g(a)=(-1)^{c(g)}a$ for all $g\in \Gamma$ and the space of $c$-twisted invariant self-adjoint unitaries is denoted ${}^\phi\Ff^\Gamma_{c,\tau,\Vv}(\Aa)$. 

\begin{enumerate}
	\item[(i)]
	One can take the inductive limit over all $(\phi,c,\tau)$-twisted representations
	$${}^\phi K^\Gamma_{0,c,\tau}(\Aa):=\lim_{\Vv}\pi_0( {}^\phi\Ff^\Gamma_{c,\tau,\Vv}(\Aa))$$ with respect to the inclusion
	$${}^\phi\Ff^\Gamma_{c,\tau,\Vv}(\Aa)\hookrightarrow {}^\phi\Ff^\Gamma_{c,\tau,\Vv\oplus \Ww}(\Aa), \quad a \mapsto a \oplus \gamma_\Ww$$
	where $\gamma_\Ww$ is the grading operator on $\Ww$ (recall that we are using $c$-graded group representations) and $\pi_0$ are equivalence classes w.r.t. norm-continuous homotopy. With the direct sum $[a_1]+[a_2]:= [a_1\oplus a_2]$ this becomes an abelian group with the inverse $-[a]= [-\gamma_\Vv a \gamma_{\Vv}]$ where the representative is in $\Aa\otimes B(\Vv^{\mathrm{op}})$ with the opposite grading $\gamma_{\Vv^{\mathrm{op}}}=-\gamma_{\Vv}$ since $a\oplus (-\gamma_\Vv a \gamma_{\Vv})$ is homotopic to $\gamma_\Vv\oplus (-\gamma_\Vv)$ in ${}^\phi\Ff^\Gamma_{c,\tau,\Vv\oplus \Vv^{\mathrm{op}}}(\Aa)$.
	\item[(ii)]
	Denote by ${}^\phi\Uu^\Gamma_{c,\tau,\Vv}(\Aa)$ the space of unitaries $u\in \Aa\otimes B(\Vv)$ such that $\alpha_g(u)=u$ if $\phi(g)+c(g)=0$ and $\alpha_g(u)=u^*$ if $\phi(g)+c(g)=1$. Then set
	$${}^\phi K^\Gamma_{-1,c,\tau}(\Aa):=\lim_{\Vv}\pi_0( {}^\phi\Uu^\Gamma_{c,\tau,\Vv}(\Aa))$$ with respect to the inclusion
	$${}^\phi\Uu^\Gamma_{c,\tau,\Vv}(\Aa)\hookrightarrow {}^\phi\Uu^\Gamma_{c,\tau,\Vv\oplus \Ww}(\Aa), \quad a \mapsto a \oplus \one_\Ww.$$
	With the direct sum $[u_1]+[u_2]:= [u_1\oplus u_2]$ this becomes an abelian group with inverse $-[u]= [u^*]$.
\end{enumerate}
In the non-unital case one defines $${}^\phi K^\Gamma_{i,c,\tau}(\Aa) = \Ker({}^\phi K^\Gamma_{i,c,\tau}(\Aa^+)\to {}^\phi K^\Gamma_{i,c,\tau}(\CM)).$$ Classes in ${}^\phi K^\Gamma_{0,c,\tau}(\Aa)$ are defined by band-flattenings $\sgn(h)$ of $c$-twisted invariant self-adjoint invertibles $h$, while in some sense $K_{-1}$ is the natural range of the boundary map as one may see below.

In general one can define higher $K$-groups by suspension $${}^\phi K^\Gamma_{p-q,c,\tau}(\Aa)={}^\phi K^\Gamma_{0,c,\tau}(S^{p,q} \Aa) \simeq {}^\phi K^\Gamma_{-1,c,\tau}(S^{p+1,q} \Aa)$$
where $S^{p,q}$ is the tensor product of $\Aa$ with the algebra $C_0(\RM^{p+q})$ where the Real structure
$$\overline{f}(x_1,...,x_p,y_1,...,y_q)= \overline{f(x_1,...,x_p,-y_1,...,-y_q)}$$
is used to extend the $\phi$-linear action from $\Aa$ to $S^{p,q} \Aa$ via
$$\alpha_g(f)(x,y) = \begin{cases}
	\alpha_g(f(x,-y)) & \text{if }\phi(g)=1\\
	\alpha_g(f(x,y)) & \text{if }\phi(g)=0	
\end{cases}.$$
These groups only depend on $p-q$ up to isomorphism.

Twisted equivariant $K$-theory is a homology theory and for every equivariant short exact sequence $$0\to \Jj \to \Aa\to \Aa/\Jj \to 0$$ of $\phi$-twisted $\Gamma$-$C^*$-algebras there exists a boundary map which fits into a long exact sequence
$$...\to{}^\phi K^\Gamma_{n,c,\tau}(\Jj) \to {}^\phi K^\Gamma_{n,c,\tau}(\Aa/\Jj)\stackrel{\partial}{\to} {}^\phi K^\Gamma_{n-1,c,\tau}(\Jj) \to {}^\phi K^\Gamma_{n-1,c,\tau}(\Aa)\to ...$$
For a class $[x]_0\in K^\Gamma_{0,c,\tau}(\Aa/\Jj)$ represented by $x\in {}^\phi \Ff_{c,\tau,\Vv}^\Gamma(\Aa/\Jj)$ it is defined by choosing any self-adjoint lift $\tilde{x}\in {}^\phi \Ff_{c,\tau,\Vv}^\Gamma(\Aa)$ and setting $$\partial[x]_0=[-\exp(-\imath \pi \tilde{x})]_{-1}.$$ This class can be seen as a ${}^\phi K^\Gamma_{-1,c,\tau}(\Jj)$-valued index for the self-adjoint lift $\tilde{x}$ in the sense that it is precisely the $K$-theoretic obstruction to its invertibility. 

\begin{remark}
	\label{rem:picturesk_theory}
{\rm The definition of the boundary map given above uses both pictures of Kubota's $K$-theory and for various constructions one needs to pass between ${}^\phi K^\Gamma_{0,c,\tau}(S^{p,q} \cdot) \simeq {}^\phi K^\Gamma_{-1,c,\tau}(S^{p+1,q} \cdot)$ which can be inconvenient. In general one prefers ``unsuspended'' versions in which higher $K$-classes can be represented in terms of projections (equivalently self-adjoint unitaries) or unitary elements with certain symmetries. This is indeed possible at least in special cases:
\begin{enumerate}
	\item[(i)] For trivial $\phi,c,\tau=0$ one has ${}^\phi K^\Gamma_{n,c,\tau}(\Aa)=K_n^\Gamma(\Aa)$, the usual complex equivariant $K$-groups defined by \cite{Phillips}. In particular the $K$-groups are $2$-periodic.
	\item[(ii)] For $\Gamma= \Gamma_0 \times \ZM_2$, $c=0$ and $(\phi,\tau)$ such that $\Gamma_0$ acts by linear and $\ZM_2$ by anti-linear automorphisms one has
	$${}^\phi K^\Gamma_{n,c,\tau}(\Aa)= KR_n^{\Gamma_0}(\Aa),$$
	where the anti-linear $\ZM_2$-action provides the Real structure on $\Aa$ which commutes with the action of $\Gamma_0$. Those $KR$-groups are $8$-periodic.
	\item[(iii)] In case (ii) one further can express the $KR$-groups in terms of twisted equivariant $K$-theory by adjoining $CT$-type symmetries $$KR_n^{\Gamma_0}(\Aa)= {}^{\phi_{n}}K^{\Gamma_0 \times \Gamma_n}_{0,c_n, \tau_n}(\Aa)$$ where either $n \in \{0,...,7\}$ stands for the anti-unitary symmetry classes (enumerated in order as $\mathrm{AI}$, $\mathrm{BDI}$, $\mathrm{D}$, $\mathrm{DIII}$, $\mathrm{AII}$, $\mathrm{CII}$, $\mathrm{C}$, $\mathrm{CI}$), $\Gamma_n$ is one of the abelian groups $0$, $\ZM_2$ or $\ZM_2\times \ZM_2$ and the grading and twist $(\phi_n, c_n,\tau_n)$ are used to implement the usual commuting or anti-commuting symmetries, one can write the complex $\mathrm{AIII}$-class  as
	$$K_1^\Gamma(A)={}^\phi K^{\ZM_2}_{0,c,\tau}(A)$$
	by adjoining a single oddly graded generator of a trivial $\ZM_2$-action and trivial twists $\phi$,$\tau$. In this way all symmetry classes of the tenfold way fit naturally into Real or complex $K$-theory.
\end{enumerate}
}$\Diamond$
\end{remark}

One helpful isomorphism that we need in the main text is that one can sometimes reduce equivariant to non-equivariant $K$-theory:
\begin{proposition}
	\label{prop:reduction_equivariant}
	For $H$ a subgroup of the finite group $\Gamma$ let $C(\Gamma/H)=C(\Gamma/H,\CM)$ be the Real $C^*$-algebra of functions on $\Gamma/H$ with pointwise complex conjugation we define the $\phi$-twisted left translation
	$$\alpha_{\tilde{g}}(f)(g H) = \begin{cases}
		\overline{f(\tilde{g}^{-1}gH)} & \text{if } \phi(\tilde{g})=1\\f(\tilde{g}^{-1}gH) & \text{if }\phi(\tilde{g})=0
	\end{cases}.$$
	For any $\phi$-twisted trivially graded $\Gamma$-$C^*$-algebra $\Aa$, $\Aa \otimes C(\Gamma/H)$ is a $\phi$-twisted $\Gamma$-$C^*$-algebra and one has
	$${}^\phi K^\Gamma_{n,c,\tau}(\Aa\otimes C(\Gamma/H)) = {}^{\phi\rvert_H}K^H_{n,c\rvert_H,\tau\rvert_H}(\Aa).$$
\end{proposition}
\begin{proof}
It is enough to discuss the case $n=0$ since one can just suspend $\Aa$. Let $\Vv$ be a $(\phi,c,\tau)$-twisted representation of $\Gamma$. Denote the action on $\Aa\otimes B(\Vv)$ by $\alpha$ and that on $(\Aa\otimes C(\Gamma/H))\otimes B(\Vv)$ by $\beta$.

The $c$-twisted $\beta$-invariant functions $f\in  C(\Gamma/H, \Aa\otimes B(\Vv))= (\Aa\otimes C(\Gamma/H))\otimes B(\Vv)$ are determined by their value at the identity coset $f(e H)$ via
\begin{equation}
	\label{eq:equivariance_tech}
	(\beta_g(f))(g H) = (-1)^{c(g)}\alpha_g(f(e H))
\end{equation}
since $\Gamma$ acts transitively on $\Gamma/H$.

Classes in ${}^\phi K^\Gamma_{0,c,\tau}(\Aa\otimes C(\Gamma/H))$ are determined by couples of $(\phi,c,\tau)$-twisted $\Gamma$-representations and $c$-twisted invariant elements $f \in {}^\phi\Ff^\Gamma_{c,\tau,\Vv}(\Aa\otimes C(\Gamma/H))$. Any finite-dimensional $(\phi,c,\tau)$-twisted $\Gamma$-representations restricts to a representation of $H$ and the map $[f]_0 \mapsto [f(eH)]_0$ results in a well-defined homomorphism ${}^\phi K^\Gamma_{0,c,\tau}(\Aa\otimes C(\Gamma/H)) \to K^H_{0,c,\tau}(\Aa)$. It depends only on the class of $f$ since evaluating a $\Gamma$-equivariant homotopy in ${}^{\phi}\Ff^\Gamma_{c,\tau,\Vv}(\Aa \otimes C(\Gamma/H))$ on $eH$ gives an $H$-equivariant homotopy in ${}^{\phi}\Ff^H_{c,\tau,\Vv}(\Aa)$. The element $f(eH)$ is $c$-twisted $H$-invariant  since $H$ acts trivially on $C(\Gamma/H)$.

To prove that this map is an isomorphism it is enough to construct an inverse: W.l.o.g. we can assume that any $(\phi\rvert_{H},c\rvert_{H},\tau\rvert_{H})$-twisted representation of $H$ is the restriction of a $(\phi,c,\tau)$-twisted representation of $\Gamma$ on some graded vector space $\Vv$, since one can always induce a $\Gamma$-representation by enlarging $\Vv$. One then has for each $a\in {}^{\phi\rvert_{H}}\Ff^H_{c\rvert_{H},\tau\rvert_{H},\Vv}(\Aa)$ a unique function $f$ with \eqref{eq:equivariance_tech} and $f(eH)=a$. 

The homomorphism well-defined and injective: Any path in ${}^{\phi\rvert_{H}}\Ff^H_{c\rvert_{H},\tau\rvert_{H},\Vv}(\Aa)$ lifts uniquely to one in ${}^{\phi}\Ff^\Gamma_{c,\tau,\Vv}(\Aa\otimes C(\Gamma/H))$. If the lift $f$ represents the neutral element of ${}^\phi K^\Gamma_{0,c,\tau}(\Aa\otimes C(\Gamma/H))$ then there exists (possibly after stabilization) a path connecting $f$ to $\gamma_\Vv$ and evaluating that path at the identity coset gives a path in  ${}^{\phi\rvert_{H}}\Ff^H_{c\rvert_{H},\tau\rvert_{H},\Vv\oplus \Ww}(\Aa)$ connecting $f(eH)$ to $\gamma_\Vv$, hence it represents the neutral element of  ${}^{\phi\rvert_H}K^H_{0,c\rvert_H,\tau\rvert_H}(\Aa)$.
\end{proof}

In particular, in the special case where $H=\{e\}$ the equivariant $K$-groups always reduce to the usual complex $K$-groups $K_0$ and $K_1$.

\vspace{.2cm}
\noindent\small{{\bf Acknowledgement}\;\; EP and DP were supported by the U.S. National Science Foundation through the grants DMR-1823800 and CMMI-2131760, and by U.S. Army Research Office through contract W911NF-23-1-0127. TS was supported by the German Research Foundation (DFG) Project-ID 521291358.}
\vspace{.2cm}

\noindent\small{{\bf Conflict of interest}\;\; The authors declare that they have no conflict of interest.}
\vspace{.2cm}

\noindent\small{{\bf Data Availability Statement}
The Mathematica codes used to generate the numerical data are available upon request from https://www.researchgate.net/profile/Emil-Prodan.}


\end{document}